\documentclass{article}

\newcommand{\hi}{\mathcal{H}}

\newcommand{\his}{\mathcal{H}_{\mathcal{S}}}

\newcommand{\hir}{\mathcal{H}_{\mathcal{R}}}
\newcommand{\hirprim}{\mathcal{H}_{\mathcal{R}'}}

\newcommand{\hik}{\mathcal{K}}

\newcommand{\Poinc}{\mathcal{P}}
\newcommand{\Poincup}{\mathcal{P}_+^\uparrow}

\newcommand{\mink}{\mathbb{M}}
\newcommand{\euc}{\mathbb{E}}
\newcommand{\Lup}{L_+^\uparrow}
\newcommand{\mLup}{\mathcal{\Lup}}
\newcommand{\Lupd}{\overline{\Lup}}

\newcommand{\supp}{\text{supp }}

\newcommand{\Y}{\yen}

\newcommand{\Bor}{\text{Bor}}

\def\framestate{\mathscr{D}(\mathcal{H}_\mathcal{R})}
\def\framestatep{\mathscr{D}(\mathcal{H}_\mathcal{R}')}

\def\state{\mathscr{D}(\mathcal{H})}
\def\systemstate{\mathscr{D}(\mathcal{H}_\mathcal{S})}

\def\operationalstate{\mathfrak{S}_\R}

\newcommand{\bh}{\mathcal{B}(\hi)}

\newcommand{\bhs}{\mathcal{B}(\his)}
\newcommand{\bhsr}{\mathcal{B}(\his \otimes \hir)}
\newcommand{\bhr}{\mathcal{B}(\hir)}
\newcommand{\bhrp}{\mathcal{B}(\hirprim)}

\newcommand{\lh}{\mathcal{L}(\hi)}

\newcommand{\ldhs}{\mathcal{L}(D_S,\his)}

\newcommand{\thi}{\mathcal{T}(\hi)}
\newcommand{\thr}{\mathcal{T}(\hir)}
\newcommand{\ths}{\mathcal{T}(\his)}
\newcommand{\thsr}{\mathcal{T}(\his \otimes \hir)}

\newcommand{\vacstate}{\systemstate^{\Poincup}}

\newcommand{\eh}{\mathcal{E}(\hi)} 
\newcommand{\Eff}{\mathcal{E}}
\newcommand{\ehir}{\mathcal{E}(\hir)}

\newcommand{\R}{\mathcal{R}}
\renewcommand{\S}{\mathcal{S}}

\newcommand{\A}{\mathcal{A}}

\newcommand{\C}{\mathcal{C}}
\newcommand{\D}{\mathcal{D}}

\newcommand{\I}{\mathcal{I}}

\newcommand{\T}{\mathcal{T}}
\newcommand{\Afrak}{\mathfrak{A}}

\newcommand{\W}{\mathcal{W}}
\newcommand{\Wfrak}{\mathfrak{W}}
\newcommand{\Wscr}{\mathscr{W}}
\newcommand{\Wbf}{\mathbf{W}}

\newcommand{\Cn}{\mathbb{C}}

\newcommand{\Nn}{\mathbb{N}}
\newcommand{\id}{\mathbb{1}}

\newcommand{\E}{\mathsf{E}}
\newcommand{\F}{\mathsf{F}}
\newcommand{\Ef}{\mathsf{F}}
\newcommand{\G}{\mathsf{G}}

\newcommand{\h}{\hspace{2pt}}

\newcommand\indep{\protect\mathpalette{\protect\independenT}{\perp}}
\def\independenT#1#2{\mathrel{\rlap{$#1#2$}\mkern2mu{#1#2}}}

\newcommand\indepER{\protect\mathpalette{\protect\independenTER}{\perp}}
\def\independenTER#1#2{\mathrel{\rlap{$#1#2$}\mkern2mu{#1#2}^{\R}}}

\newcommand\indepERsigma{\protect\mathpalette{\protect\independenTERsigma}{\perp}}
\def\independenTERsigma#1#2{\mathrel{\rlap{$#1#2$}\mkern2mu{#1#2}^{\R}_{\sigma}}}

\newcommand\indepsigma{\protect\mathpalette{\protect\independenTsigma}{\perp}}
\def\independenTsigma#1#2{\mathrel{\rlap{$#1#2$}\mkern2mu{#1#2}_{\sigma}}}

\usepackage[page,title]{appendix}

\usepackage{tocloft}
\usepackage{multicol}
\usepackage{physics}
\usepackage{soul}
\usepackage{relsize}
\usepackage[T1]{fontenc}
\usepackage{lmodern}
\usepackage{slantsc}
\usepackage{bbm}
\usepackage{graphicx}
\usepackage{amsmath}
\usepackage{bbold}
\usepackage{amsfonts}
\usepackage{amssymb}
\usepackage{amsthm}
\usepackage{amsbsy}
\usepackage{mathrsfs}
\usepackage{varioref}
\usepackage{dsfont}
\usepackage{bm}
\usepackage{color}
\usepackage[nopar]{lipsum}
\usepackage{lipsum}
\usepackage{multicol}
\usepackage{tikz-cd}
\usepackage{color}
\usepackage[T1]{fontenc}
\usepackage{authblk}
\usepackage{xcolor}
\usepackage{mathtools}
\usepackage{slashed}
\usepackage{subcaption}
\usepackage{makecell}
\usepackage{tabu}

\usepackage[compat=1.1.0]{tikz-feynman}
\usepackage{tikz-cd}
\usetikzlibrary{arrows}
\usepackage{pgfplots}
\pgfplotsset{compat=1.15}

\newtheorem{theorem}{Theorem}[section] 
\newtheorem{lemma}[theorem]{Lemma} 
\newtheorem{corollary}[theorem]{Corollary}
\newtheorem{proposition}[theorem]{Proposition}
\newtheorem{definition}[theorem]{Definition}
\newtheorem{example}[theorem]{Example}

\newtheorem{notation}[theorem]{Notation}

\newcommand{\U}{\mathcal{U}} 

\newcommand{\V}{\mathcal{V}}

\tikzset{every picture/.style={line width=0.75pt}} 

\usepackage{amsmath}
\usepackage[makeroom]{cancel}

\usepackage{enumitem}

\usepackage[utf8]{inputenc} 
\usepackage{csquotes}
\usepackage{fancyhdr} 
\usepackage[lmargin=0.9in, rmargin=0.9in, tmargin=1in, bmargin=1in]{geometry} 
\usepackage{url} 
\usepackage{listings} 
\usepackage{amsmath,amssymb} 
\usepackage{array, booktabs} 
\usepackage{graphicx} 
\usepackage{float} 
\usepackage[export]{adjustbox} 
\usepackage{color} 

\usepackage[backend = bibtex,
            style = numeric,
            date = long,     
            sorting = none,
            maxcitenames = 3,   
            citestyle=numeric-comp,
            url=false]{biblatex} 
\usepackage{comment} 
\usepackage{afterpage} 

\usepackage{adjustbox}

\usepackage{hyperref}
\usepackage{cleveref}

\usepackage{pifont}

\tikzset{
pattern size/.store in=\mcSize, 
pattern size = 5pt,
pattern thickness/.store in=\mcThickness, 
pattern thickness = 0.3pt,
pattern radius/.store in=\mcRadius, 
pattern radius = 1pt}
\makeatletter
\pgfutil@ifundefined{pgf@pattern@name@_nikqpx0v0}{
\pgfdeclarepatternformonly[\mcThickness,\mcSize]{_nikqpx0v0}
{\pgfqpoint{0pt}{0pt}}
{\pgfpoint{\mcSize+\mcThickness}{\mcSize+\mcThickness}}
{\pgfpoint{\mcSize}{\mcSize}}
{
\pgfsetcolor{\tikz@pattern@color}
\pgfsetlinewidth{\mcThickness}
\pgfpathmoveto{\pgfqpoint{0pt}{0pt}}
\pgfpathlineto{\pgfpoint{\mcSize+\mcThickness}{\mcSize+\mcThickness}}
\pgfusepath{stroke}
}}
\makeatother

\tikzset{
pattern size/.store in=\mcSize, 
pattern size = 5pt,
pattern thickness/.store in=\mcThickness, 
pattern thickness = 0.3pt,
pattern radius/.store in=\mcRadius, 
pattern radius = 1pt}
\makeatletter
\pgfutil@ifundefined{pgf@pattern@name@_ng5lqn19c}{
\pgfdeclarepatternformonly[\mcThickness,\mcSize]{_ng5lqn19c}
{\pgfqpoint{0pt}{-\mcThickness}}
{\pgfpoint{\mcSize}{\mcSize}}
{\pgfpoint{\mcSize}{\mcSize}}
{
\pgfsetcolor{\tikz@pattern@color}
\pgfsetlinewidth{\mcThickness}
\pgfpathmoveto{\pgfqpoint{0pt}{\mcSize}}
\pgfpathlineto{\pgfpoint{\mcSize+\mcThickness}{-\mcThickness}}
\pgfusepath{stroke}
}}
\makeatother
\tikzset{every picture/.style={line width=0.75pt}} 

\tikzset{
pattern size/.store in=\mcSize, 
pattern size = 5pt,
pattern thickness/.store in=\mcThickness, 
pattern thickness = 0.3pt,
pattern radius/.store in=\mcRadius, 
pattern radius = 1pt}
\makeatletter
\pgfutil@ifundefined{pgf@pattern@name@_5n7rcgqhm}{
\makeatletter
\pgfdeclarepatternformonly[\mcRadius,\mcThickness,\mcSize]{_5n7rcgqhm}
{\pgfpoint{-0.5*\mcSize}{-0.5*\mcSize}}
{\pgfpoint{0.5*\mcSize}{0.5*\mcSize}}
{\pgfpoint{\mcSize}{\mcSize}}
{
\pgfsetcolor{\tikz@pattern@color}
\pgfsetlinewidth{\mcThickness}
\pgfpathcircle\pgfpointorigin{\mcRadius}
\pgfusepath{stroke}
}}
\makeatother

\tikzset{
pattern size/.store in=\mcSize, 
pattern size = 5pt,
pattern thickness/.store in=\mcThickness, 
pattern thickness = 0.3pt,
pattern radius/.store in=\mcRadius, 
pattern radius = 1pt}
\makeatletter
\pgfutil@ifundefined{pgf@pattern@name@_7o1tsuk3p}{
\makeatletter
\pgfdeclarepatternformonly[\mcRadius,\mcThickness,\mcSize]{_7o1tsuk3p}
{\pgfpoint{-0.5*\mcSize}{-0.5*\mcSize}}
{\pgfpoint{0.5*\mcSize}{0.5*\mcSize}}
{\pgfpoint{\mcSize}{\mcSize}}
{
\pgfsetcolor{\tikz@pattern@color}
\pgfsetlinewidth{\mcThickness}
\pgfpathcircle\pgfpointorigin{\mcRadius}
\pgfusepath{stroke}
}}
\makeatother

\tikzset{
pattern size/.store in=\mcSize, 
pattern size = 5pt,
pattern thickness/.store in=\mcThickness, 
pattern thickness = 0.3pt,
pattern radius/.store in=\mcRadius, 
pattern radius = 1pt}
\makeatletter
\pgfutil@ifundefined{pgf@pattern@name@_zubggx6qk}{
\makeatletter
\pgfdeclarepatternformonly[\mcRadius,\mcThickness,\mcSize]{_zubggx6qk}
{\pgfpoint{-0.5*\mcSize}{-0.5*\mcSize}}
{\pgfpoint{0.5*\mcSize}{0.5*\mcSize}}
{\pgfpoint{\mcSize}{\mcSize}}
{
\pgfsetcolor{\tikz@pattern@color}
\pgfsetlinewidth{\mcThickness}
\pgfpathcircle\pgfpointorigin{\mcRadius}
\pgfusepath{stroke}
}}
\makeatother

\tikzset{
pattern size/.store in=\mcSize, 
pattern size = 5pt,
pattern thickness/.store in=\mcThickness, 
pattern thickness = 0.3pt,
pattern radius/.store in=\mcRadius, 
pattern radius = 1pt}
\makeatletter
\pgfutil@ifundefined{pgf@pattern@name@_mcbkrmesn}{
\makeatletter
\pgfdeclarepatternformonly[\mcRadius,\mcThickness,\mcSize]{_mcbkrmesn}
{\pgfpoint{-0.5*\mcSize}{-0.5*\mcSize}}
{\pgfpoint{0.5*\mcSize}{0.5*\mcSize}}
{\pgfpoint{\mcSize}{\mcSize}}
{
\pgfsetcolor{\tikz@pattern@color}
\pgfsetlinewidth{\mcThickness}
\pgfpathcircle\pgfpointorigin{\mcRadius}
\pgfusepath{stroke}
}}
\makeatother
\tikzset{every picture/.style={line width=0.75pt}} 

\tikzset{
pattern size/.store in=\mcSize, 
pattern size = 5pt,
pattern thickness/.store in=\mcThickness, 
pattern thickness = 0.3pt,
pattern radius/.store in=\mcRadius, 
pattern radius = 1pt}
\makeatletter
\pgfutil@ifundefined{pgf@pattern@name@_72x0qoqws}{
\makeatletter
\pgfdeclarepatternformonly[\mcRadius,\mcThickness,\mcSize]{_72x0qoqws}
{\pgfpoint{-0.5*\mcSize}{-0.5*\mcSize}}
{\pgfpoint{0.5*\mcSize}{0.5*\mcSize}}
{\pgfpoint{\mcSize}{\mcSize}}
{
\pgfsetcolor{\tikz@pattern@color}
\pgfsetlinewidth{\mcThickness}
\pgfpathcircle\pgfpointorigin{\mcRadius}
\pgfusepath{stroke}
}}
\makeatother

\tikzset{
pattern size/.store in=\mcSize, 
pattern size = 5pt,
pattern thickness/.store in=\mcThickness, 
pattern thickness = 0.3pt,
pattern radius/.store in=\mcRadius, 
pattern radius = 1pt}
\makeatletter
\pgfutil@ifundefined{pgf@pattern@name@_f1y3jgxey}{
\makeatletter
\pgfdeclarepatternformonly[\mcRadius,\mcThickness,\mcSize]{_f1y3jgxey}
{\pgfpoint{-0.5*\mcSize}{-0.5*\mcSize}}
{\pgfpoint{0.5*\mcSize}{0.5*\mcSize}}
{\pgfpoint{\mcSize}{\mcSize}}
{
\pgfsetcolor{\tikz@pattern@color}
\pgfsetlinewidth{\mcThickness}
\pgfpathcircle\pgfpointorigin{\mcRadius}
\pgfusepath{stroke}
}}
\makeatother
\tikzset{every picture/.style={line width=0.75pt}} 

\AtEveryBibitem{
    \clearfield{urlyear}
    \clearfield{urlmonth}
}

\author[1]{Samuel Fedida\thanks{sylf2@cam.ac.uk}}
\author[2,3,4]{Jan G\l{}owacki\thanks{jan.glowacki@oeaw.ac.at}}
\affil[1]{\emph{Centre for Quantum Information and Foundations, DAMTP, Centre for Mathematical Sciences, University of Cambridge, Wilberforce Road, Cambridge CB3 0WA, UK}}
\affil[2]{\emph{Basic Research Community for Physics, Leipzig, GERMANY}}
\affil[3]{\emph{Institute for Quantum Optics and Quantum Information, Vienna, AUSTRIA}}
\affil[4]{\emph{Department of Computer Science, University of Oxford, UK}}

\title{\vspace{-40pt}\textbf{Foundations of Relational Quantum Field Theory I: Scalars
}}

\date{\vspace{-30pt}}

\begin{document}
\maketitle
\thispagestyle{empty}

\begin{abstract}
    We develop foundations for a relational approach to quantum field theory (RQFT) based on the operational quantum reference frames (QRFs) framework considered in a relativistic setting. Unlike other efforts in combining QFT with QRFs, we use the latter to provide novel mathematical and conceptual foundations for the former. We focus on scalar fields in Minkowski spacetime and discuss the emergence of relational local (bounded) observables and (pointwise) fields from the consideration of Poincaré-covariant (quantum) frame observables defined over the space of (classical) inertial reference frames. We recover a relational notion of Poincaré covariance, with transformations on the system directly linked to the state preparations of the QRF. We introduce and analyse various causality conditions, and construct an explicit example of a covariant scalar relational quantum field which is causal relative to operationally meaningful preparations of a relativistic QRF. The theory makes direct contact with established foundational approaches to QFT. We demonstrate that the vacuum expectation values derived within our framework reproduce many of the essential properties of Wightman functions and carry out a detailed comparison of the proposed formalism with Wightman QFT with the frame smearing functions describing the QRF's localisation uncertainty playing the role of the Wightmanian test functions. We also show how the properties of algebras generated by relational local observables suitably extend the core axioms of Algebraic QFT. This work is an early step in revisiting the mathematical foundations of QFT from a relational and operational perspective.
\end{abstract}

\bigskip

\renewcommand{\contentsname}{}
\begin{multicols}{2}
\tableofcontents
\end{multicols}

\newpage

\section{Introduction}
\label{sec:Introduction}

In recent decades, a significant conceptual shift has been underway in quantum foundations, with relational and perspectival approaches gaining considerable traction \cite{rovelli_relational_1996,timpson_quantum_2008,gambini_montevideo_2009}. The core idea of this shift is that the properties of a quantum system, including its state, are not intrinsic and absolute. Instead, they are meaningful only in relation to another physical system.

The modern formalism of quantum reference frames (QRFs) (see \cite{loveridge_quantum_2012,giacomini_quantum_2019,vanrietvelde_change_2020,de_la_hamette_perspective-neutral_2021,castro-ruiz_relative_2025,belenchia_quantum_2018,kabel_quantum_2023,kabel_quantum_2025,giacomini_einsteins_2023,de_vuyst_gravitational_2025,hoehn_quantum_2023,loveridge_relative_2019,loveridge_symmetry_2018,carette_operational_2025,glowacki_operational_2023} for a far from exhaustive selection of contributions) provides a rigorous set of tools to investigate this relational nature of quantum theory. Classically, a reference frame is a passive, abstract coordinate system. The QRF formalism, however, is built on the premise that any real-world reference frame---be it a ruler or a clock---must ultimately be a physical object subject to quantum mechanics. Consequently, a QRF is treated as a quantum system in its own right. Transforming physical descriptions between different QRFs is no longer a simple change of coordinates; it requires a transformation of the quantum states themselves, revealing that how a system is described is inextricably linked to the quantum frame serving as~reference.

This paper develops foundations for a relational approach to QFT (RQFT) starting from first principles, namely by investigating the operational approach to QRFs \cite{carette_operational_2025,glowacki_operational_2023} (see also preceding developments \cite{loveridge_quantum_2012,loveridge_symmetry_2018,loveridge_relativity_2017,loveridge_relative_2019,loveridge_relational_2020}, and recent advances \cite{glowacki_relativization_2024,glowacki_quantum_2024,fewster_quantum_2024}) in the context of relativistic symmetries. We begin with only the Poincaré group as the underlying symmetry structure and demonstrate how local observables and the very notion of a quantum field can be understood as \emph{derived} concepts that emerge from the relational formalism. Remarkably, the framework arising from our purely QRF-theoretic considerations exhibits striking structural similarities to the established axiomatic frameworks of Quantum Field Theory (QFT)---namely, Wightman QFT (WQFT) \cite{wightman_theorie_1964,wightman_fields_1965,streater_outline_1975,streater_pct_1989} and Algebraic QFT (AQFT) \cite{haag_algebraic_1964,halvorson_algebraic_2006,fewster_algebraic_2019}. While other recent works have explored relational ideas in the context of specific QFT models (see e.g. \cite{hoehn_matter_2023,fewster_quantum_2024,de_vuyst_gravitational_2025,carrozza_edge_2024}), the work presented here is, to our knowledge, the first to propose a foundational, constructive take on Relational~QFT, following up on the ideas presented in \cite{glowacki_towards_2024}.

In this paper, we focus on developing a relational theory of scalar quantum fields on Minkowski spacetime. We believe the results achieved constitute a promising starting point for a new research program (as outlined in \cite{glowacki_towards_2024}). It is our hope that by re-deriving field-theoretic structures from operational and relational principles, and doing so by alternative means to those traditionally used in axiomatic QFT research, this approach may shed new light on the very foundations of the subject. Such a theory may not only offer new perspectives on long-standing foundational issues within QFT but could also have profound implications for quantum gravity, where a relational description of spacetime and matter is widely believed to be a necessity. To ensure this work is accessible to a broad audience, no previous exposure to any QRF formalism is assumed. \\

\subsection{Organization of the paper}

The paper is structured as follows. In Sec. \ref{sec:relativistic RFs}, we introduce the formalism via a thought experiment, with the notions of relativistic QRFs and relational quantum fields emerging from heuristic considerations. We also place the formalism in a broader context of the operational approach to QRFs. 

In Sec. \ref{sec: covariance} we investigate a relational notion of (scalar) Poincaré covariance, directly relating the system and the frame. The transformation properties of relational local observables follow immediately from the covariance of frame observables. Here we also introduce relational local quantum fields, which can be understood as (spacetime) integral kernels of relational local observables. Their transformation properties turn out to resemble those of \enquote{physicists'} QFT, with a relational twist. By considering a special class of frames -- namely globally-oriented ones -- we minimize the gap between the relational and non-relational transformation rules.

In Sec. \ref{sec:Vacuum state orthogonality} we discuss an important no-go result by Giannitrapani, which in our context leads to the conclusion that quantum frames, in order to be useful, cannot be prepared in vacuum states. In the context of Algebraic QFT, in which the result is stated, we infer that frame observables should generally be supported on the orthogonal complement of the vacuum sector.

In Sec. \ref{sec:Causality}, we explore various implementations of causality which is now generally to be understood relationally with respect to QRFs. We first discuss how Einstein causality can be seen as an operational constraint on both the system and the frames. We then discuss an ontological implementation of the relativistic no-signalling principle, namely microcausality, and how it implies the relational notion of Einstein causality previously discussed. We construct an explicit nontrivial example of a relational quantum field which is (up to finite precision) causal with respect to operationally meaningful preparations of the frame.

In Sec. \ref{sec:vevs}, we define analogues of the n-point vacuum expectation values for our relational quantum fields, and provide many properties that these satisfy---notably relativistic invariance at the level of both the system and the frame, spectral properties of its kernel, Hermiticity, local commutativity, positivity conditions and cluster decomposition. We then discuss time-ordered vacuum expectation values, and examine properties that these satisfy under microcausality.

In Sec. \ref{sec:Wightman comparison}, we examine the axioms of Wightman QFT before comparing and contrasting them with their RQFT analogues. We show that Wightman quantum fields share many similarities with the relational quantum fields developed in this paper, with certain important caveats---namely the issue of (un)boundedness, the existence and covariance properties of the kernels.

In Sec. \ref{sec:AQFT comparison}, we review the main principles of Algebraic QFT, namely isotony, covariance and causality, and show that the local algebras naturally associated to relational local observables satisfy suitably extended versions of such properties, appropriate for the relational and operational setup presented in this work. We also show that, when slightly adjusted, the relational local algebras also satisfy the time-slice axiom, which is a relevant condition for the determinism of time-evolution in spacetime.

We finish with a summary and a brief description of further research perspectives arising from the presented considerations in Sec. \ref{sec:Summary}.

\subsection{Notation}

Before we move to the main part of the paper, let us introduce some (standard) notation. We denote by $\hi$ a (separable) Hilbert space, by $\bh$ the operator algebra of bounded operators on $\hi$, and by $\state$ the set of density operators (states) on $\hi$. The set of effects, i.e., the unit interval in $\bh^{sa}$ is denoted by $\eh$, the set of trace-class operators by $\thi$. We denote by $L = SO(1,d-1)$ the Lorentz group and by $\Lup = SO(1,d-1)^\uparrow$ the proper orthochronous subgroup in $d$ spacetime dimensions. Likewise, we write $\Poinc = T(1,d-1) \rtimes L$ for the Poincaré group and $\Poincup = T(1,d-1) \rtimes \Lup$ for the proper orthochronous subgroup. We denote by $\mink$ $d$-dimensional Minkowski spacetime and work in mostly-minus signature, and by $T(1,3)$ the spacetime translation group (in $4$ dimensions). By $\bh^G$ and $\state^G$, where $G$ is a (typically locally compact second countable Hausdorff) group, we denote the space of invariant operators and states, respectively. Whenever $x,y \in \mink$ are spacelike separated, we write $x \indep y$. Likewise, whenever $\U, \V \subset \mink$ are spacelike separated, we write $\U \indep \V$. The basics of information-theoretic/operational perspective on infinite-dimensional Quantum Theory, some Measure Theory and Distribution Theory (needed for the comparison with Wightman and Algebraic QFT) are included in App. \ref{App: tech prem.}.

\section{Relational quantum fields}

\label{sec:relativistic RFs}

We will introduce our framework for relational quantum field theory via a thought experiment described, although not in such detail, already in \cite{glowacki_towards_2024}. These heuristic considerations support the definition of a relativistic quantum reference frame (Def. \ref{def: relativistic QRF} below) as a principal QRF for the orthochronous Poincar{\'e} group and relational local observables (RLOs) as conditioned relative observables, as defined within the operational approach \cite{carette_operational_2025}.

\subsection{A tale of quantum frames and fields}

\label{sec:tale RFs}

Consider first the set $F$ of (abstract, classical) inertial frames. Elements $X \in F$ are thought of as different choices of viewpoints from which physical systems can be described. Assuming we are interested in mutually inertial frames that can (in principle) be aligned with each other by means of a physically meaningful procedure, for any pair of elements $X,X' \in F$ there exist a \emph{unique} proper orthochronous Poincar{\'e} transformation, denoted $\Poincup(X,X') \in \Poincup$, relating them. This makes $F$ a \emph{torsor} for the Poincar{\'e} group, meaning that we have a (left, free and transitive) action of $\Poincup$ on $F$ fixed by requiring $(a,\Lambda) = \Poincup(X,(a,\Lambda) \cdot X)$. The set $F$ is then non-uniquely homeomorphic to $\Poincup$ itself
\begin{equation}
    F \cong T(1,d-1) \times \Lup \,.
\end{equation}
Indeed, one can construct a whole family of such homeomorphisms by picking an arbitrary frame $X_0 \in F$ and defining
\begin{equation}
    T_{X_0}: F \ni X \mapsto \Poincup(X_0,X)  \in \Poincup,
\end{equation}
so that $T_{X_0}((a,\Lambda)\cdot X_0) = (a,\Lambda)$. Now notice that the group of translations is homeomorphic to the Minkowski spacetime so we get, now uniquely
\begin{equation}
    F = \mink \times \mLup \, ,
\end{equation}
where $\mink$ is the Minkowski space-time, understood as a manifold, and $\mLup$ is a torsor for the Lorentz group $\Lup$, understood as the (global\footnote{In extensions of this kinematical picture to curved spacetimes, the choice of coordinate system should be understood locally, as a tetrad at a point.\cite{glowacki_towards_2024}}) choice of coordinate system; under this identification, we will write $(x,\lambda) \in F$. We then see the the structure of the set of frames is ultimately that of a (trivial) Lorentzian \emph{principal bundle} over the Minkowski space-time.\footnote{Results achieved in this work will be generalized to the context of non-trivial Lorentz frame bundles and spin bundles in the future; see \cite{glowacki_towards_2024} for preliminary description of this research direction.} This is portrayed in Fig. \ref{fig:Visualisation pointwise}. \\

Now consider a quantum system $\S$. According to Quantum Theory, the description of $\S$ should be given in terms of quantum states, so elements of $\systemstate$---the space of density operators on a separable complex Hilbert space.\footnote{See App. \ref{App: tech prem.} for the mathematical basics of infinite-dimensional Quantum Theory. We work with quantum systems modelled by type $I$ von Neumann algebras, postponing generalization to arbitrary such algebras and beyond, into the realm of order unit spaces (see e.g. \cite{kuramochi_compact_2020} for a recent exposition), to future work.} Likewise, according to Special Relativity, such state should be given for any choice of an inertial frame. We write
\begin{equation}
\rho^{(\cdot)}: F \ni X \mapsto \rho^{(X)} \in \systemstate
\end{equation}
for such an assignment. Before the frame has been specified, it then only makes sense to specify the state of $\S$ as the image $\rho^{(F)} \subset \framestate$, so a $\Poincup$-orbit in $\systemstate$. The expectation values of observables given a state $\rho^{(X')}$ with $X'=(a,\Lambda)\cdot X$ ought to be related to those given the state $\rho^{(X)}$ in the original frame through a projective unitary transformation\footnote{Note that unitary representations of the Poincaré group are necessarily infinite-dimensional. Further note that if one wants to work with nonlinear quantum theory (e.g. with objective collapse models) it is possible to define non-unitary actions in this scenario. One may then expect to recover a modified notion of QRFs and, consequently, one may attempt to build a nonlinear RQFT for such models. Since building a QFT for objective collapse models is notoriously tricky, this framework could bring a new light to such efforts.} $U_\S(a,\Lambda)$ on $\his$; we thus assume such representation exists and is ultraweakly continuous.\footnote{The ultraweak topology is natural to consider in operational contexts as it reflects convergence of expectation values, however it might turn out to be useful and operationally meaningful to weaken this requirement in the future. The (pre)dual topology on $\systemstate$ is referred to as \emph{operational} \cite{carette_operational_2025}. See App. \ref{App: tech prem.} for the definitions.} We take this unitary action of the Poincaré group to act on the left on $\bhs$ via
\begin{equation}
    (a,\Lambda) \cdot \phi := U_\S(a,\Lambda)\,\phi\, U^\dagger_\S(a,\Lambda) \text{  for all  } \phi \in \bhs
\end{equation}
with a dual right action on $\systemstate$ written by\footnote{We will occasionally also use the left action on the states given by $(a,\Lambda) \cdot \rho := \rho \cdot (a,\Lambda)^{-1} = U_\S(a,\Lambda)\,\rho \, U_\S(a,\Lambda)^\dagger$.}
\begin{equation}
    \rho \cdot (a,\Lambda) := U^\dagger_\S(a,\Lambda)\,\rho\, U_\S(a,\Lambda) \text{  for all  } \rho \in \systemstate.
\end{equation}
We then have\footnote{Notice that, due to the left action of $\Poincup$ on $F$ and right on $\systemstate$, the frame-to-state assignment map is not equivariant but satisfy
\begin{equation}
    \rho^{((a_1,\Lambda_1)(a_2,\Lambda_2)\cdot X)} = \rho^{((a_1,\Lambda_1) \cdot X)}\cdot (a_2,\Lambda_2)  = \rho^{(X)}\cdot (a_2,\Lambda_2)(a_1,\Lambda_1)
\end{equation}
instead. This transformation reflects the correct order in which the Poincar{\'e} transformations are applied to the system.}
\begin{equation}
    \rho^{((a,\Lambda) \cdot X)} = \rho \cdot (a,\Lambda),
\end{equation}
and the expectation value of an operator $\phi \in \bhs$ in a transformed frame reads
\begin{equation}
	\Tr[\rho^{((a,\Lambda) \cdot X)} \phi]  
    = \Tr[U_\S(a,\Lambda)^\dagger \, \rho^{(X)} U_\S(a,\Lambda) \, \phi]
    = \Tr[\rho^{(X)} \, U_\S(a,\Lambda) \phi U_\S(a,\Lambda)^\dagger]
    = \Tr[\rho^{(X)} \, (a,\Lambda) \cdot \phi] \, .
\end{equation}

Further, consider a classical physical reference system $\R$. Think of it as physical rods and a clock combined to provide a system of coordinates, say on the table of a lab\footnote{More precisely, one should consider this thought experiment without a mention of the table (which could be seen as a reference system of its own), however we include it here for visualisation purposes.}. Such a coordinate system will necessarily be deficient in the sense of unavoidably limited precision of the measurement of distances and time differences it can provide. Moreover, consider a situation when we cannot be sure of whether the frame has been rotated or otherwise transformed (with respect to the system $\S$) or not. For example, imagine you leave the quantum state system $\S$ in a state $\rho^{(X)}$ in your lab, and when you come back the next day to work you learn that your colleague has been in the lab since then and maybe touched the table on which you prepared your rods and clocks. You assign a probability $q \in [0,1]$ that your relative state has been left untouched, and a probability $1-q$ that your colleague has rotated the frame by $(a,\Lambda) \in \Poincup$ so that the frame is in $X' = (a,\Lambda) \cdot X$. In such a situation, the expectation values are given by the (proper\footnote{In unitary quantum theory, proper and improper mixtures are operationally indistinguishable. This need not be the case for nonlinear extensions of quantum theory \cite{fedida_mixture_2025}.}) mixed state that you should assign to $\S$ relative to $\R$, which is given by
\begin{equation} 
	\rho=q \rho^{(X)} + (1-q) \rho^{((a,\Lambda)\cdot X)} \Rightarrow \Tr[\rho\, \phi] = q \Tr[\rho^{(X)}\, \phi] + (1-q) \Tr[U_\S(a,\Lambda)^\dagger \rho^{(X)} U_\S(a,\Lambda) \,\phi] \, .
\end{equation}
One can see this as saying that the frame is being rotated with some probability with a fixed state or, equivalently, that the state is rotated with some probability with a fixed frame when evaluated on that operator $\phi \in \bhs$ - there is an intrinsic notion of \emph{covariance} at the level of this relational description, details of which will be uncovered in due course. Notice also, that if we think of $X$ and $X'$ as related by a physical transformation, the Poincar{\'e} group element relating them will indeed belong to the proper orthochronous subgroup $\Poincup \subset \Poinc$. \\

More generally, if one assigns a probability distribution $\mu \in \text{Prob}(F)$\footnote{The set of frames $F$ inherits a topology from $\Poinc$, and thus becomes a measurable space under the Borel $\sigma$-algebra of subsets.} describing the uncertainty of the frame orientation, the corresponding state will be given by\footnote{This integral can be understood in terms of Bochner theory for Banach space valued functions \cite{bochnera,diestel,bourgin,Mikusinski1978}.}
\begin{equation}
	\rho^{(\mu)} := \int_F \rho^{(X)} \, d\mu(X) \, .
\end{equation}

By choosing an arbitrary frame $X_0 \in F$ we can use the homeomorphism $T_{X_0}$ to parametrize the space of frames by Poincar{\'e} group elements and write

\begin{equation}
    \rho^{(\mu)} =  \int_F \rho^{(X)} \, d\mu(X) =
    \int_{\Poincup} \rho^{((a,\Lambda) \cdot X_0)} \, d\mu_{X_0}(a,\Lambda)
    =\int_{\Poincup} \rho^{(X_0)} \cdot (a,\Lambda) \, d\mu_{X_0}(a,\Lambda)
    \, ,
\end{equation}
where $\mu_{X_0}$ denotes a measure on $\Poincup$ given by $\mu \circ T_{X_0}^{-1}$. Crucially, the result of this integration does \emph{not} depend on the choice of $X_0 \in F$.\footnote{Indeed, taking $X'_0 = (a',\Lambda')\cdot X_0$ gives
\begin{equation}
    \int_F \rho^{(X)} \, d\mu(X) =
    \int_{\Poincup}\rho^{((a,\Lambda)(a',\Lambda') \cdot X_0)} \, d\mu_{X'_0}(a,\Lambda)
    = \int_{\Poincup}\rho^{((a,\Lambda)(a',\Lambda') \cdot X_0)} \, d\mu_{X_0}((a,\Lambda)(a',\Lambda'))
    = \int_{\Poincup}\rho^{((a,\Lambda) \cdot X_0)} \, d\mu_{X_0}(a,\Lambda),
\end{equation}
where we have used the easily verifiable fact that $T_{X'_0} = R_{(a',\Lambda')^{-1}} \circ T_{X_0}$ with $R$ denoting the action of $\Poincup$ on itself on the right, which gives $\mu_{X'_0} = \mu_{X_0} \circ R_{(a',\Lambda')}$, the last step being a simple change of variables.
}
In what follows, we will use this fact to allow ourselves to abuse notation  and write formulas like $\rho \cdot (x,\lambda)$ and $U_\S(x,\lambda)^\dagger \rho U_\S(x,\lambda)$ to denote $\rho^{(x,\lambda)}$, and always integrate directly over $F = \mink \times \mLup$. For any operator $\phi \in \bhs$ we then have 
\begin{equation}\begin{aligned}
    \Tr[\rho^{(\mu)} \phi] &= \Tr[\int_{F} \rho \cdot (x,\lambda) \, \phi \, d\mu(x,\lambda)] \\
    &=\Tr[ \int_{F} U_\S(x,\lambda)^\dagger \rho U_\S(x,\lambda) \, \phi] d\mu(x,\lambda) \\
    &=\int_{F} \Tr[ \rho \, U_\S(x,\lambda) \phi U_\S(x,\lambda)^\dagger] d\mu(x,\lambda) \\
    &= \Tr[\rho \, \int_{F} (x,\lambda) \cdot \phi \,d\mu(x,\lambda)]
    =: \Tr[\rho \, \phi^{(\mu)}] \,,
\end{aligned}
\end{equation}
where we have moved from the \enquote{Schrödinger}-like picture to a \enquote{Heisenberg}-like picture and introduced the notation
\begin{equation}
    \phi^{(\mu)} := 
    \int_{F} (x,\lambda) \cdot \phi \,d\mu(x,\lambda) \, 
\end{equation}
for what can be called a $\mu$-\emph{relative observable}. The integrand leads to a natural definition.

\begin{definition}
    Let $\lambda \in \mLup$. We call
    \begin{equation}
        \hat{\phi}_{\lambda} : \mink \ni x \mapsto (x,\lambda) \cdot \phi \in \bhs
    \end{equation}
    a \emph{Lorentz-oriented (absolute) quantum field}. 
\end{definition}

In the above notation we can write an $\mu$-relative observable as

\begin{equation}
    \phi^{(\mu)} = 
    \int_{F} \hat{\phi}_{\lambda}(x) \,d\mu(x,\lambda) \, .
\end{equation}

Just like the $\mu$-relative observables, relational local observables and fields (see below) will also be constructed from such Lorentz-oriented quantum fields. Under the Poincar{\'e} transformations they transform as

\begin{equation}
        (a,\Lambda) \cdot \hat{\phi}_\lambda(x) = (a,\Lambda) \cdot (x,\lambda) \cdot \phi = (\Lambda x + a,\Lambda \lambda) \cdot\phi = \hat{\phi}_{\Lambda \lambda} (\Lambda x + a) \,,
\end{equation}
which gives the following transformation law of the $\mu$-relative observables
\begin{equation}
    (a,\Lambda) \cdot \phi^{(\mu)} = \int_F\hat{\phi}_{\Lambda \lambda} (\Lambda x + a) \,d\mu(x,\lambda).
\end{equation}

Here comes the crucial step, lifting our consideration to the realm of relational quantum physics: suppose now that the probabilistic uncertainty in the frame orientation stems from the fact that $\R$ is a quantum system. If the physical rods and clock are quantum mechanical systems themselves, such uncertainty is unavoidable. Thus, a \enquote{quantum frame} will be endowed with a Hilbert space $\hir$  and an \emph{observable of orientation} modelled as a positive-operator valued measure (POVM) $\E_\R : \Bor(F) \to \ehir$ on the space of frames---the data necessary and sufficient to capture probabilistic orientation entirely within the context of operational quantum physics \cite{glowacki_operational_2023}.\footnote{See App. \ref{App: POVMs} for a brief introduction to operator-valued measures and basic constructions they are subject to.} We will refer to $\mathcal{H_S}$ as the Hilbert space of the \textit{system} $\mathcal{S}$, and to $\mathcal{H_R}$ as the Hilbert space of the \textit{frame} $\mathcal{R}$. Here is the formal definition, based on the general one as given in \cite{carette_operational_2025}.

\begin{definition}\label{def: relativistic QRF}
    A \emph{relativistic quantum reference frame} (QRF) is a tuple $\R = (U_\R,\E_\R,\hir)$, where
    \begin{itemize}
        
        \item $\hir$ is a separable complex Hilbert space,
        
        \item $U_\R : \Poincup \to U(\hir)$ is an ultraweakly continuous projective unitary representation,\footnote{For a unitary representation of the Poincaré group, ultraweak, strong and weak continuity are all equivalent \cite{dixmier_13_1982}.}
        
        \item $\E_\R : \Bor(F) \to \ehir$ is a POVM, i.e., for any \emph{countable} family $(\Delta_n)_n$ of disjoint sets in $\Bor(F)$ we have
            \begin{equation}
        \E_\R\Bigg(\bigcup_{n=1}^\infty \Delta_n\Bigg) = \sum_{n=1}^{\infty} \E_\R(\Delta_n)
            \end{equation}
        in the ultraweak topology. We say that $\E_\R$ is \emph{normalised} if $\E_\R(F) = \mathbb{1}_{\bhr}$. Moreover, $\E_\R$ is required to be $\Poincup$-\emph{covariant},\footnote{Covariant POVMs are sometimes called \emph{system of covariance}, or \emph{system of imprimitivity} when the POVM is sharp \cite{carette_operational_2025}.} meaning that for all $(a,\Lambda) \in \Poincup$ and all $W \in \Bor(F)$ we have
        \begin{equation}
             (a,\Lambda) \cdot \E_\R(W) = \E_\R((a,\Lambda) \cdot W)\,,
        \end{equation}
        where $(a,\Lambda) \cdot \E_\R(W)$ denotes the left unitary action $U_\R(a,\Lambda)\E_\R(W) U_\R(a,\Lambda)^\dagger$ and \begin{equation}(a,\Lambda) \cdot W = \{(a,\Lambda) \cdot (x,\lambda) \mid (x,\lambda) \in W\} = \{(\Lambda x + a,\Lambda \lambda) \mid (x,\lambda) \in W\}\,.
        \end{equation}
        We call this POVM the \emph{frame observable} and refer to $\R$ as being $\Poincup$-\emph{covariant}.
    \end{itemize}
    An \emph{oriented relativistic quantum reference frame} is a tuple $(\R,\omega)$, where $\omega \in \framestate$.
\end{definition}

A relativistic QRF is simply a quantum reference frame based on the $\Poincup$-set $F$, as defined in \cite{carette_operational_2025}. If the frame observable is normalised, we can now take $\mu$ to be a probability measure arising from the Born rule:
\begin{equation}
	\mu(\W) \equiv \mu^{\E_\R}_\omega(\W) := \Tr[\omega \E_{\R}(\W)] \in [0,1] \,,
\end{equation}
where $\W \in \Bor(F)$ is a (Borel) subset of $F$. The POVM should again really be thought of as capturing the essence of \enquote{quantum rods and clocks}, for which the state determines their orientations and localisations in a probabilistic fashion. Given a relativistic QRF, a quantum state $\omega \in \framestate$ gives a probability distribution to play the role of $\mu$, only now it has quantum mechanical origin. This is depicted in Fig. \ref{fig:Visualisation probability distributions}, and leads to the following crucial definition.

\begin{figure*}[t!]
    \centering
    \begin{subfigure}[t]{0.4\textwidth}
    \begin{tikzpicture}[x=0.75pt,y=0.75pt,yscale=-1,xscale=1]

\draw   (209.56,47.56) .. controls (219.21,39.02) and (439.36,31.69) .. (431.16,47.06) .. controls (422.96,62.43) and (430.68,190.51) .. (430.68,212.71) .. controls (430.68,234.91) and (198.64,236.62) .. (188.99,211) .. controls (179.34,185.39) and (199.91,56.1) .. (209.56,47.56) -- cycle ;
\draw    (458.79,42.79) -- (458.79,221.25) ;
\draw    (188.03,263.09) .. controls (207.5,244.66) and (415.24,281.87) .. (434.53,256.26) ;
\draw  [fill={rgb, 255:red, 0; green, 0; blue, 0 }  ,fill opacity=1 ] (301.22,134.21) .. controls (301.22,132.56) and (301.98,131.22) .. (302.91,131.22) .. controls (303.85,131.22) and (304.6,132.56) .. (304.6,134.21) .. controls (304.6,135.87) and (303.85,137.21) .. (302.91,137.21) .. controls (301.98,137.21) and (301.22,135.87) .. (301.22,134.21) -- cycle ;
\draw   (302.39,110.16) .. controls (302.39,100.97) and (308.84,93.52) .. (316.8,93.52) .. controls (324.75,93.52) and (331.2,100.97) .. (331.2,110.16) .. controls (331.2,119.35) and (324.75,126.8) .. (316.8,126.8) .. controls (308.84,126.8) and (302.39,119.35) .. (302.39,110.16) -- cycle ;
\draw    (328.1,100.33) -- (316.8,110.16) ;
\draw [line width=3.75]    (280.06,122.53) -- (300.01,122.53) ;
\draw [line width=3.75]    (280.06,124.24) -- (280.06,98.63) ;
\draw [line width=3.75]    (280.06,122.53) -- (294.1,107.16) ;

\draw (439.07,245.66) node [anchor=north west][inner sep=0.75pt]    {$\mink$};
\draw (467.46,31.83) node [anchor=north west][inner sep=0.75pt]    {$\mathcal{\Lup}$};
\draw (174.59,62.94) node [anchor=north west][inner sep=0.75pt]    {$F$};

\end{tikzpicture}
    \caption{A classical inertial reference frame, which can be thought of as a clock and rods, sharply localised in spacetime and with a definite Lorentz orientation.}
    \label{fig:Visualisation pointwise}
    \end{subfigure}
    \qquad \qquad
    \begin{subfigure}[t]{0.4\textwidth}
        \begin{tikzpicture}[x=0.75pt,y=0.75pt,yscale=-0.98,xscale=0.8]
\draw   (202.6,45.39) .. controls (213.62,36.76) and (465.17,29.36) .. (455.8,44.89) .. controls (446.43,60.42) and (455.25,189.83) .. (455.25,212.26) .. controls (455.25,234.69) and (190.12,236.42) .. (179.09,210.54) .. controls (168.07,184.65) and (191.57,54.02) .. (202.6,45.39) -- cycle ;
\draw    (491.79,40.58) -- (491.79,220.89) ;
\draw    (177.99,263.16) .. controls (200.24,244.54) and (437.61,282.14) .. (459.66,256.26) ;
\draw   (311.75,74.28) .. controls (311.75,65) and (318.82,57.47) .. (327.55,57.47) .. controls (336.28,57.47) and (343.36,65) .. (343.36,74.28) .. controls (343.36,83.57) and (336.28,91.1) .. (327.55,91.1) .. controls (318.82,91.1) and (311.75,83.57) .. (311.75,74.28) -- cycle ;
\draw    (339.95,64.35) -- (327.55,74.28) ;
\draw [line width=3.75]    (287.36,86.79) -- (309.97,86.79) ;
\draw [line width=3.75]    (287.36,88.51) -- (287.36,62.63) ;
\draw [line width=3.75]    (287.36,86.79) -- (303.27,71.26) ;
\draw  [pattern=_nikqpx0v0,pattern size=6pt,pattern thickness=0.75pt,pattern radius=0pt, pattern color={rgb, 255:red, 0; green, 0; blue, 0}][dash pattern={on 4.5pt off 4.5pt}] (259.02,109.42) .. controls (270.04,100.8) and (391.31,95.62) .. (410.6,111.15) .. controls (429.89,126.68) and (378.08,174.99) .. (328.47,141.34) .. controls (278.86,107.7) and (226.5,209.5) .. (215.47,183.62) .. controls (204.45,157.74) and (248,118.05) .. (259.02,109.42) -- cycle ;
\draw  [fill={rgb, 255:red, 0; green, 0; blue, 0 }  ,fill opacity=1 ] (312.29,105.86) .. controls (312.29,104.24) and (313.12,102.93) .. (314.16,102.93) .. controls (315.19,102.93) and (316.02,104.24) .. (316.02,105.86) .. controls (316.02,107.47) and (315.19,108.78) .. (314.16,108.78) .. controls (313.12,108.78) and (312.29,107.47) .. (312.29,105.86) -- cycle ;
\draw   (389.17,186.29) .. controls (380.65,184.94) and (374.75,176.4) .. (376,167.21) .. controls (377.25,158.02) and (385.17,151.67) .. (393.69,153.02) .. controls (402.21,154.37) and (408.11,162.91) .. (406.86,172.1) .. controls (405.61,181.29) and (397.69,187.64) .. (389.17,186.29) -- cycle ;
\draw    (384.09,155.16) -- (391.43,169.66) ;
\draw [line width=3.75]    (363.37,182.37) -- (374.05,158.66) ;
\draw [line width=3.75]    (364.61,183.49) -- (346.03,166.68) ;
\draw [line width=3.75]    (363.37,182.37) -- (359.73,155.6) ;
\draw  [fill={rgb, 255:red, 0; green, 0; blue, 0 }  ,fill opacity=1 ] (374.48,143.23) .. controls (374.48,141.61) and (375.32,140.31) .. (376.35,140.31) .. controls (377.38,140.31) and (378.22,141.61) .. (378.22,143.23) .. controls (378.22,144.85) and (377.38,146.15) .. (376.35,146.15) .. controls (375.32,146.15) and (374.48,144.85) .. (374.48,143.23) -- cycle ;
\draw   (263.66,167.08) .. controls (271.85,163.02) and (281.72,166.51) .. (285.7,174.86) .. controls (289.68,183.21) and (286.26,193.27) .. (278.07,197.33) .. controls (269.87,201.39) and (260,197.9) .. (256.02,189.55) .. controls (252.04,181.19) and (255.46,171.13) .. (263.66,167.08) -- cycle ;
\draw    (285.28,189.73) -- (270.86,182.2) ;
\draw [line width=3.75]    (234.47,194.05) -- (247.53,211.64) ;
\draw [line width=3.75]    (233.36,195.18) -- (250.07,178.21) ;
\draw [line width=3.75]    (234.47,194.05) -- (253.68,196.24) ;
\draw  [fill={rgb, 255:red, 0; green, 0; blue, 0 }  ,fill opacity=1 ] (238.08,171.6) .. controls (238.08,169.99) and (238.92,168.68) .. (239.95,168.68) .. controls (240.98,168.68) and (241.82,169.99) .. (241.82,171.6) .. controls (241.82,173.22) and (240.98,174.53) .. (239.95,174.53) .. controls (238.92,174.53) and (238.08,173.22) .. (238.08,171.6) -- cycle ;
\draw    (233.11,91.72) -- (267.88,125.02) ;
\draw [shift={(270.04,127.09)}, rotate = 223.77] [fill={rgb, 255:red, 0; green, 0; blue, 0 }  ][line width=0.08]  [draw opacity=0] (8.93,-4.29) -- (0,0) -- (8.93,4.29) -- cycle    ;
\draw  [dash pattern={on 0.84pt off 2.51pt}]  (213.71,172.51) -- (213.71,257.17) ;
\draw  [dash pattern={on 0.84pt off 2.51pt}]  (415.04,117.17) -- (415.71,263.84) ;
\draw  [pattern=_ng5lqn19c,pattern size=6pt,pattern thickness=0.75pt,pattern radius=0pt, pattern color={rgb, 255:red, 0; green, 0; blue, 0}] (258.37,256.64) .. controls (280.37,256.64) and (317.71,259.31) .. (340.37,260.64) .. controls (363.04,261.97) and (399.04,253.31) .. (415.71,263.84) .. controls (432.37,274.37) and (187.04,261.71) .. (213.71,257.17) .. controls (240.37,252.64) and (236.37,256.64) .. (258.37,256.64) -- cycle ;
\draw    (188.86,235.48) -- (250.23,258.26) ;
\draw [shift={(253.04,259.31)}, rotate = 200.37] [fill={rgb, 255:red, 0; green, 0; blue, 0 }  ][line width=0.08]  [draw opacity=0] (8.93,-4.29) -- (0,0) -- (8.93,4.29) -- cycle    ;

\draw (466.12,245.71) node [anchor=north west][inner sep=0.75pt]    {$\mink$};
\draw (498.94,29.68) node [anchor=north west][inner sep=0.75pt]    {$\mathcal{\Lup}$};
\draw (163.53,61.09) node [anchor=north west][inner sep=0.75pt]    {$F$};
\draw (199.28,70.3) node [anchor=north west][inner sep=0.75pt]    {$\supp \mu_{\omega}^{\E_\R}$};
\draw (123.95,223.2) node [anchor=north west][inner sep=0.75pt]    {$\supp \mu_{\omega}^{\F_\R}$};

\end{tikzpicture}

    \caption{Born probability distribution over possible clocks and rods associated to an oriented quantum reference frame. The \enquote{fuzziness} in the space of inertial reference frames projects to the level of the localisation in spacetime.}
    \label{fig:Visualisation probability distributions}
    \end{subfigure}
    \caption{Space of inertial frames for which every point represent a different viewpoint from which physical systems can be described. A classical inertial reference frame can be thought of as a Dirac delta distribution over this space, while an oriented relativistic quantum reference frame gives a Born probability distribution over it. The measure projected to $\mink$ is a Born probability measure of a marginal POVM $\F_\R$ as defined in Sec. \ref{sec: covariance} below.}
    \label{fig:Visualisation bundle}
\end{figure*}

\begin{definition}
    Given a relativistic QRF $\R$ and a quantum system $\S$, a \emph{(scalar) relational quantum field} is the map
    \begin{equation}
       \hat{\Phi}^\R: \framestate \ni \omega \mapsto \int_{F} \hat{\phi}_\lambda(x) \, d\mu^{\E_\R}_{\omega}(x,\lambda) \in \bhs,
    \end{equation}
    where $\phi \in \bhs$ is an arbitrary fixed operator. $\Phi^\R(\omega)$
    is referred to as a \emph{(scalar) relational local observable} (RLO).\footnote{This integral can be understood either in terms of Bochner \cite{bochnera,diestel,bourgin,Mikusinski1978}, or as restricted relativization of \cite{loveridge_relativity_2017} (see below).}
\end{definition}

The interplay between relational quantum fields and Lorentz-oriented absolute quantum fields can be further understood in the case where $\R$ is localisable \cite{carette_operational_2025}. Localisable QRFs are those for which $\mu^{\E_\R}_{\omega}$ can arbitrarily well approximate a Dirac delta distribution in the space of frames. In our context, this means that if $\R$ is localisable, then for any $(x,\lambda) \in F$ there exists a sequence of pure states $(\omega_n^{(x,\lambda)}) \subset \framestate$ such that \begin{equation}
    \lim_{n\to \infty} \hat{\Phi}^\R(\omega_n^{(x,\lambda)}) = \hat{\phi}_{\lambda}(x)
\end{equation}
for all $\phi \in \bhs$. That is, when described in localisable QRFs, relational quantum fields can approximate Lorentz-oriented quantum fields. In other words, for localisable QRFs, the Born probability distribution over $F$, shown in Fig. \ref{fig:Visualisation probability distributions}, can be made to \enquote{shrink} to a \emph{classical} reference frame in some appropriate limit, as shown in Fig. \ref{fig:Visualisation pointwise}. Beyond localisable QRFs, individual spacetime points lose their operational meaning, and only \enquote{averages} of pointwise quantities are operationally accessible.

Before we discuss how relational quantum fields provide a novel perspective on relativistic covariance in Quantum Theory in the next section, we place them precisely in the context of the operational approach to QRFs.

\subsection{Relativization and restriction}

The construction of relational local observables can be understood in terms of the \emph{relativisation map} \cite{loveridge_quantum_2012,carette_operational_2025}, which in our context is given by
\begin{equation}\begin{aligned}
    \yen^\R : \bhs &\to \bhsr^{\Poincup} \\
    \yen^\R(\phi) 
    &= \int_{F} \hat{\phi}_\lambda(x) \otimes d\E_\R(x,\lambda),
\end{aligned}\end{equation}
where the superscript $^{\Poincup}$ indicates that the image of the relativization map is invariant under the diagonal action of $\Poincup$ on the composite system. This map is related to the construction of relational local observables by
\begin{equation}
    \Tr[\rho \, \hat{\Phi}^\R(\omega)] = \Tr[(\rho \otimes \omega) \, \yen^\R(\phi)] \qquad \forall \rho \in \systemstate, \, \omega \in \framestate \, ,
\end{equation}
so it can really be seen as an \emph{extension} of the preceding discussion to the case when the states of the system and frame are not treated independently, e.g. allowing for entanglement between the two. In this paper we work under the operationally meaningful assumption that the two relata can be treated independently in the sense of considering only product (or separable) states and diagonal actions on $\bhsr$. However, both these assumptions might be dropped in future work.

The relativization map $\yen^\R$ is a contracting quantum channel \cite{carette_operational_2025}, i.e., it is linear, unital, adjoint-preserving, effect-preserving, bounded (thus continuous), completely positive, normal and contractive. The definition of operator-valued integrals providing the $\yen^\R$ maps was originally introduced in \cite{loveridge_relativity_2017} and further explored in \cite{glowacki_quantum_2023,fewster_quantum_2024,glowacki_quantum_2024}. To see how precisely it relates to our relational quantum fields, let us recall the \emph{restriction map} $\Gamma_\omega$ \cite{loveridge_relativity_2017}. It is a quantum channel satisfying
\begin{equation}\begin{aligned}
    \Gamma_\omega : \bhsr &\to \bhs \\
    \Tr[\rho \Gamma_\omega(\mathcal{O})] &= \Tr[(\rho \otimes \omega) \mathcal{O}]\,,
\end{aligned}\end{equation}
for all $\rho \in \systemstate, \, \omega \in \framestate$ and $\mathcal{O} \in \bhsr$. It can be viewed as a partial trace conditioned on a state $\omega \in \framestate$ in the following sense.

\begin{proposition}
    \label{prop: Gamma is partial trace}
    For all $\omega \in \framestate$ and all $\mathcal{O} \in \thsr$\footnote{This ensures the partial trace of the operator is indeed well-defined since $\mathcal{O} \in \thsr \Rightarrow (\mathbb{1}_{\bhs} \otimes \omega) \mathcal{O} \in \thsr$.},
    \begin{equation}
        \Gamma_\omega(\mathcal{O}) = \Tr_{\hir}[(\mathbb{1}_{\bhs} \otimes \omega) \mathcal{O}] \, .
    \end{equation}
\end{proposition}

\begin{proof}
    See App. \ref{proof 2.5}.
\end{proof}

The restricted relativization is then defined \cite{carette_operational_2025} as composition of the unrestricted one and the restriction map and recovers our relational local observables in the sense that for all $\phi \in \bhs$ we have
\begin{equation}
\yen^\R_\omega (\phi) := \Gamma_\omega \circ \yen^\R = \hat{\Phi}^\R(\omega) \,.
\end{equation}

While $\yen^\R$ is understood as providing observables relative to the quantum reference frame $\R$, invariant under simultaneous symmetry transformations of both the system and the frame, $\yen^\R_\omega$ may be understood as providing a reduced description in terms of $\S$ alone, contingent upon the state of the reference being $\omega$. $\Gamma_\omega$ is thus fundamentally a \enquote{conditional expectation} at the operator level \cite{loveridge_relativity_2017}. For example, $\Gamma_\omega(\sum_{i=1}^N A_i \otimes B_i) = \sum_{i=1}^N\Tr[\omega B_i] A_i$ for $N \in \mathbb{N}$. 

In accordance with the perspective presented in \cite{carette_operational_2025}, the operationally meaningful observables on $\S$ are those arising through relativization. In our case this strips the absolute quantum fields $\hat{\phi}_\lambda(x)$ of operational significance and forces one to work with smeared versions of those, i.e., relational local observables, which is aligned with the traditional motivations for the need of smearing QFT observables \cite{wightman_theorie_1964,wightman_fields_1965,haag_quantum_1955,hall_theorem_1957}---the spacetime points, and thus also the point-wise quantum fields, are not generally operationally meaningful.

\section{Relational covariance}
\label{sec: covariance}

As a simple consequence of invariance of the relativization map, the relational local observables transform covariantly in the following sense.

\begin{theorem}
    \label{thm:full covariance}
    Let $(\R,\omega)$ be an oriented relativistic QRF. Then for all $(a,\Lambda) \in \Poincup$ and $\phi \in \bhs$,
    \begin{equation}
        (a,\Lambda) \cdot \hat{\Phi}^\R(\omega) = \hat{\Phi}(\omega \cdot (a,\Lambda)^{-1}) = \hat{\Phi}((a,\Lambda) \cdot \omega) \, .
    \end{equation}
\end{theorem}

\begin{proof}
    The second equality holds by definition, for the first notice that for all $\rho \in \systemstate$ and $\omega \in \framestate$ we~have
    \begin{align*}
        \Tr[\rho \, (a,\Lambda) \cdot \hat{\Phi}^\R(\omega)] 
        &= \Tr[\rho \cdot (a,\Lambda) \, \hat{\Phi}^\R(\omega)] \\
        &= \Tr[(\rho \cdot (a,\Lambda)) \otimes \omega \, \yen^\R(\phi)] \\
        &= \Tr[(\rho \otimes (\omega\cdot (a,\Lambda)^{-1})) \, (a,\Lambda) \cdot \yen^\R(\phi)] \\
        &= \Tr[\rho \otimes (\omega\cdot (a,\Lambda)^{-1}) \, \yen^\R(\phi)] = \Tr[\rho \, \hat{\Phi}^\R\, (\omega\cdot (a,\Lambda)^{-1})] \,.
    \end{align*}
\end{proof}

We see that this transformation law is very natural: an active transformation at the level of the system shifts the description at the level of the quantum reference frame. 

\subsection{Relational local quantum fields}

We can make direct contact with the picture of point-wisely defined quantum fields smeared around space-time regions by decomposing the integral defining relational local observables into a conditional and marginal measures:

\begin{equation}
    \hat{\Phi}^\R(\omega) = 
    \int_F \hat{\phi}_{\lambda}(x) \, d\mu^{\E_\R}_\omega(x,\lambda) = \int_\mink \left(\int_{\mLup} \hat{\phi}_{\lambda}(x) d\nu^{\E_\R}_{\omega}(\lambda \mid x)\right) d\mu^{\F_\R}_{\omega}(x) 
    \,,
\end{equation}
where the POVM $\F_\R$ is defined as the marginal of $\E_\R$, i.e.,
\begin{equation}
        \F_\R : \Bor(\mink) \ni \U \mapsto
        \E_\R(\U \times \mLup) \in \ehir \,
\end{equation}
and the measure $\nu^{\E_\R}_{\omega}(\cdot\mid x)$ is the \emph{disintegration} \cite{fremlin_45_1998} (or $x$-conditioned measure) of $\mu^{\E_\R}_{\omega}$. The integrand deserves a name.

\begin{definition}
    \label{def: framed scalar quantum field}
    Let $\R$ be a relativistic QRF and $\omega \in \framestate$. The operator-valued function
    \begin{equation}
        \hat{\phi}^\R_{\omega} : \mink \ni x \mapsto \int_{\mLup} \hat{\phi}_{\lambda}(x) \, d\nu^{\E_\R}_{\omega}(\lambda \mid x) \in \bhs
    \end{equation}
    will be called a \emph{(scalar) relational local quantum field}.\footnote{Strictly speaking, the disintegration measure, and therefore the relational local quantum fields as well, are defined only on the support of the original measure, so we have $\hat{\phi}^\R_{\omega} : \pi\left(\supp(\mu^{\E_\R}_\omega)\right) \to \bhs$ where $\pi:F \ni (x,\lambda) \mapsto x \in \mink$. Here, we extend it by zero to a $\mu^{\E_\R}_\omega$-a.e. equivalent measure on the whole of $\mink$ for convenience.}
\end{definition}

We compare and contrast different types of fields in Table \ref{table:comparison between different types of fields} at the end of this section. Let us emphasize here that our relational local quantum fields are defined $\mu^{\E_\R}_\omega$-almost everywhere and have values in bounded operators. By construction, the relational local observables are then spacetime smearings of the relational local quantum fields, i.e., for any $\omega \in \framestate$ we have
\begin{equation}
    \hat{\Phi}^\R(\omega) = 
    \int_\mink \hat{\phi}^\R_\omega(x) \,  d\mu^{\F_\R}_{\omega}(x)\,.
\end{equation}

Notice that the marginal POVM $\F_\R$ is $\Poincup$-covariant. Indeed, for all $(a,\Lambda) \in \Poincup$ and all $\U \in \Bor(\mink)$ we have
    \begin{equation}
        (a,\Lambda) \cdot \F_\R(\U) = (a,\Lambda) \cdot \E_\R(\U,\Lup) = \E_\R((a,\Lambda) \cdot \U, \Lambda \cdot \Lup) = \E_\R((a,\Lambda) \cdot \U, \Lup) = \F_\R((a,\Lambda) \cdot \U) \,.
    \end{equation}

To analyse the transformation properties of the relational local quantum fields we need the following Lemma.

\begin{lemma}
	\label{lem:nu covariance}
    For all $(a,\Lambda) \in \Poincup$,  $\kappa \in \Bor(\mathcal{\Lup})$ and $\mu^{\F_\R}_{\omega\cdot(a,\Lambda)}$-almost every $x \in \mink$ we have
	\begin{equation}
		\nu^{\E_\R}_{\omega \cdot (a,\Lambda)}(\kappa \mid x) = \nu^{\E_\R}_{\omega}(\Lambda \cdot \kappa \mid \Lambda x + a).
	\end{equation}
\end{lemma}

\begin{proof}
    See App. \ref{proof 3.3}.
\end{proof}

Equipped with this result we now easily establish the transformation properties of relational local quantum fields.

\begin{proposition}
	\label{lem:covariance framed scalar field}
	Let $(\R,\omega)$ be an oriented relativistic QRF and $\phi \in \bhs$. Then for all $(a,\Lambda) \in \Poincup$ and $\mu^{\F_\R}_{\omega}$-almost every $x \in \mink$ we have
	\begin{equation}
		(a,\Lambda) \cdot \hat{\phi}^\R_{\omega}(x) = \hat{\phi}^\R_{\omega \cdot (a,\Lambda)^{-1}}(\Lambda x + a) \, .
	\end{equation}
\end{proposition}

\begin{proof}
	We calculate:
	\begin{equation}\begin{aligned}
		(a,\Lambda) \cdot \hat{\phi}^\R_{\omega}(x) &= \int_{\mLup} (\Lambda x + a, \Lambda \lambda) \cdot \phi \, d\nu^{\E_\R}_{\omega}(\lambda \mid x) \\ 
		&= \int_{\mLup} (\Lambda x + a, \lambda) \cdot \phi \, d\nu^{\E_\R}_{\omega}(\Lambda^{-1} \lambda \mid x) \\
        &= \int_{\mLup} (\Lambda x + a, \lambda) \cdot \phi \, d\nu^{\E_\R}_{\omega \cdot (a,\Lambda)^{-1}}(\lambda, \Lambda x + a) \\
		&= \hat{\phi}^\R_{\omega \cdot (a,\Lambda)^{-1}}(\Lambda x + a) \, .
	\end{aligned}\end{equation}
\end{proof}

From this, we then recover a form of scalar Poincaré covariance under the integral in the following form.

\begin{theorem}
    \label{thm: scalar covariance}
    Let $(\R,\omega)$ be an oriented relativistic QRF and $\phi \in \bhs$. Then for all $(a,\Lambda) \in \Poincup$,
    \begin{equation}\label{rlqfcov}
    		(a,\Lambda) \cdot \hat{\Phi}^\R(\omega) = \int_{\mink} \hat{\phi}^\R_{\omega \cdot (a,\Lambda)^{-1}}(\Lambda x + a) d\mu^{\F_\R}_{\omega}(x) \, .
    \end{equation}
\end{theorem}

\begin{proof}
	We have
	\begin{multline}
	(a,\Lambda) \cdot \hat{\Phi}^\R(\omega) 
    = \hat{\Phi}^\R(\omega \cdot (a,\Lambda)^{-1}) 
    =\int_\mink \hat{\phi}^\R_{\omega \cdot (a,\Lambda)^{-1}}(x) d\mu^{\F_\R}_{\omega \cdot (a,\Lambda)^{-1}}(x) 
    =\int_{\mink} \hat{\phi}^\R_{\omega \cdot (a,\Lambda)^{-1}}(x) d\mu^{\F_\R}_{\omega}(\Lambda^{-1}(x-a))\\  = \int_{\mink} \hat{\phi}^\R_{\omega \cdot (a,\Lambda)^{-1}}(\Lambda x + a) d\mu^{\F_\R}_{\omega}(x) \, .
	\end{multline}
\end{proof}

Note that this is not the \enquote{naïve physicists'} transformation law of scalar fields under Poincaré transformations---one could have expected the relational version of covariance of relational local quantum fields to be of the form
\begin{equation}
    		\text{``}\int_\mink (a,\Lambda) \cdot \hat{\phi}^\R_\omega(x) d\mu^{\F_\R}_{\omega}(x) = \int_{\mink} \hat{\phi}^\R_\omega (\Lambda x + a) d\mu^{\F_\R}_{\omega}(x) \text{''} \, .
\end{equation}
However, this could not be right, for at least two reasons. Firstly, we would need to give up the locality of relational local observables---they would need to be supported on the whole of $\mink$. Indeed, if the relational local quantum fields are only non-zero \emph{locally}, on $\supp \mu^{\F_\R}_\omega \subsetneq \mink$, the \enquote{extra} transformation factor by the frame's state is necessary for the formula to make sense in the context of arbitrary Poincar{\'e} transformations. To see why this is the case, notice that\footnote{Supports of measures will always be considered closed sets.}
\begin{equation}
x \notin \supp \mu^{\F_\R}_\omega \Rightarrow \hat{\phi}^\R_\omega(x) = 0 \Rightarrow (a,\Lambda) \cdot \hat{\phi}^\R_{\omega}(x) = 0 \qquad \forall (a,\Lambda) \in \Poincup \,,
\end{equation}
so that $\hat{\phi}^\R_\omega(x) = 0$ for all $x \in \mink$ unless $\supp \mu^{\F_\R}_{\omega} = \mink$ (we shall consider this case further below). On the other hand, the \enquote{extra} transformation factor gives
\begin{equation}
x \in \supp \mu^{\F_\R}_\omega \Rightarrow \Lambda x + a \in \supp \mu^{\F_\R}_{\omega \cdot (a,\Lambda)^{-1}} \qquad \forall (a,\Lambda) \in \Poincup \,,
\end{equation}
by the covariance of $\F_\R$, which transforms the support of the field arbitrarily all over $\mink$, as needed for the covariance formula to be meaningful in the case of relational local quantum fields that are e.g. compactly supported. \\

Secondly, the transformation law \eqref{rlqfcov} renders the relational local quantum fields neither translation- nor Poincar{\'e}-covariant, i.e., generically we have
\begin{equation}
	(a,e) \cdot \hat{\phi}^\R_\omega(x) \neq \hat{\phi}^\R_\omega(x+a) \; , \quad
	(a,\Lambda) \cdot \hat{\phi}^\R_\omega(x) \neq \hat{\phi}^\R_\omega(\Lambda x+a)\,,
\end{equation}
which is a feature and not a bug---the second property allows us to evade a no-go theorem by Wizimirski (see App.~\ref{App:no go Wizimirski}) that would again render our relational local quantum fields trivial. Thus, we endorse the newly discovered transformations for relational local quantum fields stemming from the Poincar{\'e} covariance of relational local observables.
\begin{table}[t!]
\begin{center}
\tabulinesep=1.2mm
\begin{tabu}{||c| c | c ||} 
 \hline
 Objects & Definition & Meaning \\ [0.5ex] 
 
 \hline\hline
 Operator & $\phi \in \bhs$  & \enquote{Absolute} operator on $\his$. \\ 
 \hline
 Lorentz-oriented quantum field & $\begin{aligned}
     \hat{\phi}_\lambda : \mink \ni x \mapsto (x,\lambda) \cdot \phi \in \bhs
 \end{aligned}$ & \enquote{Absolute} pointwise quantum field. \\
  \hline
 Relational local observable & $\begin{aligned}
     \hat{\Phi}^\R(\omega)  = \int_F \hat{\phi}_{\lambda}(x) d\mu^{\E_\R}_{\omega}(x,\lambda) = \yen^\R_\omega(\phi)
 \end{aligned}$ & \makecell{Locally accessible, operationally \\ meaningful operator on $\his$ \\ \emph{relative to an oriented QRF} $(\R,\omega)$.}  \\
 \hline 
  Relational local quantum field & 
  $\begin{aligned}
      \hat{\phi}^\R_\omega(x) = \int_{\mathcal{\Lup}} \hat{\phi}_\lambda(x) d\nu^{\E_\R}_{\omega}(\lambda \mid x)
  \end{aligned}$
  & \makecell{Disintegrated spacetime kernel \\ of a relational local observable.} \\
  \hline
  Relational quantum field & {$\begin{aligned}\hat{\Phi}^\R : \framestate &\to \bhs \\ \omega &\mapsto \int_{\mink} \hat{\phi}^\R_\omega(x) d\mu^{\F_\R}_{\omega}(x)\end{aligned}$} &
  \makecell{Associates relational local observables\\ to frame preparations.\\}
  \\ \hline
\end{tabu}
\caption{Comparison between different objects in Relational Quantum Field Theory.}
\label{table:comparison between different types of fields}
\end{center}
\end{table}

\subsection{Globally oriented frames}

\label{subsec:global orientation}

Inspired by this general formula, we now consider a class of oriented relativistic QRFs which come closer in resembling the transformation properties of kernels in Wightman QFT.\footnote{See Sec. \ref{sec:Wightman comparison} for the discussion of the results achieved here in the context of this established approach.}

\begin{definition}\label{def:globally oriented QRF}
    An oriented relativistic QRF $(\R,\omega)$ is called \emph{globally oriented} if the conditional measure over the Lorentz fiber is independent of position, i.e, for $\mu^{\F_\R}_\omega$-almost every $x,x' \in \mink$ we have
    \begin{equation}
        \nu^{\E_\R}_\omega(\cdot \mid x) = \nu^{\E_\R}_\omega(\cdot \mid x')\,.
    \end{equation}
\end{definition}

A conditional probability measure over a product measurable space which is independent of the variable being conditioned on is just the marginal probability measure of the other variable. That is, for a globally oriented relativistic QRF $(\R,\omega)$ we have
\begin{equation}
    \mu^{\E_\R}_{\omega} = \mu^{\F_\R}_{\omega} \times \mu^{\G_\R}_{\omega}
\end{equation}
where the POVM $\G_\R$ is defined as the Lorentz marginal of $\E_\R$, i.e.,
\begin{equation} \G_\R : \Bor(\mathcal{\Lup}) \ni \kappa \mapsto \E_\R(\mink \times \kappa) \in \ehir \, . \end{equation}

This can be thought of heuristically as having a \enquote{probabilistically smeared constant global section} in the Lorentz frame bundle. In a globally oriented relativistic QRF, a scalar relational local quantum field can be written as
\begin{equation}
    \hat{\phi}^\R_{\omega}(x) = 
    \int_{\mathcal{\Lup}} \hat{\phi}_{\lambda}(x) \, d\mu^{\G_\R}_{\omega}(\lambda) = 
    (x,e) \cdot \int_{\mathcal{\Lup}} (0,\lambda)\cdot \phi \, d\mu^{\G_\R}_{\omega}(\lambda) \, .
\end{equation}
We then see that $\hat{\phi}^\R_{\omega}$ is now translation-covariant. Indeed, we have
\begin{equation}
    (a,e) \cdot \hat{\phi}^\R_{\omega}(x)  
    = (a,e)(x,e) \cdot \int_{\mathcal{\Lup}} (e,\lambda)\cdot \phi \, d\mu^{\G_\R}_{\omega}(\lambda)
    = (x+a,e) \cdot \int_{\mathcal{\Lup}} (e,\lambda)\cdot \phi \, d\mu^{\G_\R}_{\omega}(\lambda)
    = \hat{\phi}^\R_{\omega}(x+a) 
    \,.
\end{equation}

This implies the following result.

\begin{proposition}
    \label{prop:global orientation implies noncompact support}
    Let $(\R,\omega)$ be a globally oriented relativistic QRF. Then $\supp \mu^{\F_\R}_{\omega} = \mink$.
\end{proposition}

\begin{proof}
    First, note that $\hat{\phi}^\R_{\omega}(x) = 0$ if $x \notin \supp \mu^{\F_\R}_{\omega}$. Moreover, $a \cdot \hat{\phi}^\R_{\omega}(x) \neq 0$ iff $\hat{\phi}^\R_{\omega}(x) \neq 0$. Furthermore, in a globally oriented relativistic QRF, $a \cdot \hat{\phi}^\R_{\omega}(x) = \hat{\phi}^\R_{\omega}(x+a)$. But for any $y \in \mink$ and $x \in \supp \mu^{\F_\R}_{\omega}$, $\exists a \in T(1,d-1)$ such that $y = x+a$. So for any non-zero $\hat{\phi}^\R_{\omega}(x)$ (which exists, e.g. take $\phi = \mathbb{1}_{\bhs}$) $\hat{\phi}^\R_{\omega}(x+a) \neq 0$ for all $a \in T(1,d-1)$ so $\supp \mu^{\F_\R}_{\omega} = \mink$. 
\end{proof}

Hence, globally oriented relativistic QRFs cannot have compact spacetime supports. This makes them non-operational yet interesting to study as a reasonable approximate model to make contact with previous results in the field of Wightman QFT, where quantum fields are translation covariant. In particular, they are useful to examine the links with Wightman QFT when the tails at infinity are very small, which will be further discussed in section \ref{sec:Wightman comparison}.

As a final remark, notice here that the POVM $\G_\R$ is Lorentz-covariant and translation-invariant:
\begin{equation}
	G_\R((a,\Lambda) \cdot \kappa) = 
	\E_\R((a,\Lambda) \cdot (\mink \times \kappa)) = \E_\R((0,\Lambda) \cdot (\mink \times \kappa))= 
	(0,\Lambda) \cdot \E_\R(\mink \times \kappa) = (0,\Lambda) \cdot G_\R(\kappa) \, .
\end{equation}

This gives the following general transformation law for globally relational local quantum fields
\begin{equation}
	(a,\Lambda)\cdot \hat{\phi}^\R_\omega(x) = \hat{\phi}^\R_{\Lambda \cdot \omega}(\Lambda x + a).
\end{equation}
They thus transform as if $\omega$ would be translation-invariant, which would be in tension with vacuum-orthogonality of $\F_\R$ as discussed in the next Section.

\section{Vacuum-orthogonality}
\label{sec:Vacuum state orthogonality}

Since our marginal POVMs $\F_\R$ are translation-covariant, it is appropriate to mention here a no-go result by Giannitrapani \cite{giannitrapani_quantum_1998} placing a constraint on such POVMs. We begin by some definitions.

\begin{definition}
	Denote by $\underline{\Bor}(\mink)$ the precompact elements of the Borel $\sigma$-algebra of subsets of $\mink$. We say that a POVM $\F: \Bor(\mink) \to \eh$ is \emph{operationally $\mathcal{O}$-orthogonal}, where $\mathcal{O} \subset \D(\hi)$, if
	\begin{equation}
	\Tr[\Omega \, \F(\V)]=0 \text{ for all } \V \subset \underline{\Bor}(\mink) \text{ and } \Omega \in \mathcal{O}\,.
	\end{equation}
	It is called \emph{operationally trivial} if $\F(\U) = 0$ for all precompact $\U \in \underline{\Bor}(\mink)$. 
\end{definition}

\begin{definition}
	A POVM $\F: \Bor(\mink) \to \eh$ is \emph{vacuum-orthogonal} if it is operationally $\state^{T(1,d-1)}$-orthogonal.
\end{definition}

The term \enquote{operational} here is justified with an observation that real-life measurements can only be carried out within finite regions of spacetime, while \enquote{orthogonality} with the following simple result.

\begin{lemma}
	Let $\F: \Bor(\mink) \to \eh$ be a POVM and $\mathcal{O} \subset \D(\hi)$. Then the following are equivalent:
	\begin{enumerate}
		\item $\F$ is operationally $\mathcal{O}$-orthogonal,
		\item $\hi_\mathcal{O} \subset \underline{\ker}{\F}$, where $\hi_\mathcal{O}:=\left(\bigcup_{\Omega \in \mathcal{O}} {\rm im }\, \Omega \right)^{\rm cl}$ and $\underline{\ker}\,{\F} := \bigcap_{\V \in \underline{\Bor}(\mink)} \ker \,\F(\V)$.
		\item $\underline{{\rm im }}\, \F \subset \hi^\perp_\mathcal{O}$, where $\underline{{\rm im }}\, \F := \bigcup_{\V \in \underline{\Bor}(\mink)} {\rm im }\, \F(\V)$ and $\hi \cong \hi_\mathcal{O} \oplus \hi^\perp_\mathcal{O}$.
	\end{enumerate}
\end{lemma}

\begin{proof}
	1. $\Rightarrow$ 2.: Any state $\Omega \in \mathcal{O}$ can be diagonalized such that $\Omega = \sum_{i=1}^\infty p_i \dyad{\Omega_i}$ with $p_i \geq 0$. The image of $\Omega$ is then the closed subspace spanned by $\Omega_i \in \hi$. Thus if $\F$ is operationally $\{\Omega\}$-orthogonal we have for all $\V \in \underline{\Bor}(\mink)$
	\begin{equation}
		0 = \Tr[\Omega \F(\V)] 
		= \sum_i p_i\braket{\Omega_i}{\F(\V)\Omega_i} 
		\Leftrightarrow \F(\V)\Omega_i=0 \forall_i \Leftrightarrow {\rm im}\, \Omega \subset \ker\F(\V),
	\end{equation}
	where we have used positivity of the effects of $\F$. The claim is an easy extrapolation of this simple observation.
	
	2. $\Rightarrow$ 1.: Conversely, if $\hi_\mathcal{O} \subset \ker{\F}$ holds, for any $\V \in \underline{\Bor}(\mink)$ and $\Omega \in \mathcal{O}$ we have $F(\V)\Omega=0$ so $\Tr[F(\V)\Omega]=0$.\\
	Equivalence $2. \Leftrightarrow 3.$ is elementary.
\end{proof}

Thus, when we restrict our attention to precompact Borel subsets, an operationally $\mathcal{O}$-orthogonal POVM is equivalent to its restriction to the orthogonal complement of $\hi_\mathcal{O}$, i.e.,
\begin{equation}
	\underline{\F}: \Bor(\mink) \ni U \mapsto P_\mathcal{O}\F(U) \in \Eff(\hi_\mathcal{O}^\perp)\,,
\end{equation}
where $P_\mathcal{O}: \hi \to \mathcal{O}^\perp$ is the subspace projection and $\F(U)$ is considered on the domain restricted to $\hi_\mathcal{O}^\perp$.\\

The mentioned result can now be phrased as follows.

\begin{theorem}[\cite{giannitrapani_quantum_1998}]
    \label{thm: no-go giannitrapani}
    Let $\F: \Bor(\mink) \to \eh$ be a translation-covariant POVM on the Minkowski spacetime. Then $\F$ is vacuum-orthogonal. Moreover, if $\bh$ admits a faithful\footnote{A state $\omega$ on a von Neumann algebra $\Afrak \subset \bh$ is said to be \emph{faithful} if for all nonzero positive $A \in \Afrak$, $\omega(A) > 0$ \cite{bratteli_c-algebras_1979}.} 
    translation-invariant state, $\F$ is operationally~trivial.
\end{theorem}

\begin{proof}
     Due to translation covariance of $\F$ we have for any $\U \in \Bor(\mink)$ and any $a \in T(1,d-1)$,
    \begin{equation}
        \Tr[\Omega \, \F(\U)] = \Tr[\Omega \cdot a \, \F(\U)] = \Tr[\Omega \, a \cdot \F(\U)] = \Tr[\Omega \, \F(\U+a)] \, .
    \end{equation}

    Now assume $\V \in \underline{\Bor}(\mink)$ to be precompact and take $\V$ such that $\Tr[\Omega \F(\V)]  = \epsilon > 0$. Then we can find an \emph{infinite} sequence of translations $\{a_n\}_{n \in \mathbb{N}} \in T(1,d-1)$ such that $a_i \cdot\V \cap a_j \cdot \V = \varnothing$ for all $i \neq j$. Then by additivity of $\F$ we get
    \begin{equation}
        \Tr\left[\Omega \F\left(\bigcup_{n=1}^\infty \V.a_{n}\right)\right] = \sum_{n=1}^\infty \Tr[\Omega \F(\V)] = \sum_{n=1}^\infty \epsilon = \infty \, .
    \end{equation}
    But $\F$ is bounded, which gives a contradiction. Thus, $\Tr[\Omega \F(\V)] = 0$ for all precompact $\V \in \underline{\Bor}(\mink)$.
    
    Finally, since the effects of $\F$ are positive, if $\Omega$ is faithful, from $\Tr[\Omega \, \F(\V)]=0$ for all precompact $\V \in \underline{\Bor}(\mink)$ we can conclude that $\F(\V) = 0$ for all such $\V$.
\end{proof}

\begin{corollary}
	The spacetime marginalized frame observables of relativistic QRFs are vacuum-orthogonal.
\end{corollary}

 Thus, our marginal POVMs $\F_\R$ cannot give a non-zero probability of a translation-invariant state to be localized in any pre-compact region. Since such states usually exist for representations of the Poincar{\'e} group\footnote{For finite dimensional Hilbert spaces, the set of states $\systemstate$ is a compact convex subset of a locally convex vector space $\bhs$, so for ultraweakly continuous unitary representations, Day's fixed point theorem \cite{runde_1_2002} ensures the existence of translation-invariant states by the amenability of the translation group. In infinite dimensions, the compactness of the set of states depends on the topology considered on $\bhs$, so this existence may or may not be ensured (in particular, $\systemstate$ is not compact in the ultraweak topology if $\dim \his = \infty$).} and, operationally speaking, measurements can only be carried out within a finite region of spacetime, this is of real concern. However, we do not see this limitation as worrying or surprising---it merely says that we cannot orient the frame using such a state. Physically, we do not expect to be able to prepare a QRF in the vacuum state anyway: meaningful clocks and rods are themselves expected to be some excitations of the vacuum. Thus, a slightly stronger yet arguably still very reasonable condition that relativistic QRFs may satisfy, is the following.

\begin{definition}
	A relativistic QRF $\R=(\U_\R,\E_\R,\hir)$ is \emph{strictly vacuum-orthogonal} if the image of the spacetime marginal POVM $\F_\R: \Bor(\mink) \ni U \mapsto \E_\R(U \times \mLup)$ lies in the orthogonal complement of the vacuum sector, i.e.,
	\begin{equation}
		{\rm im} \,\F_\R \subseteq \hi_{\rm vac}^\perp, \,\text{ where } \,{\rm im} \,\F_\R = \bigcup_{\U \in \Bor(\mink)} {\rm im}\,\F_\R(\U)\, \text{ and } \, \hi_{\rm vac}:=\left(\bigcup_{\Omega \in \framestate^{T(1,d-1)}} {\rm im }\, \Omega \right)^{\rm cl}.
	\end{equation}
\end{definition}

We finish the discussion of vacuum-orthogonality by placing Giannitrapani's result in the context of algebraic QFT \cite{haag_algebraic_1964,halvorson_algebraic_2006,fewster_algebraic_2019}. If one would like the effects of $\F$ to be local in the sense of belonging to the appropriate local algebras of an algebraic QFT satisfying the spectrum condition\footnote{The spectrum condition states that $\sigma(P_\mu) \subset \overline{V_+} = \{p \mid p^0 \geq 0, p^2 \geq 0\}$, i.e. the joint spectrum $\sigma(P_\mu)$ of the energy-momentum lies in the forward causal cone.}  and admitting a pure Poincaré-invariant state---both being very standard requirements, the result above renders the POVM operationally trivial as the Reeh-Schlieder theorem \cite{reeh_bemerkungen_1961} ensures that the vacuum state is cyclic and separating for all precompact regions. In \cite{giannitrapani_quantum_1998} it is shown that, under these conditions, the effects of $\F$ can not even be quasilocal, i.e., they won't belong to the algebra generated by all the local algebras together. This could be worrying for us since given a Poincar{\'e} covariant POVM on $\mink$, one can \emph{generate} a quantum field theory in the algebraic sense by taking local algebras of (Borel) regions $\U \in \Bor(\mink)$ to be
\begin{equation}
	\mathfrak{A}(\U) := \{\F(\V) \mid \V \subset \U\}''.
\end{equation}
	However, even if there is a pure vacuum vector in $\hi$, either considering the operationally equivalent restricted observable $\underline{\F}$ or assuming strict vacuum-orthogonality will lead to an AQFT that does \emph{not} admit a vacuum state. Indeed, since all $\F(\U)$ can be embedded into $B(\hi^\perp_{\rm vac}) \subseteq B(\hi)$ which is a von Neumann algebra itself, the local algebras $\mathfrak{A}(\U)$ need to be subalgebras of $B(\hi^\perp_{\rm vac})$ and thus won't admit pure vacuum states by construction. Thus, one can safely assume that $[\F_\R(\U),\F_\R(\V)]=0$ for causally separated regions in $\mink$ without rendering $\F$ operationally trivial.

\section{Relational causality}
\label{sec:Causality}

A cornerstone of modern physics is the understanding that relativistic no-signalling is paramount to avoid paradoxes related to killing one's beloved grandparents. The relativistic no-signalling principle is fundamentally expressed in a probabilistic fashion: if Alice and Bob can perform experiments across spacelike-separated locations using input parameters $a$ and $b$ yielding outcomes $x$ and $y$, respectively, then the conditional expectation values for each agent should not depend on the other's choice. Mathematically, this is expressed as
\begin{equation}
    p(x \mid a,b) = p(x \mid a) \qquad \& \qquad p(y \mid a,b) = p(y \mid b) \, .
\end{equation}

This principle is however highly non-trivial to precisely implement at the level of individual relativistic theories. For example, a plethora of implementations of this principle have been proposed in relativistic quantum theory (e.g. microcausality, Einstein causality, split property, $C^*$- and $W^*$-independence, the spectrum condition, etc.), with no unique universally accepted set of conditions. In this paper, we do not try to solve this conundrum, but rather discuss different implementations of some well-known quantum causality principles in the framework of relational QFT. 

\subsection{The spacelike resolution of relativistic QRFs}

Importantly, relativistic no-signalling is an \emph{epistemic}, rather than an ontological principle: it asserts what can be determined by observers, not what ``really happens" beyond observations. Motivated by this, a conservative approach to implementing the relativistic no-signalling principle should be taking available resources of the QRFs into account. First, we introduce some notation. 

\begin{definition}
    Let $\sigma \geq 0$. If $\U, \V \subset \mink$, we say that $\U$ and $\V$ are \emph{$\sigma$-spacelike separated} and write $\U \indepsigma \V$ if
    \begin{equation}
        \sup \{(x-y)^2 \mid x \in \U, \, y \in \V\} < -\sigma \, .
    \end{equation} Moreover, let $\R = (\U_\R,\E_\R,\hir)$ be a relativistic QRF and $\omega_1,\omega_2 \in \framestate$. We say that $\omega_1$ and $\omega_2$ are $(\R,\sigma)$-spacelike separated, written $\omega_1 \indepERsigma \omega_2$, if $\supp \mu^{\F_\R}_{\omega_1} \indepsigma \supp \mu^{\F_\R}_{\omega_2}$. We write $\omega_1 \indepER \omega_2$ if $\omega_1 \indepERsigma \omega_2$ for $\sigma=0$. 
\end{definition}

That is, if $\omega_1 \indepERsigma \omega_2$ then $(\R,\omega_1)$ and $(\R,\omega_2)$ are such that the support of their respective spacetime marginal measures are spacelike separated in the sense of this ``thickened" light cone. The rationale for this is that local commutativity holds beyond some distance scale (e.g. the Planck scale), after which it becomes operationally meaningful to discuss superluminal signalling in the first place. We also introduce another operationally meaningful notion: that of compact QRF preparations.

\begin{definition}
    Let $\R$ be a relativistic QRF. A preparation $\omega \in \framestate$ is said to be $\R$-compact if $\supp \mu^{\E_\R}_{\omega}$ is compact. We write
    \begin{equation}
        \operationalstate^c := \{\omega \in \framestate \mid \omega \text{ is } \R\text{-compact}\}
    \end{equation}
    as the (operationally meaningful) subset of all states which yield a compact $\mu^{\E_\R}_{\omega}$.
\end{definition}

Operationally, this is meaningful because the laboratory's spacetime localization is supported in a bounded region, and the preparation does not place the frame in superpositions and mixtures of arbitrarily large rapidities. Note that $\operationalstate^c$ is closed under Poincaré transformations. However, the union of those sets with compact rapidities need not be uniformly bounded in rapidities, and likewise the union of those sets with bounded spacetime support need not be uniformly bounded in spacetime. We thus consider the following further restriction.
\begin{notation}
    Let $\R$ be a relativistic QRF and $K \subseteq F$. We write
    \begin{equation}
        \operationalstate^K := \{\omega \in \framestate \mid \supp \mu^{\E_\R}_{\omega} \subset K \} \, .
    \end{equation}
    If $K$ is compact, then $\operationalstate^K \subset \operationalstate^c$.
\end{notation}
This is the subset of all states $\omega$ such that the oriented QRF $(\R,\omega)$ is supported in $K \subset F$. Both $\operationalstate^K$ and $\operationalstate^c$ are easily seen to be convex: for example, if $\supp \mu^{\E_\R}_{\omega_1} \subset K$ and $\supp \mu^{\E_\R}_{\omega_2} \subset K$ for $\omega_1,\omega_2 \in \operationalstate^K$, then for any $\omega = p \omega_1 + (1-p) \omega_2$, $\supp \mu^{\E_\R}_{\omega} = \supp \mu^{\E_\R}_{\omega_1} \cup \supp \mu^{\E_\R}_{\omega_2} \subset K$. However, unlike $\operationalstate^c$, $\operationalstate^K$ is not generally closed under Poincaré transformations, as we now show.

\begin{proposition}
    \label{prop: SRK covariance}
    Let $\R$ be a relativistic QRF and $K \subseteq F$. Then $\forall (a,\Lambda) \in \Poincup$,
    \begin{equation}
        (a,\Lambda) \cdot \operationalstate^K = \operationalstate^{(a,\Lambda) \cdot K} \, .
    \end{equation}
\end{proposition}

\begin{proof}
    We have
    \begin{align}
        (a,\Lambda) \cdot \operationalstate^K &= \{(a,\Lambda) \cdot \omega \in \framestate \mid \supp \mu^{\E_\R}_{\omega} \subset K \} \\
        &= \{\tilde{\omega} \in \framestate \mid \supp \mu^{\E_\R}_{\tilde{\omega} \cdot (a,\Lambda)} \subset K \} \\
        &= \{\tilde{\omega} \in \framestate \mid (a,\Lambda)^{-1} \cdot \supp \mu^{\E_\R}_{\tilde{\omega}} \subset K \} \\
        &= \{\omega \in \framestate \mid  \supp \mu^{\E_\R}_{\omega} \subset (a,\Lambda) \cdot K \} = \operationalstate^{(a,\Lambda) \cdot K} \, .
    \end{align}
\end{proof}

To remain operational, one should understand local commutativity with finite resources. Indeed, to avoid superluminal signalling, one should ensure there is no violations of Einstein causality \emph{within the resolution of the frame}. This is further motivated by the fact that while Poincaré-covariant POVMs on Minkowski spacetime can be constructed, in their most general form \cite{toller_localization_1999,mazzucchi_observables_2001}, their localizability properties seem to be constrained \cite{beneduci_note_2013}. This supports the claim that there is a generic lack of \enquote{localizability} in the spacelike directions of the marginal POVM $\F_\R$. In effect, the distinguished family of states naturally defines a \enquote{cut-off scale} for spacelike separation as follows.

\begin{definition}
    Let $\R$ be a relativistic QRF and $\operationalstate \subseteq \framestate$ be convex. The quantity
    \begin{equation}
        \sigma(\operationalstate) := \inf_{\omega \in \operationalstate} \sup \{\abs{(x-y)^2} \mid x\indep y \in \supp \mu^{\F_\R}_\omega\}
    \end{equation}
    is the \emph{$\operationalstate$-spacelike resolution of $\R$}.
\end{definition}
Any spacelike separation $(x-y)^2 \geq -\sigma(\operationalstate)$ is, from the \enquote{blurry} point of view of the QRF, potentially causally separated. It is only if the frame can, in principle, perfectly distinguish spacelike-separated points that $\sigma(\operationalstate)=0$. From this viewpoint, it is quite natural to assume that a \enquote{relevant} (in some appropriate sense, e.g. irreducible in the sense of Def.~ \ref{def:irreducible}) $\phi \in \bhs$ must be $(\operationalstate,\sigma(\operationalstate))$-(micro)causal, in the sense we will define below. This idea is depicted pictorially in Fig. \ref{fig:approximate causality}.

\begin{figure}
    \centering
    \begin{subfigure}[t]{0.3\textwidth}
        \begin{tikzpicture}[x=0.75pt,y=0.75pt,yscale=-1,xscale=1]

\draw   (353.6,97.1) .. controls (374.6,89.1) and (383.6,99.1) .. (389.6,104.1) .. controls (395.6,109.1) and (418.93,122.43) .. (423.6,153.1) .. controls (428.27,183.77) and (367.6,221.1) .. (355.67,223.93) .. controls (343.73,226.77) and (314.27,183.77) .. (317.6,153.1) .. controls (320.93,122.43) and (332.6,105.1) .. (353.6,97.1) -- cycle ;
\draw    (260.67,175.67) -- (260.03,137.27) ;
\draw [shift={(260,135.27)}, rotate = 89.05] [color={rgb, 255:red, 0; green, 0; blue, 0 }  ][line width=0.75]    (10.93,-3.29) .. controls (6.95,-1.4) and (3.31,-0.3) .. (0,0) .. controls (3.31,0.3) and (6.95,1.4) .. (10.93,3.29)   ;
\draw    (260.67,175.67) -- (295.33,175.29) ;
\draw [shift={(297.33,175.27)}, rotate = 179.37] [color={rgb, 255:red, 0; green, 0; blue, 0 }  ][line width=0.75]    (10.93,-3.29) .. controls (6.95,-1.4) and (3.31,-0.3) .. (0,0) .. controls (3.31,0.3) and (6.95,1.4) .. (10.93,3.29)   ;
\draw    (320.6,153.1) -- (420.6,153.1) ;
\draw [shift={(423.6,153.1)}, rotate = 180] [fill={rgb, 255:red, 0; green, 0; blue, 0 }  ][line width=0.08]  [draw opacity=0] (8.93,-4.29) -- (0,0) -- (8.93,4.29) -- cycle    ;
\draw [shift={(317.6,153.1)}, rotate = 0] [fill={rgb, 255:red, 0; green, 0; blue, 0 }  ][line width=0.08]  [draw opacity=0] (8.93,-4.29) -- (0,0) -- (8.93,4.29) -- cycle    ;

\draw (256,123.4) node [anchor=north west][inner sep=0.75pt]  [font=\scriptsize]  {$t$};
\draw (300,168.07) node [anchor=north west][inner sep=0.75pt]  [font=\scriptsize]  {$\vec{x}$};
\draw (356.67,158.73) node [anchor=north west][inner sep=0.75pt]  [font=\footnotesize]  {$\sigma ( \operationalstate)$};
\draw (330,111.4) node [anchor=north west][inner sep=0.75pt]    {$\supp \mu _{\omega_\text{min}}^{\F_{\R}}$};
\end{tikzpicture}

    \caption{Spacetime support of the state preparation $\omega_\text{min} \in \operationalstate$ which provides the ``smallest`` spacelike resolution $\sigma(\operationalstate)$ for the frame $\R$ out of all ``operationally meaningful'' states $\omega \in \operationalstate$.}
    \end{subfigure}
    \qquad
    \begin{subfigure}[t]{0.6\textwidth}
        \begin{tikzpicture}[x=0.75pt,y=0.75pt,yscale=-0.85,xscale=0.85]

\draw   (118,116.27) .. controls (138,106.27) and (144,112.93) .. (156.67,118.93) .. controls (169.33,124.93) and (173.33,135.6) .. (180,160.93) .. controls (186.67,186.27) and (153.33,208.93) .. (126,204.93) .. controls (98.67,200.93) and (98,126.27) .. (118,116.27) -- cycle ;
\draw  [dash pattern={on 0.84pt off 2.51pt}]  (172,189.67) -- (245.33,88.93) ;
\draw   (290.67,118) .. controls (310.67,108) and (352,112.4) .. (380.67,118) .. controls (409.33,123.6) and (418,152.93) .. (380.67,178) .. controls (343.33,203.07) and (310.67,208) .. (290.67,178) .. controls (270.67,148) and (270.67,128) .. (290.67,118) -- cycle ;
\draw  [dash pattern={on 0.84pt off 2.51pt}]  (280,126) -- (225.33,208.27) ;
\draw    (189.6,170.62) -- (247,170.92) ;
\draw [shift={(250,170.93)}, rotate = 180.3] [fill={rgb, 255:red, 0; green, 0; blue, 0 }  ][line width=0.08]  [draw opacity=0] (8.93,-4.29) -- (0,0) -- (8.93,4.29) -- cycle    ;
\draw [shift={(186.6,170.6)}, rotate = 0.3] [fill={rgb, 255:red, 0; green, 0; blue, 0 }  ][line width=0.08]  [draw opacity=0] (8.93,-4.29) -- (0,0) -- (8.93,4.29) -- cycle    ;
\draw    (224.6,145.64) -- (263.67,146.22) ;
\draw [shift={(266.67,146.27)}, rotate = 180.85] [fill={rgb, 255:red, 0; green, 0; blue, 0 }  ][line width=0.08]  [draw opacity=0] (8.93,-4.29) -- (0,0) -- (8.93,4.29) -- cycle    ;
\draw [shift={(221.6,145.6)}, rotate = 0.85] [fill={rgb, 255:red, 0; green, 0; blue, 0 }  ][line width=0.08]  [draw opacity=0] (8.93,-4.29) -- (0,0) -- (8.93,4.29) -- cycle    ;
\draw   (486,84.4) .. controls (506,74.4) and (537.33,89.2) .. (554.67,97.87) .. controls (572,106.53) and (603.33,115.87) .. (576,144.4) .. controls (548.67,172.93) and (506,174.4) .. (486,144.4) .. controls (466,114.4) and (466,94.4) .. (486,84.4) -- cycle ;
\draw  [dash pattern={on 0.84pt off 2.51pt}]  (396.67,164.33) -- (470,63.6) ;
\draw  [dash pattern={on 0.84pt off 2.51pt}]  (402.67,193) -- (476,92.27) ;
\draw    (416.6,142.16) -- (436.33,142.57) ;
\draw [shift={(439.33,142.63)}, rotate = 181.19] [fill={rgb, 255:red, 0; green, 0; blue, 0 }  ][line width=0.08]  [draw opacity=0] (8.93,-4.29) -- (0,0) -- (8.93,4.29) -- cycle    ;
\draw [shift={(413.6,142.1)}, rotate = 1.19] [fill={rgb, 255:red, 0; green, 0; blue, 0 }  ][line width=0.08]  [draw opacity=0] (8.93,-4.29) -- (0,0) -- (8.93,4.29) -- cycle    ;
\draw    (458.47,85.98) -- (497.53,86.56) ;
\draw [shift={(500.53,86.6)}, rotate = 180.85] [fill={rgb, 255:red, 0; green, 0; blue, 0 }  ][line width=0.08]  [draw opacity=0] (8.93,-4.29) -- (0,0) -- (8.93,4.29) -- cycle    ;
\draw [shift={(455.47,85.93)}, rotate = 0.85] [fill={rgb, 255:red, 0; green, 0; blue, 0 }  ][line width=0.08]  [draw opacity=0] (8.93,-4.29) -- (0,0) -- (8.93,4.29) -- cycle    ;

\draw (207.67,175.4) node [anchor=north west][inner sep=0.75pt]  [font=\scriptsize]  {$\sigma _{12}$};
\draw (234.33,129.07) node [anchor=north west][inner sep=0.75pt]  [font=\scriptsize]  {$\sigma ( \operationalstate)$};
\draw (109.33,148.73) node [anchor=north west][inner sep=0.75pt]  [font=\small]  {$\supp \mu _{\omega _{1}}^{\F_{\R}}$};
\draw (305.33,136.73) node [anchor=north west][inner sep=0.75pt]  [font=\small]  {$\supp \mu _{\omega _{2}}^{\F_{\R}}$};
\draw (498.67,113.4) node [anchor=north west][inner sep=0.75pt]  [font=\small]  {$\supp \mu _{\omega _{3}}^{\F_{\R}}$};
\draw (415.6,145.5) node [anchor=north west][inner sep=0.75pt]  [font=\scriptsize]  {$\sigma _{23}$};
\draw (467.33,69.07) node [anchor=north west][inner sep=0.75pt]  [font=\scriptsize]  {$\sigma ( \operationalstate)$};

\end{tikzpicture}

    \caption{Examples of $\sigma$-spacelike separations for different state preparations of the frame $\R$ with $\operationalstate$-spacelike resolution $\sigma(\operationalstate)$. We see that $\sigma_{12} > \sigma(\operationalstate)$ i.e. the preparations $\omega_1$ and $\omega_2$ are ``operationally spacelike separated'': the frame can resolve this spacelike separation. On the other hand, $\sigma_{23} < \sigma(\operationalstate)$ i.e. the frame's spacelike resolution is too small to determine whether $\omega_2$ and $\omega_3$ should be understood as yielding causally disconnected supports: relational quantum fields generated by these preparations may not commute.}
    \end{subfigure}
    \caption{Pictorial representations of $\operationalstate$-spacelike resolution of frames and $\sigma(\operationalstate)$-spacelike separation of regions.}
    \label{fig:approximate causality}
\end{figure}

Conversely, an observer may have finite resources available for experiments, and in particular a finite resolution $\sigma_\text{ex}$ for resolving spacelike separations, and may thus want to restrict the set of allowed states of the frame to respect that, i.e. considering $\operationalstate$ be such that $\sigma(\operationalstate) \geq \sigma_\text{ex}$. With these notions in mind, let us examine different implementations of the no-superluminal signalling principle.

\subsection{Einstein causality}

\label{subsec:Einstein causality}

As we highlighted, relativistic no-signalling is an epistemic rather than an ontological principle. An epistemic approach to implementing the relativistic no-signalling principle is via Einstein causality, which has a natural definition in RQFT.

\begin{definition}
    Let $\R=(\U_\R,\E_\R,\hir)$ be a relativistic QRF\footnote{Here, we fix a single QRF $\R$ for the causality conditions, although the definitions could easily be generalised to pairs $\R_1$ and $\R_2$ of relativistic QRFs. Such an approach may lead to some notion of relative locality \cite{amelino-camelia_principle_2011}. It would also be interesting to examine how causality conditions change under QRF transformations. We leave these possibilities to future work.}, $\operationalstate \subseteq \framestate$ be a convex subset of frame preparations\footnote{One can think of this subset as an \enquote{operationally meaningful} subset of all state preparations. For example, we may want to restrict to state preparations $\omega$ for which $\supp \mu^{\E_\R}_{\omega}$ is compact. Convexity is natural to allow for statistical mixtures of operationally meaningful preparations. Ultimately, the choice of a suitable $\operationalstate$ is reliant on the resources available to agents, and a deeper understanding of the interplay between resource theory and QRFs may shed light on the precise properties of such sets.} and $\sigma \geq 0$. Then
    \begin{itemize}
        \item An operator $\phi \in \bhs$ is said to be $(\operationalstate,\sigma)$-\emph{causal}, if for all $\omega_i \in \operationalstate$,
     \begin{equation}
        \label{eqn:phi Einstein causality}
        \omega_1 \indepERsigma \omega_2 \; \Longrightarrow \; \comm{\hat{\Phi}^\R(\omega_1)}{\hat{\Phi}^\R(\omega_2)} =
        \comm{\hat{\Phi}^\R(\omega_1)^\dagger}{\hat{\Phi}^\R(\omega_2)} = 0 \, .
    \end{equation}
        \item A subset\footnote{It may be appropriate to assume $\mathcal{O}_\S$ to be closed under adjoints. In particular, if $\mathcal{O}_\S$ is $(\operationalstate,\sigma)$-causal, then $\tilde{\mathcal{O}}_\S = \mathcal{O}_\S \cup \{\phi^\dagger \mid \phi \in \mathcal{O}_\S\}$ is also $(\operationalstate,\sigma)$-causal.} $\mathcal{O}_\S \subseteq \bhs$ is said to be $(\operationalstate,\sigma)$-\emph{causal}, if for all $\phi_i \in \mathcal{O}_\S$ and all $\omega_i \in \operationalstate$:
    \begin{equation}
        \label{eqn:Einstein causality}
        \omega_1 \indepERsigma \omega_2 \; \Longrightarrow \; \comm{\hat{\Phi}_1^\R(\omega_1)}{\hat{\Phi}_2^\R(\omega_2)} =
        \comm{\hat{\Phi}_1^\R(\omega_1)^\dagger}{\hat{\Phi}_2^\R(\omega_2)} = 0\, .
    \end{equation}
    \end{itemize}
\end{definition}

Requiring the observable $\phi$ to be $(\operationalstate,\sigma)$-causal can be justified along the following lines. The relational local observable $\hat{\Phi}^\R(\omega)$, being the restricted relativized version of the observable $\phi \in \bhs$ in that we have $\hat{\Phi}^\R(\omega)=\Y^\R_\omega (\phi)$, admits the following interpretation---it is the observable $\phi$ accessed via the QRF $\R$ prepared in the state $\omega \in \operationalstate$.\footnote{See \cite{carette_operational_2025,glowacki_operational_2023} for more of the conceptual discussion of the operational QRF framework. The interpretation of local QFT operators as observables is supported by the Fewster-Versch approach to measurement theory in Algebraic QFT \cite{fewster_measurement_2023}, and arguably also by the issues arising from interpreting them as accessible local operations, famously highlighted by Sorkin \cite{Sorkin1993-SORIMO}, and resolved in \cite{fewster_measurement_2023}.} Further, the frame prepared in $\omega$ is understood as being located in $\supp \mu^{\F_\R}_\omega \subseteq \mink$. Thus, requiring \eqref{eqn:phi Einstein causality} is understood as assuring the local observables generated by $\phi$ with respect to the frame $\R$ are commensurable whenever they are spacelike separated up to the resolution provided by $\sigma$; whether a collection of operators mutually satisfy this operational causality principle is determined by the condition \eqref{eqn:Einstein causality}.\\

In constructive quantum field theory, the conditions with $\sigma > 0$ can be shown to be equivalent to the case of $\sigma=0$, highlighting the \enquote{global nature of local commutativity} \cite{streater_pct_1989,wightman_quantum_1960}. This fact depends however on some properties of the kernels of Wightman functions, such as translation invariance that does not generally hold in relational quantum field theory (unless one assumes global orientation). It is thus an open question to determine when these conditions are equivalent to those with $\sigma=0$;  we expect generically they are not.\\

We note that finite-precision relational causality is insensitive to Poincar{\'e} transformations in the following sense.

\begin{lemma}
    \label{lem:causality and covariance}
    Let $\R$ be a relativistic QRF, $\operationalstate \subseteq \framestate$, $\mathcal{O}_\S \subseteq \bhs$, $\sigma \geq 0$ and $(a,\Lambda) \in \Poincup$. Then $\mathcal{O}_\S$ is $(\operationalstate,\sigma)$-causal iff it is $(\operationalstate \cdot (a,\Lambda),\sigma)$-causal.
\end{lemma}

\begin{proof}
    See App. \ref{proof causality and covariance}.
\end{proof}

\begin{corollary}
    \label{cor:causality and covariance}
    Let $\R$ be a relativistic QRF and $\operationalstate \subseteq \framestate$ be convex. Then for all $\phi \in \bhs$, $\sigma \geq 0$ and $(a,\Lambda) \in \Poincup$, $\phi$ is $(\operationalstate,\sigma)$-causal iff it is $(\operationalstate \cdot (a,\Lambda),\sigma)$-causal.
\end{corollary}

In other words, finite-precision $(\operationalstate,\sigma)$-causality is stable under Poincaré transformations.\footnote{Note that this does not mean that $\phi$ is $(\operationalstate,\sigma)$-causal iff it is $(\framestate,\sigma)$-causal, as can easily be seen if $\operationalstate$ is a proper convex subset of $\framestate$ which is closed under Poincaré transformations. This is unlike, for example, nonempty subsets $\W$ of $F$, for which the transitivity of $\Poincup \curvearrowright F$ implies that $$\bigcup_{(a,\Lambda) \in \Poincup} (a,\Lambda) \cdot \W  = F \, .$$}  Note that some examples of $\operationalstate$ can be closed under all $(a,\Lambda) \in \Poincup$. This includes, for example, \enquote{trivial} cases such as $\operationalstate = \framestate$ and $\operationalstate = \varnothing$ as well as more interesting ones such as $\operationalstate = \framestate^{\Poincup}$. Other nontrivial examples include superselection sectors, e.g. if $\hir = \hir^{(0)} \oplus \hir^{(1)} \oplus \cdots$ is some scalar Fock space, where $\hir^{(0)}$ is the vacuum sector, $\hir^{(1)}$ the $1$-particle sector etc; then $\operationalstate=\mathscr{D}(\hir^{(1)}) \subset \framestate$ is indeed convex and closed under unitary actions of Poincaré.
Further note that combining Prop. \ref{prop: SRK covariance} with Lem. \ref{lem:causality and covariance}, it follows that for any $\sigma > 0$ and $(a,\Lambda) \in \Poincup$, a subset $\mathcal{O}_\S \subseteq \bhs$ is $(\operationalstate^K,\sigma)$-causal iff it is $(\operationalstate^{(a,\Lambda)\cdot K)},\sigma)$-causal. Thus, when considering the $\operationalstate^K$-causal operators, the set $K$ should be thought of as being specified only up to Poincar{\'e} transformations.

Regarding the $(\operationalstate,\sigma)$-causality of the whole system $\S$, i.e., taking $\mathcal{O}_\S = \bhs$, consider the following. If observers in spacetime can access everything (and only that) located within the support of the spacetime marginal of their respective oriented quantum reference frames, then $(\operationalstate,\sigma)$-causality prevents faster-than-light communication through non-selective measurements. The usual argument---adjusted to our present relational context---goes as follows: suppose Alice has access to $\S$ through the oriented relativistic QRF $(\R,\omega_1)$ whilst Bob has access to $\S$ through the oriented relativistic QRF $(\R,\omega_2)$, where both Alice and Bob are using the same QRF $\R$ but prepared in two different states such that $\omega_1 \indepERsigma \omega_2$. Then if Bob does nothing, the expectation value of an observable $\phi \in \bhs$ with respect to some state $\rho \in \systemstate$ from Alice's perspective is 
$$\expval{\hat{\Phi}^\R(\omega_1)}_{\rho} = \Tr[\rho \hat{\Phi}^\R(\omega_1)] \, . $$ 

A plausible (though nontrivial and arguably contentious) assumption that one can make is that an observer can only implement operations associated to the relational local observables accessible via their oriented frame.\footnote{Perhaps, as mentioned before, the relational local observables correspond to those operators that observers can \emph{measure} rather than those they can \emph{use} to manipulate the system, in which case this assumption, standard as it is, would fail to hold. Certainly, this discussion depends heavily on which resources observers have access to, and is contingent on one's approach to measurement schemes in relativistic quantum theory. It would be important to figure out how measurements are conducted in RQFT to get a fully satisfactory picture--- this however constitutes a research program on its own. See the Outlook for an outline of how this can be approached.} Then, if Bob instead performs a measurement in $\supp \mu^{\F_\R}_{\omega_2}$ based on some measurement operators which themselves stem from relational local observables $\hat{\Psi}_i^\R(\omega_2)$ where $\psi_i \in \bhs$, $i=1,\cdots,n \in \mathbb{N}$, such that 
\begin{equation}
    \sum_{i=1}^n \hat{\Psi}_i^\R(\omega_2)^\dagger \hat{\Psi}_i^\R(\omega_2) = \mathbb{1}_{\bhs} \, ,
\end{equation}
and assuming that the states then update following non-selective measurements as
$$\rho \mapsto \rho'=\sum_{i=1}^n \hat{\Psi}_i^\R(\omega_2) \rho \hat{\Psi}_i^\R(\omega_2)^\dagger \, ,$$
the expectation value of Alice's relational local observable takes the form
$$\expval{\hat{\Phi}^\R(\omega_1)}_{\rho'} = \sum_{i=1}^n \Tr[\hat{\Psi}_i^\R(\omega_2) \rho \hat{\Psi}_i^\R(\omega_2)^\dagger \hat{\Phi}^\R(\omega_1)] \, .$$

If this is not equal to $\expval{\hat{\Phi}^\R(\omega_1)}_{\rho}$, then Bob could signal to Alice faster-than-light by choosing (or not) to perform a series of measurements on his side. One way to ensure no-superluminal signalling is thus by imposing Einstein causality\footnote{Note that in the formulation of Einstein causality of Eqn. \eqref{eqn:Einstein causality}, we require that fields and their adjoints commute. For scalar quantum fields in RQFT, this is redundant: since $\yen^\R_{\omega}$ is adjoint-preserving, if one restricted relativisation of $\phi \in \bhs$ commutes with \emph{any} spacelike separated restricted relativised operator $\psi \in \bhs$, then it also commutes with the restricted relativisation of $\psi^\dagger$. This may not be the case when one considers more general spinors, should restricted relativisation not be adjoint-preserving.}, so that
\begin{multline}
    \expval{\hat{\Phi}^\R(\omega_1)}_{\rho'} = \sum_{i=1}^n \Tr[\rho \hat{\Psi}_i^\R(\omega_2)^\dagger \hat{\Phi}^\R(\omega_1)\hat{\Psi}_i^\R(\omega_2) ] \stackrel{\eqref{eqn:Einstein causality}}{=} \sum_{i=1}^n \Tr[\rho \hat{\Psi}_i^\R(\omega_2)^\dagger \hat{\Psi}_i^\R(\omega_2) \hat{\Phi}^\R(\omega_1)]  \\ = \Tr\left[\rho \left(\sum_{i=1}^n \hat{\Psi}_i^\R(\omega_2)^\dagger \hat{\Psi}_i^\R(\omega_2) \right) \hat{\Phi}^\R(\omega_1) \right] = \Tr[\rho \hat{\Phi}^\R(\omega_1)] = \expval{\hat{\Phi}^\R(\omega_1)}_{\rho} \, .
\end{multline}

This restriction is operationally meaningful and naturally links with common considerations of quantum causality in other contexts of relativistic quantum theory. Note however that Einstein causality does not impose any restriction on globally oriented relativistic QRFs (Def. \ref{def:globally oriented QRF}), since the spacetime marginal measures of those have the whole of $\mink$ as support. It remains a rather weak epistemic requirement; let us now examine a stronger, yet also common, assumption which also relates to the no-superluminal signalling principle.

\subsection{Microcausality}

A possible \emph{ontological} implementation of the no-superluminal signalling principle in relativistic quantum theory is that of microcausality. In RQFT, this assumption can be written in different ways.

\begin{definition}
    Let $\R=(\U_\R,\E_\R,\hir)$ be a relativistic QRF, $\operationalstate \subseteq \framestate$ and $\sigma \geq 0$. Then
    \begin{itemize}
        \item An operator $\phi \in \bhs$ is said to be \emph{strongly $(\operationalstate,\sigma)$-microcausal} if for all $\omega_i \in \operationalstate$ and $x_i \in \supp \mu^{\F_\R}_{\omega_i}$:
     \begin{equation}
        \label{eqn:phi strong microcausality}
        x_1 \indepsigma x_2 \; \Longrightarrow \; \comm{\hat{\phi}^\R_{\omega_1}(x_1)}{\hat{\phi}^\R_{\omega_2}(x_2)} =
        \comm{\hat{\phi}^\R_{\omega_1}(x_1)^\dagger}{\hat{\phi}^\R_{\omega_2}(x_2)} = 0 \;.
    \end{equation}
        \item An operator $\phi \in \bhs$ is said to be \emph{weakly $(\operationalstate,\sigma)$-microcausal} if for all $\omega_i \in \operationalstate$:
     \begin{equation}
        \label{eqn:phi weak microcausality}
        \omega_1 \indepERsigma \omega_2 \; \Longrightarrow \; \comm{\hat{\phi}^\R_{\omega_1}(x_1)}{\hat{\phi}^\R_{\omega_2}(x_2)} =
        \comm{\hat{\phi}^\R_{\omega_1}(x_1)^\dagger}{\hat{\phi}^\R_{\omega_2}(x_2)} = 0 \qquad \forall x_i \in \supp \mu^{\F_\R}_{\omega_i}\;.
    \end{equation}
        \item A subset $\mathcal{O}_\S \subseteq B(\his)$ is said to be \emph{strongly $(\operationalstate,\sigma)$-microcausal} if for all $\phi_i \in \mathcal{O}_\S$, $\omega_i \in \operationalstate$ and $x_i \in \supp \mu^{\F_\R}_{\omega_i}$:
    \begin{equation}
    \label{eqn:strong microcausality}
        x_1 \indepsigma x_2 \; \Longrightarrow \comm{(\hat{\phi}_1)^\R_{\omega_1}(x_1)}{(\hat{\phi}_2)^\R_{\omega_2}(x_2)} =
            \comm{(\hat{\phi}_1)^\R_{\omega_1}(x_1)^\dagger}{(\hat{\phi}_2)^\R_{\omega_2}(x_2)} = 0 \; .
    \end{equation}
        \item A subset $\mathcal{O}_\S \subset \bhs$ is said to be \emph{weakly $(\operationalstate,\sigma)$-microcausal} if for all $\phi_i \in \mathcal{O}_\S$ and $\omega_i \in \operationalstate$:
    \begin{equation}
    \label{eqn:weak microcausality}
        \omega_1 \indepERsigma \omega_2 \; \Longrightarrow \; \comm{(\hat{\phi}_1)^\R_{\omega_1}(x_1)}{(\hat{\phi}_2)^\R_{\omega_2}(x_2)} =
            \comm{(\hat{\phi}_1)^\R_{\omega_1}(x_1)^\dagger}{(\hat{\phi}_2)^\R_{\omega_2}(x_2)} = 0 \qquad \forall x_i \in \supp \mu^{\F_\R}_{\omega_i} \; .
    \end{equation}
    \end{itemize}
\end{definition}

These requirements impose that the physics \emph{is} causal in such a sense, rather than just \emph{appearing} causal. The difference between strong and weak $(\operationalstate,\sigma)$-microcausality comes from the fact that the former really represents the ontic requirement of local commutativity while the latter only asserts the commutativity of pointwise-defined fields should the preparations already be spacelike-separated, and as such carries slightly less ontic weight. For example, if $\phi$ is weakly $(\operationalstate,\sigma)$-microcausal, then if $\omega_1$ and $\omega_2$ are not $(\R,\sigma)$-spacelike separated, it does not follow that $\comm{\hat{\phi}^\R_{\omega_1}(x_1)}{\hat{\phi}^\R_{\omega_2}(x_2)} = 0$, even if $x_1 \indepsigma x_2$. The difference between both requirements is highlighted in Fig. \ref{fig:weak vs strong microcausality}.

\begin{figure}
    \centering
    \begin{subfigure}[t]{0.4\textwidth}
        \begin{tikzpicture}[x=0.75pt,y=0.75pt,yscale=-0.6,xscale=0.6]

\draw   (253.71,12) -- (584.36,12) -- (442.65,139) -- (112,139) -- cycle ;
\draw  [dash pattern={on 0.84pt off 2.51pt}]  (287.36,112.42) -- (361.68,45) ;
\draw   (256.71,163) -- (587.36,163) -- (445.65,290) -- (115,290) -- cycle ;
\draw  [pattern=_5n7rcgqhm,pattern size=6pt,pattern thickness=0.75pt,pattern radius=0.75pt, pattern color={rgb, 255:red, 0; green, 0; blue, 0}] (200,247.68) .. controls (200,245.1) and (202.1,243) .. (204.68,243) -- (274.68,243) .. controls (277.26,243) and (279.36,245.1) .. (279.36,247.68) -- (279.36,261.74) .. controls (279.36,264.32) and (277.26,266.42) .. (274.68,266.42) -- (204.68,266.42) .. controls (202.1,266.42) and (200,264.32) .. (200,261.74) -- cycle ;
\draw  [pattern=_7o1tsuk3p,pattern size=6pt,pattern thickness=0.75pt,pattern radius=0.75pt, pattern color={rgb, 255:red, 0; green, 0; blue, 0}] (378,242.68) .. controls (378,240.1) and (380.1,238) .. (382.68,238) -- (452.68,238) .. controls (455.26,238) and (457.36,240.1) .. (457.36,242.68) -- (457.36,256.74) .. controls (457.36,259.32) and (455.26,261.42) .. (452.68,261.42) -- (382.68,261.42) .. controls (380.1,261.42) and (378,259.32) .. (378,256.74) -- cycle ;
\draw  [dash pattern={on 0.84pt off 2.51pt}]  (274.68,266.42) -- (349,199) ;
\draw  [dash pattern={on 0.84pt off 2.51pt}]  (382.68,261.42) -- (315.36,202.42) ;
\draw  [pattern=_zubggx6qk,pattern size=6pt,pattern thickness=0.75pt,pattern radius=0.75pt, pattern color={rgb, 255:red, 0; green, 0; blue, 0}] (232.36,77.42) .. controls (252.36,67.42) and (255.36,72.42) .. (279.36,76.42) .. controls (303.36,80.42) and (329.36,76.42) .. (306.36,95.42) .. controls (283.36,114.42) and (278.36,126.42) .. (243.36,106.42) .. controls (208.36,86.42) and (212.36,87.42) .. (232.36,77.42) -- cycle ;
\draw   (333.36,97.42) .. controls (305.36,87.42) and (326.36,75.42) .. (306.36,95.42) .. controls (286.36,115.42) and (292.36,108.42) .. (287.36,112.42) .. controls (282.36,116.42) and (307.36,110.42) .. (312.36,109.42) .. controls (317.36,108.42) and (361.36,107.42) .. (333.36,97.42) -- cycle ;
\draw  [dash pattern={on 0.84pt off 2.51pt}]  (341.36,106.42) -- (400.36,52.42) ;
\draw   (400.36,52.42) .. controls (428.36,30.42) and (383.36,34.42) .. (380.36,31.42) .. controls (377.36,28.42) and (355.36,52.42) .. (350.36,56.42) .. controls (345.36,60.42) and (363.36,63.42) .. (377.36,66.42) .. controls (391.36,69.42) and (372.36,74.42) .. (400.36,52.42) -- cycle ;
\draw  [pattern=_mcbkrmesn,pattern size=6pt,pattern thickness=0.75pt,pattern radius=0.75pt, pattern color={rgb, 255:red, 0; green, 0; blue, 0}] (464.36,50.42) .. controls (391.36,30.42) and (418.36,40.42) .. (406.36,48.42) .. controls (394.36,56.42) and (366.36,76.42) .. (402.36,83.42) .. controls (438.36,90.42) and (441.36,78.42) .. (455.36,76.42) .. controls (469.36,74.42) and (537.36,70.42) .. (464.36,50.42) -- cycle ;
\end{tikzpicture}
    \caption{An example of strong $(\operationalstate,\sigma)$-microcausality: whether or not two preparations are $(\R,\sigma)$-spacelike separated, the ensuing relational local quantum fields commute pointwise over spacelike-separated (to resolution $\sigma$) points.}
    \end{subfigure}
    \qquad \qquad
    \begin{subfigure}[t]{0.4\textwidth}
        \begin{tikzpicture}[x=0.75pt,y=0.75pt,yscale=-0.6,xscale=0.6]

\draw   (233.71,9) -- (564.36,9) -- (422.65,136) -- (92,136) -- cycle ;
\draw  [dash pattern={on 0.84pt off 2.51pt}]  (267.36,109.42) -- (341.68,42) ;
\draw   (236.71,160) -- (567.36,160) -- (425.65,287) -- (95,287) -- cycle ;
\draw  [pattern=_72x0qoqws,pattern size=6pt,pattern thickness=0.75pt,pattern radius=0.75pt, pattern color={rgb, 255:red, 0; green, 0; blue, 0}] (180,244.68) .. controls (180,242.1) and (182.1,240) .. (184.68,240) -- (254.68,240) .. controls (257.26,240) and (259.36,242.1) .. (259.36,244.68) -- (259.36,258.74) .. controls (259.36,261.32) and (257.26,263.42) .. (254.68,263.42) -- (184.68,263.42) .. controls (182.1,263.42) and (180,261.32) .. (180,258.74) -- cycle ;
\draw  [pattern=_f1y3jgxey,pattern size=6pt,pattern thickness=0.75pt,pattern radius=0.75pt, pattern color={rgb, 255:red, 0; green, 0; blue, 0}] (358,239.68) .. controls (358,237.1) and (360.1,235) .. (362.68,235) -- (432.68,235) .. controls (435.26,235) and (437.36,237.1) .. (437.36,239.68) -- (437.36,253.74) .. controls (437.36,256.32) and (435.26,258.42) .. (432.68,258.42) -- (362.68,258.42) .. controls (360.1,258.42) and (358,256.32) .. (358,253.74) -- cycle ;
\draw  [dash pattern={on 0.84pt off 2.51pt}]  (254.68,263.42) -- (329,196) ;
\draw  [dash pattern={on 0.84pt off 2.51pt}]  (362.68,258.42) -- (295.36,199.42) ;
\draw   (212.36,74.42) .. controls (232.36,64.42) and (235.36,69.42) .. (259.36,73.42) .. controls (283.36,77.42) and (309.36,73.42) .. (286.36,92.42) .. controls (263.36,111.42) and (258.36,123.42) .. (223.36,103.42) .. controls (188.36,83.42) and (192.36,84.42) .. (212.36,74.42) -- cycle ;
\draw   (313.36,94.42) .. controls (285.36,84.42) and (306.36,72.42) .. (286.36,92.42) .. controls (266.36,112.42) and (272.36,105.42) .. (267.36,109.42) .. controls (262.36,113.42) and (287.36,107.42) .. (292.36,106.42) .. controls (297.36,105.42) and (341.36,104.42) .. (313.36,94.42) -- cycle ;
\draw  [dash pattern={on 0.84pt off 2.51pt}]  (321.36,103.42) -- (380.36,49.42) ;
\draw   (380.36,49.42) .. controls (408.36,27.42) and (363.36,31.42) .. (360.36,28.42) .. controls (357.36,25.42) and (335.36,49.42) .. (330.36,53.42) .. controls (325.36,57.42) and (343.36,60.42) .. (357.36,63.42) .. controls (371.36,66.42) and (352.36,71.42) .. (380.36,49.42) -- cycle ;
\draw   (444.36,47.42) .. controls (371.36,27.42) and (398.36,37.42) .. (386.36,45.42) .. controls (374.36,53.42) and (346.36,73.42) .. (382.36,80.42) .. controls (418.36,87.42) and (421.36,75.42) .. (435.36,73.42) .. controls (449.36,71.42) and (517.36,67.42) .. (444.36,47.42) -- cycle ;
\end{tikzpicture}
    \caption{An example of weak $(\operationalstate,\sigma)$-microcausality: if two preparations are $(\R,\sigma)$-spacelike separated (bottom picture), the ensuing relational local quantum fields commute pointwise. However, if two preparations are not $(\R,\sigma)$-spacelike separated (top picture), the ensuing relational local quantum fields need not commute pointwise even over spacelike-separated (to resolution $\sigma$) points.}
    \end{subfigure}
    \caption{A comparison of weak and strong microcausality: pointwise commutation is shown through the circles. On the left, $\phi$ is strongly $(\operationalstate,\sigma)$-microcausal: for any two preparations, the relational local quantum fields commute over spacelike-separated (to resolution $\sigma$) points. On the right, $\phi$ is weakly $(\operationalstate,\sigma)$-microcausal: relational local quantum fields commute (over spacelike-separated---to resolution $\sigma$---points) only if the preparations are $(\R,\sigma)$-spacelike separated.}
    \label{fig:weak vs strong microcausality}
\end{figure}

It is easily seen that strong $(\operationalstate,\sigma)$-microcausality implies weak $(\operationalstate,\sigma)$-microcausality. In fact, weak relational microcausality implies relational causality, as we now show.\footnote{It is unclear for now if the two conditions are equivalent, perhaps under some additional assumptions---like localizability of the QRF--- or if microcausality is strictly stronger. This issue will be subject to further investigation.}

\begin{theorem}
    \label{thm:microcausal implies causal}
    Let $\R$ be a relativistic QRF, $\operationalstate \subseteq \framestate$ and $\sigma \geq 0$. If an operator $\phi \in \bhs$ is weakly $(\operationalstate,\sigma)$-microcausal, then it is $(\operationalstate,\sigma)$-causal. Likewise, if a subset $\mathcal{O}_\S \subseteq \bhs$ is weakly $(\operationalstate,\sigma)$-microcausal, it is also $(\operationalstate,\sigma)$-causal.
\end{theorem}

\begin{proof}
    Suppose $\R$ is a relativistic QRF such that $\mathcal{O}_\S$ is weakly $(\operationalstate,\sigma)$-microcausal, and $\omega_1 \indepERsigma \omega_2$. Then for all $\phi_1,\phi_2 \in \mathcal{O}_\S$,
    \begin{equation}\begin{aligned}
        \hat{\Phi}^\R_1(\omega_1) \hat{\Phi}^\R_2(\omega_2) &= \left(\int_{\supp \mu^{\F_\R}_{\omega_1}} (\hat{\phi}_1)^\R_{\omega_1}(x_1) d\mu^{\F_\R}_{\omega_1}(x_1) \right) \left( \int_{\supp \mu^{\F_\R}_{\omega_2}} (\hat{\phi}_2)^\R_{\omega_2}(x_2) d\mu^{\F_\R}_{\omega_2}(x_2)\right) \\
        &\stackrel{\ref{thm:Fubini-Tonelli Bochner}}{=} \iint_{\supp \mu^{\F_\R}_{\omega_1} \times \supp \mu^{\F_\R}_{\omega_2}} (\hat{\phi}_1)^\R_{\omega_1}(x_1) (\hat{\phi}_2)^\R_{\omega_2}(x_2)  \, d(\mu^{\F_\R}_{\omega_1} \times \mu^{\F_\R}_{\omega_2})(x_1,x_2) \\
        &\stackrel{\omega_1 \indepERsigma \omega_2}{=} \iint_{\supp \mu^{\F_\R}_{\omega_1} \times \supp \mu^{\F_\R}_{\omega_2}} (\hat{\phi}_2)^\R_{\omega_2}(x_2) (\hat{\phi}_1)^\R_{\omega_1}(x_1) \, d(\mu^{\F_\R}_{\omega_1} \times \mu^{\F_\R}_{\omega_2})(x_1,x_2) \\
        &\stackrel{\ref{thm:Fubini-Tonelli Bochner}}{=} \iint_{\supp \mu^{\F_\R}_{\omega_2} \times\supp \mu^{\F_\R}_{\omega_1}} (\hat{\phi}_2)^\R_{\omega_2}(x_2) (\hat{\phi}_1)^\R_{\omega_1}(x_1) \, d(\mu^{\F_\R}_{\omega_2} \times \mu^{\F_\R}_{\omega_1})(x_2,x_1) \\
        &\stackrel{\ref{thm:Fubini-Tonelli Bochner}}{=} \left(\int_{\supp \mu^{\F_\R}_{\omega_2}} (\hat{\phi}_2)^\R_{\omega_2}(x_2) d\mu^{\F_\R}_{\omega_2}(x_2) \right) \left( \int_{\supp \mu^{\F_\R}_{\omega_1}} (\hat{\phi}_1)^\R_{\omega_1}(x_1) d\mu^{\F_\R}_{\omega_2}(x_1)\right) \\
        &= \hat{\Phi}^\R_2(\omega_2) \hat{\Phi}^\R_1(\omega_1) \, .
    \end{aligned}\end{equation}
    The same argument holds for $\hat{\Phi}^\R_1(\omega_1)^\dagger \hat{\Phi}^\R_2(\omega_2)$, which concludes the proof regarding $\mathcal{O}_\S$. The proof of the first statement is a simplified version of this argument.
\end{proof}

Microcausality thus ensures that Einstein causality holds, while allowing one to make contact with \enquote{physicists'} approaches to QFT. However, as for the case of Poincaré-covariance, making the logical leap of imposing local commutativity pointwise (as an ontological principle) rather than at the level of relational local observables (as an epistemic principle) can sometimes be too strong. This is explicit in Wightman's no go theorem.

\begin{theorem}[Wightman \cite{wightman_theorie_1964}]
    \label{thm:Wightman}
	Let $\hat{\phi} : \mink \to \bhs$ be an operator-valued function and $U$ be a weakly continuous unitary representation of the translation group on $\his$ such that
	\begin{enumerate}
		\item $\comm{\hat{\phi}(x)}{\hat{\phi}(y)} = 0$ whenever $x$ and $y$ are spacelike separated,
		\item $U(x)^\dagger \hat{\phi}(y) U(x) = \hat{\phi}(y+x)$ for all $x,y \in \mink$,
		\item $U$ satisfies the spectrum condition,
		\item There is a unique translation invariant vector $\ket{\Omega} \in \his$.
	\end{enumerate}
	Then there is a $c \in \mathbb{C}$ such that $\hat{\phi}(x) \ket{\Omega} = c\ket{\Omega}$ for all $x \in \mink$.
\end{theorem}

See \cite{wightman_theorie_1964,halvorson_algebraic_2006,horuzhy_axiomatic_1990,fedida_pointwise_2025} for proofs of the above. We must thus evade this theorem by either rejecting microcausality of the fields or pointwise translation covariance (the spectrum condition being arguably a very reasonable assumption). That is, the joint ontological implementation of two epistemic principles are too strong to be made. In RQFT, the relational local quantum fields are generally not pointwise translation covariant (since $\omega \stackrel{a}{\mapsto} \omega \cdot a^{-1}$, see Thm. \ref{thm: scalar covariance}), so microcausality can be implemented nontrivially for some relativistic QRFs. However, globally oriented relativistic QRFs do give rise to pointwise covariant quantum fields defined over all of $\mink$, so such QRFs cannot be microcausal. This leads to the following no-go result for strong $(\operationalstate,0)$-microcausality by taking $\omega_1 = \omega_2 \equiv \omega$.\footnote{Whether strong $(\operationalstate,\sigma)$-microcausality can be saved under the joint assumption with global orientation whenever $\sigma \gneq 0$ is an interesting open question. Perhaps the global nature of local commutativity \cite{streater_pct_1989} can be translated from the Wightman QFT setting to that of RQFT when one assumes global orientation, which would thus generalise this no-go result to all $\sigma \geq 0$. Note that the relevance of strong $(\operationalstate,\sigma)$-microcausality for $\sigma > 0$ only appears if $\sigma$ is understood as a fundamental scale rather than the resolution of the frame---otherwise the ontic nature of the principle would be very much questionable.}

\begin{theorem}
    \label{thm:no go globally oriented vacuum}
    Let $(\R,\omega)$ be an oriented relativistic QRF and $\omega \in \operationalstate \subseteq \framestate$. Then the following cannot all jointly hold:
    \begin{enumerate}
        \item $\phi \in \bhs$ is strongly $(\operationalstate,0)$-microcausal,
        \item $(\R,\omega)$ is globally oriented,
        \item $U_\S$ satisfies the spectrum condition,\footnote{The spectrum condition states that $\sigma(P_\mu) \subset \overline{V_+} = \{p \mid p^0 \geq 0, p^2 \geq 0\}$, i.e. the joint spectrum $\sigma(P_\mu)$ of the energy-momentum lies in the forward causal cone.}
        \item There exists a unique (up to scalar multiples) translation-invariant vector $\ket{\Omega} \in \his$,
        \item There exists a $\phi \in \bhs$ and $x_1,x_2 \in \mink$ such that $\hat{\phi}^\R_{\omega}(x_1) \ket{\Omega} \neq \hat{\phi}^\R_{\omega}(x_2) \ket{\Omega}$.
    \end{enumerate}
\end{theorem}
    
Hence, the globally oriented preparations of QRFs in which strong microcausality holds yield a trivial description of physics, in which the vacuum expectation values are constant over spacetime. It is thus tricky to impose a causality requirement for globally oriented QRFs: these are not affected by relational causality or weak relational microcausality, and cannot satisfy strong microcausality without running into triviality.

We now provide an explicit example of operators which satisfy finite-precision weak $(\operationalstate^K,\sigma)$-microcausality for any chosen compact $K \subset F$ and arbitrarily small $\sigma > 0$.

\begin{theorem}
    \label{thm:Weyl approximate microcausality}
    Let $\R$ be a relativistic QRF, $K$ be any compact subset of $F$, $W_{\hat{\phi}} : \mathscr{S}(\mink,\mathbb{C}) \to \bhs$ be the Weyl operator for a free Klein-Gordon Wightmanian quantum field $\hat{\phi}$\footnote{These satisfy $\comm{W_{\hat{\phi}}(f)}{W_{\hat{\phi}}(g)}=0$ for all $f,g \in \mathscr{S}(\mink,\mathbb{C})$ whenever $\supp f \indep \supp g$, and $W_{\hat{\phi}}(f)^\dagger = W_{\hat{\phi}}(-f)$. We will review Wightman QFT in Sec. \ref{sec:Wightman comparison}.}. Then $\forall \sigma>0$, $\exists f_{\sigma,K} \in C^\infty_c(\mink)$ such that $W_{\hat{\phi}}(f_{\sigma,K}) \in \bhs$ is weakly $(\operationalstate^K,\sigma)$-microcausal.
\end{theorem}

\begin{proof}
    See App. \ref{proof approximate commutativity}.
\end{proof}
This constitutes an explicit example of (finite-precision) relational microcausality for a nontrivial subset of frame preparations. Whether this result can be extended to exact relational (micro)causality over $\operationalstate^c$ is an interesting open question. Likewise, it would be interesting to discover other (nontrivial) examples of system operators satisfying relational causality. 

Regardless of the outcome such investigation---which we postpone to further work---Thm. \ref{thm:Weyl approximate microcausality} provides explicit examples of (covariant) scalar microcausal relational local quantum fields, and highlights that causality can be understood relationally.

\section{Vacuum expectation values}

\label{sec:vevs}

It is important that the mathematical framework developed up until now makes some contact with the empirical content of nonrelational QFT. In \enquote{physicists'} QFT, the main quantities of interest from the operational point of view are arguably the vacuum expectation values. We will thus assume the existence of a (potentially mixed) Poincaré-invariant vacuum state $\Omega \in \systemstate^{\Poincup}$ on the system. Though we are aware of potential issues this may induce down the line, in particular on the side of Haag's no-go theorem \cite{haag_quantum_1955,hall_theorem_1957}, we take it to be a good starting point, at least to discuss in free scalar theories. It may be that RQFT avoids falling into Haag's theorem in an interesting subtle way, or that some assumptions may need to be weakened. Notice also that in the context of operationally meaningful---i.e., localized in precompact regions of spacetime (so e.g. not globally oriented)---oriented QRFs, Haag’s theorem is not an obstacle that should apply anyway.

\subsection{Wightman functions}

\label{sec:Wightman functions}

We now derive the properties of vacuum expectation values in RQFT, establishing their close resemblance to those encountered in Wightman QFT \cite{streater_pct_1989}.

\begin{definition}
	Let $\R$ be a relativistic QRF, $\Omega \in \vacstate$, $n \in \mathbb{N}$, $\omega_1,\cdots,\omega_n\in \framestate$ and $\phi_1,\cdots,\phi_n \in \bhs$. We define the \emph{n-point vacuum expectation values} as
	\begin{equation}\begin{aligned}
	\Wscr_n^{(\Omega,\R)}[\omega_1,\cdots,\omega_n] : \bhs^n &\to \mathbb{C} \\
	\Wscr_n^{(\Omega,\R)}[\omega_1,\cdots,\omega_n](\phi_1,\cdots,\phi_n) &:= \Tr\left[\Omega \prod_{i=1}^n \hat{\Phi}_i^{\R}(\omega_i) \right] \, .
	\end{aligned}\end{equation}
\end{definition}

$\Wscr_n^{(\Omega,\R)}$ has a clear operational meaning: if one prepares a frame in several different orientations, and measures different observables in these respective orientations, what are the vacuum expectation values for these observables that the theory predicts? We can then define the kernel of these operators by introducing the relational local quantum fields into the n-point expectation values above:
\begin{equation}
	\Wscr_n^{(\Omega,\R)}[\omega_1,\cdots,\omega_n](\phi_1,\cdots,\phi_n) = \idotsint_{\mink^n} \Tr\left[\Omega \left(\prod_{i=1}^n (\hat{\phi}_i)^\R_{\omega_i}(x_i)\right)\right] \, d\mu^{\F_\R}_{\omega_1}(x_1) \cdots d\mu^{\F_\R}_{\omega_n}(x_n) \, .
\end{equation}

\begin{definition}\label{def:Wightman kernels}
	Let $\R$ be a relativistic QRF, $\Omega \in \vacstate$, $n \in \mathbb{N}$ and $\phi_1, \cdots, \phi_n \in \bhs$. Given $\omega_1,\cdots,\omega_n \in \framestate$ and $x_1,\cdots,x_n \in \mink$, we define the \emph{n-point vacuum kernels} as
	\begin{equation}\begin{aligned}
		W_n^{(\Omega,\R)}[\omega_1,\cdots,\omega_n;x_1,\cdots,x_n] : \bhs^n &\to \mathbb{C} \\
		W_n^{(\Omega,\R)}[\omega_1,\cdots,\omega_n;x_1,\cdots,x_n](\phi_1,\cdots,\phi_n) &:= \Tr\left[\Omega \left(\prod_{i=1}^n (\hat{\phi}_i)^\R_{\omega_i}(x_i)\right) \right] \, .
	\end{aligned}\end{equation}
\end{definition}

In Wightman QFT, the Schwartz nuclear theorem (Thm. \ref{thm:Schwartz nuclear theorem}) ensures that kernels of n-point vacuum expectation values exist, but these are only \enquote{morally} the n-point vacuum expectation values of the pointwise quantum fields. Here, we see that $W_n^{(\Omega,\R)}$ is more than just being \enquote{morally} $\Tr[\Omega (\hat{\phi}_1)^\R_{\omega_1}(x_1) \cdots (\hat{\phi}_n)^\R_{\omega_n}(x_n)]$---it actually always is exactly that; $W_n^{(\Omega,\R)}$ is really the integral kernel of $\Wfrak_n^{(\Omega,\R)}$, i.e.,
\begin{equation}
    \Wscr_n^{(\Omega,\R)}[\omega_1,\cdots,\omega_n](\phi_1,\cdots,\phi_n)= \idotsint_{\mink^n}  
    W_n^{(\Omega,\R)}[\omega_1,\cdots,\omega_n;x_1,\cdots,x_n](\phi_1,\cdots,\phi_n)
    \, d\mu^{\F_\R}_{\omega_1}(x_1) \cdots d\mu^{\F_\R}_{\omega_n}(x_n) \,.
\end{equation}

Again, if the oriented QRFs are not localisable, then these pointwise kernels do not have an operational meaning, but are of mathematical convenience for several important theorems.
Analogous results to Wightman QFT \cite{streater_pct_1989} can then immediately be recovered in the language of RQFT.

\begin{proposition}[Relativistic Transformation Law]
    \label{prop:symmetry of vevs}
    Let $\R$ be a relativistic QRF, $\Omega \in \vacstate$. Then for all $\omega_1,\cdots,\omega_n \in \framestate$, $n \in \mathbb{N}$, and $\forall (a,\Lambda) \in \Poincup$,
    \begin{equation}
        \Wscr_n^{(\Omega,\R)}[\omega_1 \cdot (a,\Lambda),\cdots,\omega_n \cdot (a,\Lambda)] = \Wscr_n^{(\Omega,\R)}[\omega_1,\cdots,\omega_n] \, .
    \end{equation}
    Moreover, for $\mu^{\F_\R}_{\omega_i}$-a.e. $x_i$, $i=1,\cdots,n$,
    \begin{equation}
        W_n^{(\Omega,\R)}[\omega_1 \cdot (a,\Lambda),\cdots, \omega_n \cdot (a,\Lambda); x_1,\cdots,x_n] = W_n^{(\Omega,\R)}[\omega_1,\cdots,\omega_n;\Lambda x_1 + a,\cdots,\Lambda x_n + a] \, . \label{eqn:Poincare covariance Wightman}
    \end{equation}    
\end{proposition}

\begin{proof}
    See App. \ref{proof symmetry of vevs}.
\end{proof}

We see that the n-point vacuum expectation values for scalars are invariant under these relativistic transformation laws. This is important and expected for such quantities. Note however that the n-point vacuum kernels \emph{shift} in their supports as one applies Poincaré transformations pointwise, which is expected from the probability measures that may have bounded supports, though deviates from the kernel results of Wightman QFT for which Poincaré invariance also holds pointwise. Let us examine what happens when one considers pointwise translation covariance with no frame state shift.

\begin{proposition}[Global Orientation]
    \label{prop:global orientation}
    Let $\R$ be a relativistic QRF, $\Omega \in \vacstate$ and $\omega_1,\cdots,\omega_n \in \framestate$, $\xi_1,\cdots,\xi_{n-1} \in \mink^{n-1}$, $n \in \mathbb{N}$. Let
    \begin{equation}\begin{aligned}
        \Wbf_n^{(\Omega,\R)}[\omega_1,\cdots,\omega_n;\xi_1,\cdots,\xi_{n-1}] &: \bhs^n \to \mathbb{C} \\
        \Wbf_n^{(\Omega,\R)}[\omega_1,\cdots,\omega_n;\xi_1,\cdots,\xi_{n-1}]&(\phi_1,\cdots,\phi_n) \nonumber \\:= &\Tr\left[\Omega U_\S\left(\sum_{i=1}^{n-1} \xi_i,e\right)^\dagger (\hat{\phi}_1)^\R_{\omega_1}(0) \prod_{j=2}^{n} \left(U_\S(\xi_{j-1},e) (\hat{\phi}_j)^\R_{\omega_j}(0) \right)\right] \, .
    \end{aligned}\end{equation}
    If the $(\R,\omega_i)$ are globally oriented, $i=1,\cdots,n$, then 
    \begin{equation}
        \Wbf_n^{(\Omega,\R)}[\omega_1,\cdots,\omega_n;\xi_1,\cdots,\xi_{n-1}] = W_n^{(\Omega,\R)}[\omega_1,\cdots,\omega_n;x_1,\cdots,x_n]
    \end{equation}
    where $\xi_j = x_j - x_{j+1}$ for $j=1,\cdots,n-1$. In particular, if $(\R,\omega_1)$ and $(\R,\omega_2)$ are globally oriented, then for all $x,y \in \mink$,
        \begin{equation}
            W^{(\Omega,\R)}_{2}[\omega_1,\omega_2;x,y] = W^{(\Omega,\R)}_2[\omega_1,\omega_2;-y,-x] \, .
        \end{equation}
\end{proposition}

\begin{proof}
    See App. \ref{proof global orientation}.
\end{proof}

Hence, for globally oriented QRFs in which the relational local quantum fields are pointwise translation covariant with no shift in the frame's state, we see that the n-point vacuum kernels only depend on the distance between the field locations. This is a standard property of \enquote{physicists'} QFT as well as Wightman QFT. Assuming this property holds, one can analyse the spectral properties of the n-point vacuum kernels.

\begin{proposition}[Spectral Conditions]
    \label{prop:spectrum condition vevs}
    Let $\R$ be a relativistic QRF, $\Omega \in \vacstate$ and $\omega_1,\cdots,\omega_n \in \framestate$, $\xi_1,\cdots,\xi_{n-1} \in \mink^{n-1}$, $n \in \mathbb{N}$. If the $(\R,\omega_i)$ are globally oriented, $i=1,\cdots,n$, then the Fourier transforms of $W_n^{(\Omega,\R)}[\omega_1,\cdots,\omega_n;x_1,\cdots,x_n]$ and $\Wbf_n^{(\Omega,\R)}[\omega_1,\cdots,\omega_n;\xi_1,\cdots,\xi_n]$, defined for all $\phi_1,\cdots,\phi_n \in \bhs$ by
    \begin{equation}\begin{aligned}
        \tilde{W}_n^{(\Omega,\R)}&[\omega_1,\cdots,\omega_n;p_1,\cdots,p_n](\phi_1,\cdots,\phi_n) \nonumber \\ &:= \idotsint \exp\left(i \sum_{i=1}^n p_i x_i \right) W_n^{(\Omega,\R)}[\omega_1,\cdots,\omega_n;x_1,\cdots,x_n](\phi_1,\cdots,\phi_n) d^dx_1 \cdots d^dx_n\\
        \tilde{\Wbf}_n^{(\Omega,\R)}&[\omega_1,\cdots,\omega_n;q_1,\cdots,q_{n-1}](\phi_1,\cdots,\phi_n) \nonumber \\ &:= \idotsint \exp\left(i \sum_{j=1}^{n-1} q_j \xi_j \right) \Wbf_n^{(\Omega,\R)}[\omega_1,\cdots,\omega_n;\xi_1,\cdots,\xi_{n-1}](\phi_1,\cdots,\phi_n) d^d\xi_1 \cdots d^d\xi_{n-1}
    \end{aligned}\end{equation}
    for all $p_1,\cdots,p_n \in \mathbb{R}^d$ and all $q_1,\cdots,q_{n-1} \in \mathbb{R}^d$, are related by
    \begin{equation}
        \tilde{W}_n^{(\Omega,\R)}[\omega_1,\cdots,\omega_n;p_1,\cdots,p_n] = (2\pi)^d \delta\left(\sum_{i=1}^n p_j\right) \tilde{\Wbf}_n[\omega_1,\cdots,\omega_n;p_1,p_1+p_2,\cdots,p_1 + p_2 + \cdots +p_{n-1}] \, . 
    \end{equation}
    Furthermore, provided the spectrum condition holds for $U_\S$,\footnote{The spectrum condition states that $\sigma(P_\mu) \subset \overline{V_+} = \{p \mid p^0 \geq 0, p^2 \geq 0\}$, i.e. the joint spectrum of the generators of spacetime translations of $U_\S$ lies in the forward causal cone.}
        \begin{equation}
            \tilde{\Wbf}^{(\Omega,\R)}_n[\omega_1,\cdots,\omega_n;q_1,\cdots,q_{n-1}] = 0
        \end{equation}
        if any $q_1,\cdots,q_{n-1} \notin \sigma(P_\mu)$.
\end{proposition}

\begin{proof}
    See App. \ref{proof spectrum condition vevs}.
\end{proof}

For globally oriented QRFs, one can therefore understand momentum-space vacuum expectation values in terms of the spectrum of the generators of Poincaré transformations on the system. These momentum space contributions vanish outside the spectrum of the generators of translations, as is well-understood in \enquote{standard} quantum field theory. Let us now also highlight the Hermiticity of these Wightman functions.

\begin{proposition}[Hermiticity Conditions]
    \label{prop:Hermiticity Wightman functions}
    For all $\phi_1,\cdots,\phi_n \in \bhs$ and all $x_1,\cdots,x_n \in \mink$,
        \begin{equation}\begin{aligned}
        W_n^{(\Omega,\R)}[\omega_1,\cdots,\omega_n;x_1,\cdots,x_{n}](\phi_1,\cdots,\phi_n) &= \overline{W_n^{(\Omega,\R)}[\omega_n,\cdots,\omega_1;x_{n},\cdots,x_1](\phi_n^\dagger,\cdots,\phi_1^\dagger)} \\
            \Wbf_n^{(\Omega,\R)}[\omega_1,\cdots,\omega_n;\xi_1,\cdots,\xi_{n-1}](\phi_1,\cdots,\phi_n) &= \overline{\Wbf_n^{(\Omega,\R)}[\omega_n,\cdots,\omega_1;-\xi_{n-1},\cdots,-\xi_1](\phi_n^\dagger,\cdots,\phi_1^\dagger)}
        \end{aligned}\end{equation}
        where $\xi_j = x_j - x_{j+1}$ for all $j \in \{1,\cdots,n-1\}$ and the latter holds if the $(\R,\omega_i)$ are globally oriented, $i=1,\cdots,n$, and so
        \begin{equation}
            \Wscr_n^{(\Omega,\R)}[\omega_1,\cdots,\omega_n](\phi_1,\cdots,\phi_n) = \overline{\Wscr_n^{(\Omega,\R)}[\omega_n,\cdots,\omega_1](\phi_n^\dagger,\cdots,\phi_1^\dagger)} \, .
        \end{equation}
\end{proposition}

This is an important property to highlight if one is interested in reconstructing a Hilbert space from the vacuum expectation values alone. It would be interesting to understand whether such a reconstruction theorem, well-established in Wightman QFT and in Osterwalder-Schrader QFT \cite{osterwalder_axioms_1973}, has a counterpart in such a relational context---in particular, could one reconstruct the \enquote{full absolute} system Hilbert space $\his$ from these vacuum expectation values in one single QRF $\R$?  \\

Another important feature of vacuum expectation values, now in the context of general, not necessarily globally oriented relativistic QRFs, is the local commutativity conditions, which relate to the causal conditions of the QRFs.

\begin{proposition}[Local Commutativity Conditions]
    \label{prop:local commutativity vevs}
    Let $\R$ be a relativistic QRF, $\sigma \geq 0$, $\Omega \in \vacstate$, $\mathcal{O}_\S \subseteq \bhs$ and $\omega_1,\cdots,\omega_n \in \operationalstate$. If 
        \begin{enumerate}
            \item $\mathcal{O}_\S$ is $(\operationalstate,\sigma)$-causal and $\omega_j \indepERsigma \omega_{j+1}$, $j \leq n \in \mathbb{N}$, then if $\phi_j,\phi_{j+1} \in \mathcal{O}_\S$,
            \begin{multline}
            \Wscr_n^{(\Omega,\R)}[\omega_1,\cdots,\omega_n](\phi_1,\cdots,\phi_n) \\ = \Wscr_n^{(\Omega,\R)}[\omega_1,\cdots,\omega_{j-1},\omega_{j+1},\omega_j,\omega_{j+2},\cdots,\omega_n](\phi_1,\cdots,\phi_{j-1},\phi_{j+1},\phi_j,\phi_{j+2},\cdots,\phi_n) \, .
        \end{multline}
            \item $\mathcal{O}_\S$ is weakly $(\operationalstate,\sigma)$-microcausal, $\omega_j \indepERsigma \omega_{j+1}$ and $\phi_j,\phi_{j+1} \in \mathcal{O}_\S$, then $\forall x_i \in \supp \mu^{\F_\R}_{\omega_i}$, $i =1,\cdots,n$,
            \begin{multline}
            W_n^{(\Omega,\R)}[\omega_1,\cdots,\omega_n;x_1,\cdots,x_n](\phi_1,\cdots,\phi_n) \\ = W_n^{(\Omega,\R)}[\omega_1,\cdots,\omega_{j-1},\omega_{j+1},\omega_j,\omega_{j+2},\cdots,\omega_n;x_1,\cdots,x_{j-1},x_{j+1},x_j,x_{j+2},\cdots,x_n] \\(\phi_1,\cdots,\phi_{j-1},\phi_{j+1},\phi_j,\phi_{j+2},\cdots,\phi_n) \, .
        \end{multline}
            \item $\mathcal{O}_\S$ is strongly $(\operationalstate,\sigma)$-microcausal, $x_j \indepsigma x_{j+1}$ and $\phi_j,\phi_{j+1} \in \mathcal{O}_\S$, then 
            \begin{multline}
            W_n^{(\Omega,\R)}[\omega_1,\cdots,\omega_n;x_1,\cdots,x_n](\phi_1,\cdots,\phi_n) \\ = W_n^{(\Omega,\R)}[\omega_1,\cdots,\omega_{j-1},\omega_{j+1},\omega_j,\omega_{j+2},\cdots,\omega_n;x_1,\cdots,x_{j-1},x_{j+1},x_j,x_{j+2},\cdots,x_n] \\(\phi_1,\cdots,\phi_{j-1},\phi_{j+1},\phi_j,\phi_{j+2},\cdots,\phi_n) \, .
        \end{multline}
        \end{enumerate}
\end{proposition}

Different causal assumptions thus lead to different algebraic conditions for the vacuum expectation values and their kernels. $(\operationalstate,\sigma)$-causality yields the \enquote{standard} commutativity conditions of the n-point vacuum expectation values, while weak and strong $(\operationalstate,\sigma)$-microcausality imply the \enquote{standard} commutativity conditions for the kernels.  \\

We finish by showing the positivity of the n-point vacuum expectation values in RQFT, which relates to the positive definiteness of the scalar product in Hilbert space in the Wightman framework \cite{streater_pct_1989}.

\begin{proposition}[Positive Definite Conditions]
    \label{prop:Positivity vevs}
    Let $\R$ be a relativistic QRF, $\Omega \in \vacstate$ and $n \in \mathbb{N}$. For all $(\phi_i)_{i=1}^{2n} \in \bhs$ and all $(\omega_{lm})_{l \leq m = 1}^n \in \framestate$,
    \begin{equation}
        \label{eqn:positivity}
        \sum_{j=1}^n\sum_{k=1}^n\Wscr_{j+k}^{(\Omega,\R)}[\omega_{1j},\cdots,\omega_{jj},\omega_{1k},\cdots,\omega_{kk}](\phi_{jj}^\dagger,\cdots,\phi_{j1}^\dagger,\phi_{k1},\cdots,\phi_{kk}) \geq 0
    \end{equation}
\end{proposition}

\begin{proof}
    See App. \ref{proof of Positivity vevs}.
\end{proof}

We thus recover many of the results from Wightman QFT which pertain to the properties of the vacuum expectation values. Let us now examine whether the cluster decomposition property also holds for these vacuum expectation values in the RQFT setting.

\subsection{Relational cluster decomposition theorems}
\label{subsec:cluster decomposition}

In standard Wightman QFT, cluster decomposition (i.e. the factorisation of expectation values across large separations) is related to the uniqueness of a pure vacuum state \cite{streater_pct_1989}. Let us first examine some conditions on this purity and uniqueness.

\begin{lemma}
    \label{lem:pure unique vacuum implies unique invariant vector}
    Suppose $\Poincup \stackrel{U_\S}{\curvearrowright} \his$. If $\exists! \, \Omega \in \systemstate^{T(1,d-1)}$, then $\Omega \in \vacstate$ (and is the unique Poincaré-invariant state). Moreover, $\Omega$ is also pure iff $\exists! \ket{\Omega} \in \his$ normalised such that $\Omega = \ket{\Omega}\bra{\Omega}$ and $\forall a \in T(1,d-1), \, U(a,e)\ket{\Omega} = \ket{\Omega}$ (i.e. $\Omega$ is strongly invariant under $T(1,d-1)$).
\end{lemma}

\begin{proof}
    See App. \ref{proof pure unique vacuum implies unique invariant vector}.
\end{proof}

In Wightman QFT, one important property that the vacuum has is that it is cyclic for the polynomial algebra of the Wightman quantum fields. Let us see how that translates to the relational setting.

\begin{definition}
    Let $\R$ be a relativistic QRF, $\operationalstate \subseteq \framestate$ and $\mathcal{O}_\S \subseteq \bhs$. The \emph{polynomial relational algebra} is the set
    \begin{equation}
        \mathscr{P}(\mathcal{O}_\S)_{\operationalstate} := \left\{\sum_{i=1}^n \prod_{j=1}^m c_{ij} \hat{\Phi}_{ij}^\R(\omega_{ij}) \mid n,m \in \mathbb{N}, \omega_{ij}\in \operationalstate, c_{ij} \in \mathbb{C}, \phi_{ij} \in \mathcal{O}_\S\right\} \, .
    \end{equation}
\end{definition}

\begin{definition}
    A vector $\ket{\psi} \in \his$ is said to be \emph{cyclic} for a subalgebra $\A \subseteq \bhs$ if
    \begin{equation}
        \{A \ket{\psi} \mid A \in \A\}^{cl} = \his \, .
    \end{equation}
    $\ket{\psi} \in \his$ is said to be $(\mathcal{O}_\S,\operationalstate)$-cyclic if it is cyclic for the polynomial algebra $\mathscr{P}(\mathcal{O}_\S)_{\operationalstate}$.
\end{definition}

The notion of cluster property can be described in the context of Hilbert spaces directly, as follows \cite{bogolubov_relativistic_1990}.

\begin{definition}
    Let $\hi$ be a separable complex Hilbert space which is being acted on $\Poincup \stackrel{U}{\curvearrowright} \hi$ where $U$ is an ultraweakly continuous representation of the universal cover of the proper orthochronous Poincaré group. We say that the \emph{cluster property} holds in $\hi$ if $\exists \ket{\psi_0} \in \hi$ normalised such that $\forall a \in T(1,d-1)$ spacelike and $\forall \ket{\Psi}, \ket{\Phi} \in \hi$,
    \begin{equation}
        \label{eqn:cluster property HS}
        \lim_{\lambda \to \infty} \braket{\Phi}{U(\lambda a,e) \Psi} = \braket{\Phi}{\psi_0}\braket{\psi_0}{\Psi} \, .
    \end{equation}
\end{definition}

This notion of cluster property of the Hilbert space is indeed related to the uniqueness of a pure vacuum state.

\begin{theorem}[Prop. 7.1 in \cite{bogolubov_relativistic_1990}]
    \label{thm:Cluster property HS}
    Suppose $\Poincup \stackrel{U_\S}{\curvearrowright} \his$ where $U_\S$ is an ultraweakly continuous representation of $\Poincup$. Then
    \begin{enumerate}
        \item If $\his$ has the cluster property, then $\exists ! \, \Omega = \ket{\Omega}\bra{\Omega} \in \vacstate$ such that $\forall a \in T(1,d-1), \, U_\S(a,e) \ket{\Omega} = \ket{\Omega}$,
        \item If $\exists ! \, \Omega = \ket{\Omega}\bra{\Omega} \in \vacstate$ such that $\forall a \in T(1,d-1), \, U_\S(a,e) \ket{\Omega} = \ket{\Omega}$ and $U_\S$ satisfies the spectrum condition, then $\his$ has the cluster property.
    \end{enumerate}
\end{theorem}

We now relate this notion of cluster property of a Hilbert space to the context of RQFT.

\begin{definition}
    Let $\R$ be a relativistic QRF, $\Omega \in \vacstate$, $\mathcal{O_\S} \subseteq \bhs$ and $\operationalstate \subseteq \framestate$. Then the tuple $(\mathcal{O}_\S,\Omega,\operationalstate)$ is said to satisfy
    \begin{enumerate}
        \item the \emph{strong relational cluster decomposition property} if $\forall a \in T(1,d-1)$ spacelike, $\forall n \in \mathbb{N}_{\geq 2}$, $\forall j \in \{1,\cdots,n\}$, $\forall \omega_1,\cdots,\omega_n \in \operationalstate$, $\forall x_1,\cdots,x_n \in \mink$ and $\forall \phi_1,\cdots,\phi_n \in \mathcal{O}_\S$,
        \begin{multline}
            \lim_{\lambda \to \infty} W_{n}^{(\Omega,\R)}[\omega_1,\cdots,\omega_j,(\lambda a,e) \cdot \omega_{j+1}, \cdots,(\lambda a,e) \cdot\omega_n;x_1,\cdots,x_j,x_{j}+\lambda a,\cdots,x_n + \lambda a](\phi_1,\cdots,\phi_n) \\ = W_{j}^{(\Omega,\R)}[\omega_1,\cdots,\omega_j;x_1,\cdots,x_j](\phi_1,\cdots,\phi_j) \\ \cdot W_{n-j}^{(\Omega,\R)} [\omega_{j+1},\cdots,\omega_n;x_{j+1},\cdots,x_n](\phi_{j+1},\cdots,\phi_n)
        \end{multline}
        \item the \emph{weak relational cluster decomposition property} if $\forall a \in T(1,d-1)$ spacelike, $\forall n \in \mathbb{N}_{\geq 2}$, $\forall j \in \{1,\cdots,n\}$, $\forall \omega_1,\cdots,\omega_n \in \operationalstate$ and $\forall \phi_1,\cdots,\phi_n \in \mathcal{O}_\S$,
        \begin{multline}
            \lim_{\lambda \to \infty} \Wscr_{n}^{(\Omega,\R)}[\omega_1,\cdots,\omega_j,(\lambda a,e) \cdot\omega_{j+1},\cdots,(\lambda a,e) \cdot \omega_n](\phi_1,\cdots,\phi_n) \\ = \Wscr_{j}^{(\Omega,\R)}[\omega_1,\cdots,\omega_j](\phi_1,\cdots,\phi_j) \cdot \Wscr_{n-j}^{(\Omega,\R)} [\omega_{j+1},\cdots,\omega_n](\phi_{j+1},\cdots,\phi_n)
        \end{multline}
    \end{enumerate}
\end{definition}

We now show that the cluster property of $\his$ implies the strong relational cluster decomposition property, which in turn implies the weak relational cluster decomposition property.

\begin{lemma}
    \label{lem:cluster}   
    Let $\R$ be a relativistic QRF, $\Omega \in \vacstate$, $\mathcal{O_\S} \subseteq \bhs$ and $\operationalstate \subseteq \framestate$. Then
    \begin{enumerate}
        \item If $\his$ has the cluster property, then $\Omega$ is pure and is the unique vacuum state and is strongly translation-invariant and $(\mathcal{O}_\S,\Omega,\operationalstate)$ satisfies the strong relational cluster decomposition property (for any $\mathcal{O}_\S,\operationalstate$).
    \item If $(\mathcal{O}_\S,\Omega,\operationalstate)$ satisfies the strong relational cluster decomposition property, then it satisfies the weak relational cluster decomposition property.
    \end{enumerate}
\end{lemma}

\begin{proof}
    See App. \ref{proof cluster}.
\end{proof}

Under additional assumptions, these three notions of cluster property actually come together as follows.

\begin{theorem}
    \label{thm:Cluster properties}
    Let $\R$ be a relativistic QRF, $\mathcal{O_\S} \subseteq \bhs$ be closed under taking adjoints, $\operationalstate \subseteq \framestate$ and $\Omega = \ket{\Omega}\bra{\Omega} \in \vacstate$ be $(\mathcal{O}_\S,\operationalstate)$-cyclic. Then the following are equivalent:
    \begin{enumerate}
        
        \item $(\mathcal{O}_\S,\Omega,\operationalstate)$ satisfies the weak relational cluster decomposition property,
        \item $(\mathcal{O}_\S,\Omega,\operationalstate)$ satisfies the strong relational cluster decomposition property,
        \item $\his$ has the cluster property,
        \item $\Omega$ is the unique vacuum state and is strongly translation-invariant with respect to $U_\S$. 
    \end{enumerate}
\end{theorem}

\begin{proof}
    See App. \ref{proof Cluster properties}.
\end{proof}

We thus see that under certain common assumptions, all these notions of cluster property become equivalent. The derivation of the spin-statistics and CPT theorems, as well as of a reconstruction theorem from vacuum expectation values, is left to future work.

\subsection{Time-ordering}

\label{subsec:Feynman}

In physics, an important class of objects are the time-ordered correlation functions. These are related to the Wightman vacuum expectation values. First, let us define time-ordering, which for now is only well-defined under the assumption of (exact) strong microcausality.

\begin{definition}
    Let $\mathcal{O}_\S \subseteq \bhs$ be strongly $(\operationalstate,0)$-microcausal and $\omega_1,\cdots,\omega_n \in \operationalstate$ where $n \in \mathbb{N}$. We define the \emph{time-ordering operator} as the map
    \begin{equation}\begin{aligned}
        \T^\R_{n}[\omega_1,\cdots,\omega_n;x_1,\cdots,x_n] : \mathcal{O}_\S^n &\to \bhs \\
        \T^\R_{n}[\omega_1,\cdots,\omega_n;x_1,\cdots,x_n](\phi_1,\cdots,\phi_n) &:= \sum_{\sigma \in S_n} \left(\prod_{i=1}^{n-1} \Theta(\tau_{\sigma(i)} - \tau_{\sigma(i+1)})\right) \prod_{j=1}^{n} (\hat{\phi}_{\sigma(j)})^\R_{\omega_{\sigma(j)}}(x_{\sigma(j)}) \label{def: time-ordering}
    \end{aligned}\end{equation}
    where the sum runs over all permutations of the symmetry group $S_n$, $\Theta$ is the Heaviside step function and $\tau_i$ is the time-coordinate of $x_i$, $i=1,\cdots,n$, in some coordinate chart of $\mink$.
\end{definition}

The strong microcausality assumption is important for this map to be unambiguously defined: over spacelike separations, since the fields commute, the choice of coordinate system is irrelevant. For example, we have
\begin{equation}
\T^\R_{2}[\omega_1,\omega_2;x_1,x_2](\phi_1,\phi_2) = \begin{cases}
    (\hat{\phi}_1)^\R_{\omega_1}(x_1) (\hat{\phi}_2)^\R_{\omega_2}(x_2) \quad \text{ if } \tau_{x_1} > \tau_{x_2} \\
    (\hat{\phi}_2)^\R_{\omega_2}(x_2) (\hat{\phi}_1)^\R_{\omega_1}(x_1) \quad \text{ if } \tau_{x_2} > \tau_{x_1}
\end{cases} \quad .
\end{equation}

From there, time-ordered correlation functions can be defined as follows.

\begin{definition}
    Let $\R$ be a relativistic QRF, $\mathcal{O}_\S \subseteq \bhs$ be strongly $(\operationalstate,0)$-microcausal, $\Omega \in \vacstate$, $n \in \mathbb{N}$, $\omega_1,\cdots,\omega_n \in \operationalstate$, $x_1, \cdots,x_n \in \mink$. The \emph{time-ordered vacuum correlation functions} are maps
    \begin{equation}\begin{aligned}
    \Delta^{(\Omega,\R)}_{n}[\omega_1,\cdots,\omega_n;x_1,\cdots,x_n] : \mathcal{O}_\S^n &\to \mathbb{C} \\
    \Delta^{(\Omega,\R)}_{n}[\omega_1,\cdots,\omega_n;x_1,\cdots,x_n](\phi_1,\cdots,\phi_n) &= \Tr\left[\Omega \T^\R_n[\omega_1,\cdots,\omega_n;x_1,\cdots,x_n](\phi_1,\cdots,\phi_n)\right] \, .
    \end{aligned}\end{equation}
\end{definition}

\begin{lemma}
    \label{lem:time-ordered}
    Let $\R$ be a relativistic QRF, $\mathcal{O}_\S \subseteq \bhs$ be strongly $(\operationalstate,0)$-microcausal, $\Omega \in \vacstate$. Then for all $\omega_1,\cdots,\omega_n \in \operationalstate$ and $\phi_1,\cdots,\phi_n \in \mathcal{O}_\S$, $n \in \mathbb{N}$,
    \begin{enumerate}
        \item For $\mu^{\F_\R}_{\omega_i}$-a.e. $x_i$, $i=1,\cdots,n$, \begin{multline}
            \Delta^{(\Omega,\R)}_{n}[\omega_1,\cdots,\omega_n;x_1,\cdots,x_n](\phi_1,\cdots,\phi_n) \\ =\sum_{\sigma \in S_n} \left(\prod_{i=1}^{n-1} \Theta(\tau_{\sigma(i)} - \tau_{\sigma(i+1)})\right)  W^{(\Omega,\R)}_n[\omega_{\sigma(1)},\cdots,\omega_{\sigma(n)};x_{\sigma(1)},\cdots,x_{\sigma(n)}](\phi_{\sigma(1)},\cdots,\phi_{\sigma(n)}) \, . \label{eqn:time-ordered in terms of Wightman}
        \end{multline}
        \item For $\mu^{\F_\R}_{\omega_i}$-a.e. $x_i$, $i=1,\cdots,n$ and all $(a,\Lambda) \in \Poincup$,
        \begin{equation}
            \Delta_n^{(\Omega,\R)}[\omega_1 \cdot (a,\Lambda),\cdots, \omega_n \cdot (a,\Lambda); x_1,\cdots,x_n] = \Delta_n^{(\Omega,\R)}[\omega_1,\cdots,\omega_n;\Lambda x_1 + a,\cdots,\Lambda x_n + a] \, .
        \end{equation}
        \item For $\mu^{\F_\R}_{\omega_i}$-a.e. $x_i$, $i=1,2$, and any coordinate chart in which $\xi = x_1 - x_2=(\xi^0,\vec{\xi})$,
        \begin{equation}\label{eqn:Spectral rep time-ordered}
            \Delta^{(\Omega,\R)}_2[\omega_1,\omega_2;x_1,x_2](\phi_1,\phi_2) = \Theta(\xi^0) W_2^{(\Omega,\R)}[\omega_1,\omega_2;x_1,x_2](\phi_1,\phi_2) + \Theta(-\xi^0) \Wbf_2^{(\Omega,\R)}[\omega_2,\omega_1;x_2,x_1](\phi_2,\phi_1) \, .
        \end{equation}
    \end{enumerate}
\end{lemma}

\begin{proof}
    See App. \ref{proof 5.11}.
\end{proof}

Note that the fact that the $W^{(\Omega,\R)}_n$ and $\Delta^{(\Omega,\R)}_{n}$ are \emph{not} Lorentz-invariant because of the shift of the states $\omega_i \stackrel{\Lambda}{\mapsto} \omega_i \cdot \Lambda^{-1}$ provides a roadblock for other results, notably the \enquote{standard} K{\"a}ll{\'e}n-Lehmann spectral representation of these time-ordered correlation functions\footnote{These could have been retrieved from Eqn. \eqref{eqn:Spectral rep time-ordered} assuming the spectrum condition holds, had the time-ordered correlation functions been Lorentz-invariant.}. It is plausible that some of the \enquote{usual} results of Wightman QFT which rely on these pointwise symmetries arise in some reasonable approximate regime, with subleading corrections which stem from the un-sharpness of the QRFs. It may also be that such deviations are key in avoiding other no-go theorems in collision theory, most notably Haag's theorem. A further study of this important nuance is thus necessary to understand how to bridge the relational formalism with \enquote{physicists'} collision theory and other results of particle physics, both perturbative and non-perturbative.

\section{Relational vs Wightman QFT}
\label{sec:Wightman comparison}

Wightman QFT (e.g. \cite{streater_pct_1989,wightman_fields_1965,wightman_theorie_1964}), also called Constructive QFT, formalises the way physicists traditionally work with quantum fields in a rigorous, axiomatised and analytic fashion. It relies on the assumption that quantum fields should fundamentally not be understood as pointwise objects, but rather that these are probed by certain test functions which vanish very quickly at infinity---Schwartz functions, living in the Schwartz space $\mathscr{S}(\mathbb{R}^4,\mathbb{C})$. Let us review the very basics of this approach before comparing and contrasting some aspects of it with the relational setup investigated in this work. For technical details of the Wightman distributional setup the reader is referred to App. \ref{App: distributions}.\\

\subsection{Wightman's axioms}
\label{wightman axioms}

Wightman's theory in $d=4$-dimensions is based on the following axioms \cite{streater_pct_1989}:
\begin{enumerate}
    \item[W0]  Assumptions of Relativistic Quantum Theory: The states of the theory are described by unit rays in a separable Hilbert space $\his$. The relativistic transformation law of the states is given by an ultraweakly continuous unitary representation $U$ of $ISL(2,\mathbb{C})$. Since $U(a,e)$ is unitary and ultraweakly continuous it can be written as $U(a,e) = e^{i P^\mu a_\mu}$, where $P^\mu$ is an unbounded self-adjoint operator interpreted as the energy-momentum operator of the theory. The operator $P_\mu P^\mu = m^2 \mathbb{1}_{\bhs}$ is interpreted as the square of the mass. The eigenvalues of $P^\mu$ lie in or on the plus cone\footnote{That is, in the future causal cone of the origin of the dual of Minkowski spacetime.} (\emph{spectrum condition}). There is a unique (up to a scalar) $ISL(2,\mathbb{C})$-invariant vector $\ket{\Omega} = U(a,\Lambda) \ket{\Omega}$.\footnote{The existence of the vacuum is sometimes weakened to the existence of a translation invariant vector, and the uniqueness can also be dropped. See \cite{streater_outline_1975} and references therein.}
    \item[W1] Assumptions about the Domain and Continuity of the Field: For each test function $f \in \mathscr{S}(\mathbb{R}^4,\mathbb{C})$, there exists a set $\{\hat{\Phi}^{(W)}_{1}(f),\cdots,\hat{\Phi}^{(W)}_{n}(f)\}$ of operators (for a spinorial field $ (\hat{\Phi}^{(W)}_{1},\cdots,\hat{\Phi}^{(W)}_{n})^T$ which transforms in an $n$-dimensional representation of $SL(2,\mathbb{C})$). These operators, together with their adjoints, are defined on a domain $\D_S$ of vectors, dense in $\his$. $\D_\S$ is a linear set containing $\ket{\Omega}$, with
    \begin{equation}
        U(a,\Lambda) \D_\S \subset \D_\S \, , \quad \hat{\Phi}^{(W)}_j(f) \D_\S \subset \D_\S \, , \quad \hat{\Phi}^{(W)}_j(f)^\dagger \D_\S \subset \D_\S
    \end{equation}
    for all $f \in \mathscr{S}(\mathbb{R}^4,\mathbb{C})$, where $j=1,\cdots,n$. If $\ket{\psi},\ket{\chi} \in \D_\S$ then $\braket{\chi}{\hat{\Phi}^{(W)}_j(f)\psi}$ is a tempered distribution regarded as a functional on $\mathscr{S}(\mathbb{R}^4,\mathbb{C})$. The fields form an irreducible set of operators in $\his$.\footnote{Formally, this means that if $B \in \bhs$ is any bounded operator satisfying $\braket{\chi}{B \hat{\Phi}^{(W)}_j(f) \eta} = \braket{\hat{\Phi}^{(W)}_j(f)^\dagger \chi}{B \eta}$ for all $\chi,\eta \in \D_\S$, all $j$ and all $f \in \mathscr{S}(\mathbb{R}^4,\mathbb{C})$, then $B \propto \mathbb{1}_{\bhs}$.} 
    \item[W2] Transformation Law of the Field: The equation
    \begin{equation}
        U(a,\Lambda) \hat{\Phi}^{(W)}_j(f) U(a,\Lambda)^\dagger = \sum_j S_{jk}[\Lambda^{-1}] \hat{\Phi}^{(W)}_k((a,\Lambda) \cdot f)
    \end{equation}
    is valid for all $f \in \mathscr{S}(\mathbb{R}^4,\mathbb{C})$ when each side is applied to any vector in $\D_\S$, where $(a,\Lambda) \cdot f(x) = f(\Lambda^{-1}(x-a))$.
    \item[W3] Local Commutativity\footnote{This form of local commutativity, which is very much analogous to the statement of Einstein causality, is sometimes replaced by the stronger requirement of microcausality (in particular when $\supp f = \supp g = \mink$), which states that the fields commute pointwise in the sense of distributions.}: If $\supp f \indep \supp g$ then either
    \begin{equation}
        \comm{\hat{\Phi}^{(W)}_j(f)}{\hat{\Phi}^{(W)}_k(g)} = 0 \quad \text{ or } \quad \acomm{\hat{\Phi}^{(W)}_j(f)}{\hat{\Phi}^{(W)}_k(g)} = 0
    \end{equation}
    holds for all $j,k$ (on $\D_\S$), and likewise
    \begin{equation}
        \comm{\hat{\Phi}^{(W)}_j(f)^\dagger}{\hat{\phi}_k(g)} = 0 \quad \text{ or } \quad \acomm{\hat{\Phi}^{(W)}_j(f)^\dagger}{\hat{\Phi}^{(W)}_k(g)} = 0 \,.
    \end{equation}
\end{enumerate}

The last Wightman axiom (below) is only relevant in the context of collision theory which we do not cover in this paper so we won't refer to it in what follows. Nevertheless, we state it here for completeness.

\begin{enumerate}
    \item[W4] Asymptotic Completeness: $\his = \his^{in} = \his^{out}$, where the introduced Hilbert spaces are thought of as the state spaces of incoming and outgoing quantum fields.
\end{enumerate}

We will now discuss in some detail how the axioms above, one by one, relate to the Relational QFT setup introduced in this work. The comparison is summarised in Table \ref{table:comparison} at the end of this section. 

\subsection{Relativistic Quantum Theory (W0)}
We make similar assumptions to axiom W0. In this paper, since we only cover scalar fields, it is the orthochronous Poincar{\'e} group $\Poincup$ that is unitarily represented on all the systems; this is to be adapted in future work. Under the interpretation of this action in terms of the energy momentum operator we may also assume the spectrum condition, though it does not seem necessary at this point. The existence (but not the uniqueness) of a Poincaré-invariant state (not necessarily pure) in $\his$ is needed for the theory of vacuum expectation values and can safely be assumed.

\subsection{Quantum fields (W1)}

A scalar quantum field in the setup of Wightman is represented by an operator-valued distribution, i.e., a linear map
    \begin{equation}
    \hat{\Phi}^{(W)} : \mathscr{S}(\mathbb{R}^4,\mathbb{C}) \to \ldhs \,,
    \end{equation}
where $\D_S$ be a dense subset of $\his$. It is \enquote{morally} understood as smearing a quantum field with a test function, i.e.,
	\begin{equation}
		``\hat{\Phi}^{(W)}(f) = \int_\mink  \hat{\phi}(x) f(x) \, d^4x "\,,
	\end{equation}
where $d^4x$ is the Lebesgue measure on $\mink$ and $\hat{\phi}: \mink \to \ldhs$ would be an operator-valued function; if it exists, it is called the kernel of $\hat{\Phi}^{(W)}$.\\

\paragraph*{Wightmanian relational quantum fields} We can make direct contact with this formalism by the means of the following definition.

\begin{definition}
    Let $(\R,\omega)$ be an oriented relativistic QRF. We call $f^{\F_\R}_{\omega} \equiv \frac{d\mu^{\F_\R}_{\omega}}{d^4x}$, if it exists\footnote{The assumption of the existence of this Radon-Nikodym derivative is mild, and several explicit examples of POVMs (e.g. coherent state POVMs) satisfy it.}, the \emph{frame smearing function}. If the frame smearing function is Schwartz, i.e., $f^{\F_\R}_{\omega} \in \mathscr{S}(\mathbb{R}^4,\mathbb{C})$, we say that $(\R,\omega)$ is a \emph{Schwartz oriented relativistic QRF}. 
\end{definition}

As we have seen, the \enquote{physicists'} intuition of smearing a quantum field with a (possibly Schwartz) test function is embodied  in RQFT very literally in the sense that we have
\begin{equation}
    \hat{\Phi}^\R(\omega) = \int_\mink \hat{\phi}^\R_\omega(x)f^{\F_\R}_\omega(x)d^4x.
\end{equation}
Let us emphasize that while Wightman quantum fields give rise to unbounded operators, our relational local quantum fields $\hat{\phi}^\R_\omega(x)$ and relational local observables $\hat{\Phi}^\R(\omega)$ are always bounded. Moreover, they can \emph{always} be written explicitly through an integral kernel---namely the relational local quantum field---which not only significantly eases the mathematical weight of the theory, but may also help avoid certain roadblocks faced by the Wightman's theory. \\

Further, the frame smearing functions arising in RQFT---as opposed to the Schwartz functions in the Wightman setup---have a direct interpretation: they describe the uncertainty of the spacetime localization of the quantum reference frame. The role of test functions in the Wightman approach is then taken by the \emph{frames' states} in RQFT. However, the smearing functions arising from states are necessarily real, and integrable to one if we work with normalized frame observables. To push the analogy further (at the cost of interpretational clarity) we can extend the definition of a relational quantum field to general trace class operators. We can still see such relational quantum fields as associating relational local observables $\yen^\R_\chi(\phi)$ to every choice of trace-class operator $\chi$. The $\yen^\R_\chi$ construction has the following properties.

\begin{lemma}
    \label{lem: yenRchi properties}
    Let $\R$ be a QRF and $\chi \in \thr$. Then
    \begin{enumerate}
        \item (Linearity in $\thr$) $\forall a, b \in \mathbb{C}$ and all $\chi_1,\chi_2 \in \thr$, $\yen^\R_{a\chi_1 + b \chi_2} = a \yen^\R_{\chi_1} + b \yen^\R_{\chi_2}$,
        \item (Linearity in $\bhs$) $\forall a, b \in \mathbb{C}$ and all $\phi_1,\phi_2 \in \bhs$, $\yen^\R_{\chi}(a\phi_1+b\phi_2) = a \yen^\R_\chi(\phi_1) + b \yen^\R_\chi(\phi_2)$,
        \item $\yen^\R_\chi(\mathbb{1}_{\bhs}) = \Tr[\chi] \mathbb{1}_{\bhs}$ i.e. is unital iff $\Tr[\chi]=1$,
        \item For all $\phi \in \bhs$, $\yen^\R_{\chi}(\phi)^\dagger = \yen^\R_{\chi^\dagger}(\phi^\dagger)$ i.e. $\yen^\R_{\chi}$ is adjoint-preserving if $\chi$ is self-adjoint, 
        \item $\yen^\R_{\chi}$ is completely positive if $\chi \geq 0$,
        \item For all $\phi \in \ths$, $\Tr[\yen^\R_{\chi}(\phi)] = \Tr[\phi] \cdot \Tr[\chi]$ i.e. $\yen^\R_{\chi}$ is trace-preserving iff $\Tr[\chi]=1$,
        \item $\yen^\R_\chi$ is normal,
        \item $\yen^\R_{\chi}$ is effect-preserving if $\chi \geq 0$ and $\Tr[\chi] \leq 1$,
        \item $\norm{\yen^\R_{\chi}} \leq \norm{\chi}_1$ (trace norm) with equality if $\chi \geq 0$, i.e. $\yen^\R_\chi$ is bounded (thus continuous),
        \item $\yen^\R_{\chi}$ is a contraction if $\norm{\chi}_1 \leq 1$. If $\chi \geq 0$, then $\yen^\R_{\chi}$ is a contraction iff $\Tr[\chi] \leq 1$.
    \end{enumerate}
\end{lemma}

\begin{proof}
    See App. \ref{proof lem yenRchi properties}.
\end{proof}

We can then extend the domain of definition of a relational quantum field as follows
\begin{equation}
    \hat{\Phi}^\R: \T(\hir) \ni \chi \mapsto \yen^\R_\chi(\phi) := \int_F \hat{\phi}_\lambda (x) \, d\mu^{\E_\R}_\chi(x,\lambda) \in \bhs \,,
\end{equation}
where $\mu^{\E_\R}_\chi: \W \mapsto \Tr[\chi\,\E_\R(\W)]$ is now a (finite) \emph{complex-valued} measure.\footnote{Note that this makes relational quantum fields be linear maps, which is a desirable mathematical property especially to discuss derivatives.} From there, we recover properties which are very much analogous to those of complex scalar Wightman quantum fields: Lem. \ref{lem: yenRchi properties} and Thm. \ref{thm:full covariance} imply the following.

\begin{corollary}
    \label{cor:properties RQF}
    Let $\R$ be a relativistic QRF and $\phi \in \bhs$. Then
    \begin{enumerate}
        \item (Linearity in $\thr$) $\forall a, b \in \mathbb{C}$ and all $\chi_1,\chi_2 \in \thr$, $\hat{\Phi}^\R(a \chi_1 + b \chi_2) = a \hat{\Phi}^\R(\chi_1) + b \hat{\Phi}^\R(\chi_2)$,
        \item (Linearity in $\bhs$) $\forall a, b \in \mathbb{C}$ and all $\phi_1,\phi_2 \in \bhs$, $\widehat{\Psi}^\R = a \hat{\Phi}_1 + b \hat{\Phi}_2$ where $\psi = a\phi_1 + b \phi_2$,
        \item $\forall \chi \in \thr$, $\hat{\Phi}^\R(\chi)^\dagger = \widehat{\Phi^\dagger}(\chi^\dagger)$, and if $\phi \in \bhs^{s.a}$, $\hat{\Phi}^\R(\chi)^\dagger = \hat{\Phi}^\R(\chi^\dagger)$,
        \item $\forall (a,\Lambda) \in \Poincup$, $(a,\Lambda) \cdot \hat{\Phi}^\R(\chi) = \hat{\Phi}^\R((a,\Lambda)\cdot \chi)$, 
        \item $\norm{\hat{\Phi}^\R} \equiv \sup_{\norm{\chi}_1 \leq 1} \norm{\hat{\Phi}^\R(\chi)} \leq \norm{\phi}$.
    \end{enumerate}
\end{corollary}

Restricting to the trace-class operators that give rise to Schwartz functions we get a picture with very strong structural analogies with the Wightman setup. For such fields, we have access to some tools from distribution theory: given a differential operator $D: \mathscr{S}(\mathbb{R}^4,\mathbb{C}) \to \mathscr{S}(\mathbb{R}^4,\mathbb{C})$ we could define 
\begin{equation}
    D(\hat{\Phi}^\R)(\chi) := -\int_\mink \hat{\phi}^\R_\chi(x) D[f^{\F_\R}_\chi](x)d^4x
\end{equation}
whenever $D$ is such that the right hand side converges for all $\chi$.\footnote{Perhaps, since the relational local quantum fields are bounded (by $\norm{\phi}\norm{\chi}$) operator-valued functions, demanding the codomain of $D$ to be contained in $\mathscr{S}(\mathbb{R}^4,\mathbb{C}) \cap L^1(\mathbb{R}^4,d^4x)$ would be enough. This is to be explored in detail in future work---here we only aim to point to how the relational notion of a quantum field developed here relates to the distributional definition of Wightman.}

\paragraph*{Irreducibility.} 

The irreducibility of quantum fields has to do with the fact that, as Streater and Wightman put it \cite{streater_pct_1989}, \enquote{every operator is a function of the field operators}. In Wightman QFT, it is given as a commutativity constraint over individual vectors in the dense domain of the fields---this is due to their unboundedness. An analogous definition of irreducibility for \emph{bounded} relational quantum fields is given as follows.

\begin{definition}\label{def:irreducible}
    Let $\R$ be a relativistic QRF, $\operationalstate \subseteq \framestate$ be convex and $\mathcal{O}_\S \subseteq \bhs$ contain the identity and be closed under adjoints. Then $\phi \in \bhs$ is said to be $(\mathcal{O}_\S,\operationalstate)$-irreducible if $\forall A \in \mathcal{O}_\S$, \begin{equation}
        \comm{A}{\hat{\Phi}^\R(\omega)} = 0 \quad \forall \omega \in \operationalstate \Rightarrow A \propto \mathbb{1}_{\bhs} \, .
    \end{equation}
\end{definition}

Equivalent definitions which highlight the nature of irreducible fields are given as follows. First, we write
\begin{equation}
    B^{\operationalstate}_\phi := \{\hat{\Phi}^\R(\omega) \mid \omega \in \operationalstate\} \subset \bhs
\end{equation}
as the set of operators in $\bhs$ which can be \enquote{reached} from different oriented relativisations of $\phi$ in operationally meaningful preparations of the frame belonging to $\operationalstate$. The algebraic closure of this set will be denoted by
\begin{equation}
    \mathcal{A}^{\operationalstate}_\phi := \{\hat{\Phi}^\R(\omega) \mid \omega \in \operationalstate\}'' \subset \bhs \, .
\end{equation}

\begin{proposition}
    Let $\R$ be a relativistic QRF, $\operationalstate \subseteq \framestate$ be convex and $\mathcal{O}_\S \subseteq \bhs$ contain the identity and be closed under adjoints, and $\phi \in \bhs$. Then the following are equivalent:
    \begin{enumerate}
        \item $\phi$ is $(\mathcal{O}_\S,\operationalstate)$-irreducible,
        \item $\left(B^{\operationalstate}_\phi\right)' = \mathbb{C} \cdot \mathbb{1}_{\mathcal{O}_\S}$, where the commutant $'$ is taken in $\mathcal{O}_\S$,
        \item $\left(B^{\operationalstate}_\phi\right)'' = \mathcal{O}_\S$, where the double commutant $''$ is taken in $\mathcal{O}_\S$.
    \end{enumerate}
    Moreover, $\phi$ is $(B(\his),\operationalstate)$-irreducible iff $\mathcal{A}^{\operationalstate}_\phi=B(\his)$.
\end{proposition}

\begin{proof}
    $1. \Leftrightarrow 2.$ follows by the definition of the commutant, while $2. \Leftrightarrow 3.$ is immediate. The last statement follows from $1. \Leftrightarrow 3.$ and the von Neumann's bicommutant theorem \cite{takesaki_theory_2001}.
\end{proof}

The $(\mathcal{O}_\S,\operationalstate)$-irreducibility of an operator $\phi \in \bhs$ is dependent both on $\phi$ (and the unitary representation of $\Poincup$ on $\his$, and on the choice of $\mathcal{O}_\S$) and on $\R$ (and the choice of $\operationalstate$). Indeed, if $\R$ is \enquote{informational} for $\operationalstate$ in the sense that, for appropriate (non-trivial) operators $\phi \in \bhs$, it can fully describe $\mathcal{O}_\S$ through different preparations of its frame observable in $\operationalstate$, then it yields irreducibility. It also clearly depends on $\phi \in \bhs$ and on $\mathcal{O}_\S$: for example, if $\phi \propto \mathbb{1}_{\mathcal{O}_\S}$ then $\phi$ cannot be $(\mathcal{O}_\S,\operationalstate)$-irreducible for any QRF $\R$ and nontrivial choice of $\operationalstate$ and $\mathcal{O}_\S$ (unless $\dim(\his)=1$).

The $(\mathcal{O}_\S,\operationalstate)$-irreducibility of an operator $\phi$ can in fact be shown from another perspective: that of the cyclicity of the vacuum for the polynomial algebra generated by the relational quantum fields. Indeed, in Wightman QFT, the irreducibility of the quantum fields follows from the cyclicity of the pure vacuum for the polynomial algebra of the smeared fields assuming the spectrum condition holds \cite{streater_pct_1989}. We now prove that the same reasoning holds in RQFT.

\begin{theorem}\label{thm: R-ireducibility}
    Let $\R$ be a relativistic QRF, $\operationalstate \subseteq \framestate$ be convex and closed under $T(1,d-1)$, and $\mathcal{O}_\S \subseteq \bhs$ contain the identity and be closed under adjoints, and $\phi \in \bhs$. Suppose
    \begin{enumerate}
        \item There exists a unique (up to scalar multiples) translation-invariant vector $\ket{\Omega} \in \his$,
        \item $\ket{\Omega}$ is $(\mathcal{O}_\S,\operationalstate)$-cyclic,
        \item $U_\S$ satisfies the spectrum condition.
    \end{enumerate}
    Then $\phi$ is $(\mathcal{O}_\S,\operationalstate)$-irreducible.
\end{theorem}

\begin{proof}
    See App. \ref{proof 6.4}
\end{proof}

This theorem provides some additional justification for the definition of a field theory (in the Wightmanian sense) in terms of the cyclicity of the vacuum rather than in terms of the $(\mathcal{O}_\S,\operationalstate)$-irreducibility of the relational quantum fields. It also highlights that this $(\mathcal{O}_\S,\operationalstate)$-irreducibility property is really to be understood as the ability to fully describe the whole (absolute) state space of the system from the description of the observable relative to different preparations of the QRF alone. From an operational perspective, an observer who has access to a single preparation of the QRF need not necessarily care about such a property, though it does allow one to characterise which observables of $\S$ carry \enquote{all the information about $\S$} when seen through the lens of quantum rods and clocks whose orientations and localisations can be made to vary arbitrarily. In particular, it may be interesting to \enquote{reconstruct} the whole absolute state space of $\S$ from a single observable if different observers (which each carry a different preparation $\omega$ of the QRF $\R$) communicate with one another. It may also shed some light on how \enquote{(im)precise} some quantum rods and clocks can be, especially if a given choice of QRF $\R$, of $\operationalstate$ and of $\mathcal{O}_\S$, is such that there does not exist any $(\mathcal{O}_\S,\operationalstate)$-irreducible operator in~$\S$.

\subsection{Covariance (W2)}

Poincaré covariance is recovered at a similar level in both Wightman and relational QFT, which becomes explicit when assuming that the oriented QRF admit frame smearing functions. Indeed, a scalar quantum Wightman field $\hat{\Phi}^{(W)}$ that admits a kernel $\hat{\phi}: \mink \to \ldhs$ transform as (in our notation)
\begin{equation}\label{Wightman scalar covariance}
    (a,\Lambda) \cdot \hat{\Phi}^{(W)}(f) =  \hat{\Phi}^{(W)}((a,\Lambda) \cdot f) ``=" \int_\mink  \hat{\phi}(x)f(\Lambda^{-1}(x - a)) \, d^4x\,,
\end{equation}
while our relational local observables satisfy the following \eqref{thm: scalar covariance}
\begin{equation}
    (a,\Lambda) \cdot \hat{\Phi}^\R(\chi) = 
    \hat{\Phi}^\R((a,\Lambda) \cdot \chi) = \int_{\mink} \hat{\phi}^\R_{(a,\Lambda) \cdot \chi}(\Lambda x + a) \,d\mu^{\F_\R}_{\chi}(x)\, .
\end{equation}
Changing variables and due to the Poincar{\'e}-invariance of the Lebesgue measure, we get 
\begin{equation}
    (a,\Lambda) \cdot \hat{\Phi}^\R(\chi) = 
     \int_{\mink} \hat{\phi}^\R_{(a,\Lambda) \cdot \chi}(x) f^{\F_\R}_\chi(\Lambda^{-1}(x - a)) \, d^4x\,,
\end{equation}
which largely resembles \eqref{Wightman scalar covariance} with the sole difference that the integral kernel in the relational framework is sensitive to the Poincar{\'e} transformations (plus the fact that the smearing functions may not be Schwartz). This discrepancy can be weakened by considering the following class of oriented frames.

\begin{definition}
    An oriented relativistic QRF that is both Schwartz and globally oriented will be called \emph{Wightmanian}.
\end{definition}

Relational local observables relative to Wightmanian QRFs can be seen as translation-covariant operator-valued functions smeared with Schwartz test functions, and they transform almost exactly like the Wightman fields with kernels would, i.e., we have
\begin{equation}
    (a,\Lambda) \cdot \hat{\Phi}^\R(\chi) =  
     \int_{\mink} \hat{\phi}^\R_{\Lambda \cdot \chi}(x) f^{\F_\R}_\chi(\Lambda^{-1}(x - a)) \, d^4x\,.
\end{equation}
We see that the relational local quantum fields are now only sensitive to the Lorentz transformations. Thus, when only translations are considered, the transformation properties of fields described with respect to Wightmanian QRFs are exactly the same as in the Wightman's theory.

\subsection{Causality (W3)}

Local commutativity is recovered when one considers relativistic QRFs and $(\operationalstate,0)$-causal fields. Indeed, Wightman scalar fields satisfy
\begin{equation}
    \comm{\hat{\Phi}^{(W)}(f_1)}{\hat{\Phi}^{(W)}(f_2)}= \comm{\hat{\Phi}^{(W)}(f_1)^\dagger}{\hat{\Phi}^{(W)}(f_2)}=0 \; \text{ whenever } \; f_1 \indep f_2 \,,
\end{equation}
while for $(\operationalstate,0)$-causal $\phi \in \bhs$ and $\omega_1,\omega_2 \in \operationalstate$ we have
\begin{equation}
    \comm{\hat{\Phi}^\R(\omega_1)}{\hat{\Phi}^\R(\omega_2)} =
        \comm{\hat{\Phi}^\R(\omega_1)^\dagger}{\hat{\Phi}^\R(\omega_2)} = 0 \; \text{ whenever } \; \omega_1 \indepER \omega_2 \,.
\end{equation}

Whenever the frame smearing functions for $\omega_1$ and $\omega_2$ exist, $\omega_1 \indepER \omega_2$ is exactly equivalent to $f^{\F_\R}_{\omega_1} \indep f^{\F_\R}_{\omega_2}$. Similarly, our (stronger) $(\operationalstate,0)$-microcausality condition corresponds to microcausality sometimes assumed for scalar Wightman fields with kernels. Indeed, this condition reads
\begin{equation}
    \comm{\hat{\phi}(x_1)}{\hat{\phi}(x_2)}= \comm{\hat{\phi}(x_1)^\dagger}{\hat{\phi}(x_2)}=0 \; \text{ whenever } \; x_1 \indep x_2 \,,
\end{equation}
while for strongly $(\operationalstate,0)$-microcausal $\phi \in \bhs$ we have\footnote{Recall here that globally oriented QRFs (and thus Wightmanian QRFs) cannot be strongly microcausal.}
     \begin{equation}
     \comm{\hat{\phi}^\R_{\omega_1}(x_1)}{\hat{\phi}^\R_{\omega_2}(x_2)} =
        \comm{\hat{\phi}^\R_{\omega_1}(x_1)^\dagger}{\hat{\phi}^\R_{\omega_2}(x_2)} = 0 \text{ whenever } \; x_1 \indep x_2\,, \; x_i \in \supp \mu^{\F_\R}_{\omega_i}\,.
    \end{equation}

Furthermore, the stronger assumption of spacelike commutativity between different fields in Wightman QFT---which in RQFT would take form of $\mathcal{O}_\S$ being $(\operationalstate,0)$-(micro)causal (see \eqref{eqn:Einstein causality}, \eqref{eqn:strong microcausality}, \eqref{eqn:weak microcausality})---is related to the theory of Klein transformations, superselection rules and even-odd rules \cite{streater_pct_1989}. In a QFT with only bosons, the assumption that $\S$ (or some appropriate class of system operators) is $(\operationalstate,\sigma)$-causal is appropriate; however, the consideration of spinors complicates this discussion, and should be the content of future work.

The comparison is summarised in the table below.

\begin{table}[H]
\begin{center}
\tabulinesep=1.2mm
\begin{tabu}{||c| c | c ||} 
 \hline
 Objects & Relational QFT & Wightman QFT \\ [0.5ex] 
 
 \hline\hline
 Smeared field & Relational quantum field  & Wightman quantum field \\ 
 \hline
 Definition & $\hat{\Phi}^\R : \framestate \to \bhs$ & $\hat{\Phi}^{(W)} : \mathscr{S}(\mathbb{R}^4,\mathbb{C}) \to \mathcal{L}(D_\S,\his)$  \\
  & $\hat{\Phi}^\R(\omega) = \int_{\mink} \hat{\phi}^\R_{\omega}(x) \, d\mu^{\F_\R}_{\omega}(x)$ & $\hat{\Phi}^{(W)}(f)$ \\
  & $ = \int_{\mink} \hat{\phi}^\R_{\omega}(x) f^{\F_\R}_{\omega}(x) d^4 x$ & $\exists$ kernel $\Rightarrow \hat{\Phi}^{(W)}(f) = \int_\mink \hat{\phi}(x) f(x) d^4x$ \\
  Domain & $\his$ (bounded) & $D_\S$ a dense subset of $\his$ (unbounded)\\
  Causality & $\omega_1 \indepER \omega_2 \Rightarrow \comm{\hat{\Phi}^\R(\omega_1)}{\hat{\Phi}^\R(\omega_2)} = 0$ & $f_1 \indep f_2 \Rightarrow \comm{\hat{\Phi}^{(W)}(f_1)}{\hat{\Phi}^{(W)}(f_2)} = 0$
  \\ Covariance & $(a,\Lambda)\cdot \hat{\Phi}^\R(\omega) = \hat{\Phi}^\R((a,\Lambda) \cdot \omega)$ & $(a,\Lambda) \cdot \hat{\Phi}^{(W)}(f) = \hat{\Phi}^{(W)}((a,\Lambda) \cdot f)$ \\ & $ = \int_\mink \hat{\phi}^\R_{(a,\Lambda) \cdot \omega}(x) f^{\F_\R}_{\omega}(\Lambda^{-1}(x-a)) d^4x$ & $\exists$ kernel $\Rightarrow = \int_\mink \hat{\phi}(x) f(\Lambda^{-1}(x-a)) d^4x$ \\
   [1ex] \hline 
   \end{tabu}
   \end{center}
\end{table}

\begin{table}[H]
\begin{center}
\tabulinesep=1.2mm
\begin{tabu}{||c| c | c ||} 
 \hline
 Pointwise field & Relational local quantum field  & Absolute quantum field \\ 
 \hline
  Definition & $\hat{\phi}^\R_{\omega} : \mink \to \bhs$ & $\hat{\phi} : \mink \to \mathcal{L}(D_\S,\his)$  \\
  Existence & always exists & sometimes exists as a kernel  \\
  Domain & $\his$ (bounded) & $D_\S$ a dense subset of $\his$ (unbounded)\\
  Microcausality & $x_1 \indep x_2 \Rightarrow \comm{\hat{\phi}^\R_{\omega_1}(x_1)}{\hat{\phi}^\R_{\omega_2}(x_2)} = 0$ & $x_1 \indep x_2 \Rightarrow\comm{\hat{\phi}(x_1)}{\hat{\phi}(x_2)} = 0$
  \\
  Poincaré cov. & $(a,\Lambda)\cdot \hat{\phi}^\R_{\omega}(x) = \hat{\phi}^\R_{(a,\Lambda) \cdot \omega}(\Lambda x+a)$ & $(a,\Lambda) \cdot \hat{\phi}(x) \text{``=”} \hat{\phi}(\Lambda x+a)$  \\ & Poincaré covariant with shifted $\omega$ & Poincaré covariant under the integral \\
  Translation cov. & $(\R,\omega)$ globally oriented $\Rightarrow a\cdot \hat{\phi}^\R_{\omega}(x) = \hat{\phi}^\R_{\omega}(x+a)$ & $a \cdot \hat{\phi}(x) \text{``=”} \hat{\phi}(x+a)$ \\
  [1ex] 
 \hline
\end{tabu}
\caption{Comparison between fields in Relational Quantum Field Theory and Wightman Quantum Field Theory.}
\label{table:comparison}
\end{center}
\end{table}

\section{Relational vs Algebraic QFT}

\label{sec:AQFT comparison}

Algebraic Quantum Field Theory (AQFT) is an alternative foundational framework to rigorously treat quantum fields. The main idea is that the causal structure of spacetime, encoded by the collection of regions with a partial order given by the possibility of subluminal communication, is being mapped to the algebraic structure of quantum theory. This is achieved by the means of associating local algebras to space-time regions in a way respecting the causal structure. Algebraic QFT, being  mathematically elegant and conceptually appealing, has been only partly successful at the task of formalising the physical theory of quantum fields. In particular, the attempts to treat interacting gauge theories in $4$-dimensions, necessary to model the experiments carried out in our biggest colliders like the LHC, remain incomplete. Initially proposed by Haag in the '$60$s \cite{haag_algebraic_1964}, AQFT is still very much an active research area with a variety of models differing in mathematical details.\footnote{See \cite{fewster_algebraic_2019} for an introduction, \cite{rejzner_algebraic_2016} for a recent exposition and \cite{brunetti_generally_2003,Fewster:2015kua} for the curved spacetimes context.} In the context of Minkowski spacetime, the core of the structure shared by the majority of the AQFT approaches is the following.

\begin{definition}
    An Algebraic QFT (AQFT) is an assignment
    \begin{equation}
       \A: {\rm Reg}(\mink) \ni \U \mapsto {\A(U)} \subset \bh,
    \end{equation}
    sometimes called a \emph{net of local algebras}, where ${\rm Reg}(\mink)$ is a distinguished family of spacetime subsets, $\A(\U)$ is closed under algebraic operations (and possibly also in a chosen topology) and $\hi$ carries a representation of the Poincar{\'e} group.\footnote{Spacetime regions can for example be assumed to be relatively compact, while a conservative choice of the class of subalgebras considered is to assume they are von Neumann factors, generically of type III \cite{buchholz_universal_1987,YNGVASON2005135}.} An AQFT is assumed to satisfy the following axioms:
    \begin{enumerate}
        \item (Isotony) For all $\U, \V \in {\rm Reg}(\mink)$ such that $\U \subset \V$, $\A(\U) \subset \A(\V)$.
        \item (Covariance) For all $(a,\Lambda) \in \Poincup$ and all $\U \in {\rm Reg}(\mink)$, $(a,\Lambda) \cdot \A(\U) = \A((a,\Lambda) \cdot \U)$.  
        \item (Causality) For all $\U,\V \in {\rm Reg}(\mink)$ such that $\U \indep \V$, we have $\comm{\A(\U)}{\A(\V)} = \varnothing$.
    \end{enumerate}
Additional properties often assumed are the following:
\begin{itemize}
    \item (Time-slice axiom) For all $\U, \V \in {\rm Reg}(\mink)$ such that $\U \subset \V$ and  $\U$ contains a \emph{Cauchy hypersurface}\footnote{A Cauchy hypersurface $\Sigma$ is a subset of a Lorentzian spacetime $(\mathcal{M},g)$ such that every inextendible causal curve in $\mathcal{M}$ intersects $\Sigma$ exactly once; such surfaces are suitable for specifying the initial data for the dynamical equations of a relativistic theory. We say that $\U$ contains a Cauchy hypersurface for $\V$ if there exists a hypersurface $\Sigma \subset \U$ such that $\Sigma$ is a Cauchy hypersurface for the spacetime $(\text{ch}(\V),\eta)$, where $\text{ch}(\V) := J^+(\V) \cap J^-(\V)$ is the \emph{causal hull} of $\V$ (which defines a spacetime in itself).} for
    $\V$, the corresponding algebras are isomorphic, i.e., $\A(\U) \cong \A(\V)$.
    \item (Statistical independence) For all $\U,\V \in {\rm Reg}(\mink)$ such that $\U \indep \V$ and any pair of states $\rho_1: \A(\U) \to \mathbb{C}$, $\rho_2: \A(\V) \to \mathbb{C}$ there exist a state\footnote{The algebra $\A(\U) \wedge\A(\V)$ is the one 
\enquote{generated by} $\A(\U)$ and $\A(\V)$, i.e., the smallest one (satisfying the properties required from local algebras) containing them both.} $\rho: \A(\U) \wedge\A(\V) \to \mathbb{C}$ such that for all $A \in \A(\U)$ and $B \in \A(\V)$
    \begin{equation}
        \rho(AB) = \rho_1(A)\rho_2(B).
    \end{equation}
    \item (Haag property) For any region $\U \in {\rm Reg}(\mink)$, the algebra associated to the causal complement\footnote{The causal complement of a region $\U$ is the biggest admissible (contained in ${\rm Reg}(\mink)$) set of spacetime points that are causally separated from all the points in $\U$.} $\U^\perp$ equals the commutant\footnote{The commutant of a set of operators $\mathcal{O} \subset B(\hi)$ is the set of all operators in $B(\hi)$ commuting with every element of $\mathcal{O}$.} of the algebra of $\U$, i.e., $\A(\U^\perp) = \A(\U)'$.
\end{itemize}
\end{definition}

The motivation for the three axioms: isotony, covariance and causality should be clear by now. We briefly discuss the other properties.

Operationally, an experimenter who has access to all possible measurements within a very short time interval, but over a sufficiently large region of space, should in principle be able to gather all the information there is to know about the system at later times if the time evolution of the system is known. This idea underpins the \emph{time-slice axiom} above (see e.g. \cite{Fewster:2015kua})---the algebra of any subset containing a Cauchy hypersurface for $\V$ should be isomorphic to $\A(\V)$. Indeed, if the algebra of $\A(\U)$ was strictly smaller than $\A(\V)$, it would imply that new degrees of freedom or new types of observables could spontaneously appear at later times without being determined by some initial conditions in $\U$. This would break the deterministic aspect of time evolution that one may expect the laws of physics to uphold away from quantum measurements.

Statistical independence (see e.g. \cite{fewster_split_2016}) can be seen as complementary to Einstein causality in assuring independence of state preparations in causally separated regions. Indeed, a natural notion of a local state in AQFT is that of a (continuous normalized linear) functional on the local algebra, understood as assigning expectation values to local observables. Statistical independence then assures that arbitrary pairs of preparations in spacelike separated regions can coexist as a single one on a bigger algebra.

Finally, the Haag property \cite{haag_local_1996} can be understood as assuring that the correspondence between the algebraic structure of subalgebras in $\bh$ and the causal structure of $\mink$ is tight in the following sense: all operators commuting with the local algebra are local to the algebra of the biggest causally separated region.

We will now discuss how natural definitions of relational local algebras arising in the context of RQFT are compatible with the Algebraic framework for QFT.

\subsection{Relational local algebras}

One immediate meaningful definition is that of algebras of local observables, defined as follows.

\begin{definition}
    Let $\R$ be a relativistic QRF, $\operationalstate \subseteq \framestate$ be convex, $\sigma \geq 0$ and $\mathcal{O}_\S \subseteq \bhs$ be $(\operationalstate,\sigma)$-causal. Given a subset $\U \subseteq \mink$, we call
    \begin{equation}
    \A^{(\operationalstate,\sigma)}_{\mathcal{O}_\S}(\U) := \{\hat{\Phi}^\R(\omega) \mid \phi \in \mathcal{O}_\S, \,  \omega \in \operationalstate \text{ s.t. } 
 \supp \mu^{\F_\R}_{\omega} \subset \U\}''
\end{equation}
    a \emph{relational local algebra}.
\end{definition}

Notice that this definition is suitable for any chosen class of regions ${\rm Reg}(\mink)$. An alternative possibly meaningful definition, similar to $\A^{(\operationalstate,\sigma)}_{\mathcal{O}_\S}$ but explicitly satisfying the time-slice axiom, can be given as follows.

\begin{definition}
    Let $\R$ be a relativistic QRF, $\operationalstate \subseteq \framestate$ be convex, $\sigma \geq 0$  and $\mathcal{O}_\S \subseteq \bhs$ be $(\operationalstate,\sigma)$-causal. We call
    \begin{equation}
    \Afrak^{(\operationalstate,\sigma)}_{\mathcal{O}_\S}(\U):= \{\hat{\Phi}^\R(\omega) \mid \phi \in \mathcal{O}_\S, \, \omega \in \operationalstate \text{ s.t. } \supp \mu^{\F_\R}_{\omega} \subset {\rm ch}(\U)\}'',
\end{equation}
    where ${\rm ch}(\U) = J^+(\U) \cap J^-(\U)$ is the \emph{causal hull} of $\U$, a \emph{deterministic relational local algebra}.
\end{definition}

These algebras represent the locally accessible relational local observables, and are thus very natural to consider. They satisfy the properly extended properties of local algebras in AQFT, coinciding with the usual ones in some cases.

\begin{theorem}\label{RAQFT1}
    Let $\R$ be a relativistic QRF, $\operationalstate \subseteq \framestate$ be convex, $\sigma \geq 0$  and $\mathcal{O}_\S \subseteq \bhs$ be $(\operationalstate,\sigma)$-causal. Then
    \begin{enumerate}
        \item (Isotony) For all $\U \subseteq \V \subseteq \mink$, $\A^{(\operationalstate,\sigma)}_{\mathcal{O}_\S}(\U) \subseteq \A^{(\operationalstate,\sigma)}_{\mathcal{O}_\S}(\V)$.
        \item (Covariance) For all $(a,\Lambda) \in \Poincup$ and $\U \subseteq \mink$, $(a,\Lambda) \cdot \A^{(\operationalstate,\sigma)}_{\mathcal{O}_\S}(\U) = \A_{\mathcal{O}_\S}^{((a,\Lambda) \cdot \operationalstate,\sigma)}((a,\Lambda) \cdot \U)$.  
        \item (Causality) For all $\U \indepsigma \V \subset \mink$, $\comm{\A^{(\operationalstate,\sigma)}_{\mathcal{O}_\S}(\U)}{\A^{(\operationalstate,\sigma)}_{\mathcal{O}_\S}(\V)} = \varnothing$.
    \end{enumerate}
    Likewise, for deterministic relational local algebras,
    \begin{enumerate}
        \item (Isotony) For all $\U \subset \V \in \mink$, $\Afrak^{(\operationalstate,\sigma)}_{\mathcal{O}_\S}(\U) \subset \Afrak^{(\operationalstate,\sigma)}_{\mathcal{O}_\S}(\V)$.
        \item (Covariance) For all $(a,\Lambda) \in \Poincup$, $(a,\Lambda) \cdot \Afrak^{(\operationalstate,\sigma)}_{\mathcal{O}_\S}(\U) = \Afrak^{((a,\Lambda) \cdot \operationalstate,\sigma)}_{\mathcal{O}_\S}((a,\Lambda) \cdot \U)$.
        \item (Causality) For all $\U \indepsigma \V \subset \mink$, $\comm{\Afrak^{(\operationalstate,\sigma)}_{\mathcal{O}_\S}(\U)}{\Afrak^{(\operationalstate,\sigma)}_{\mathcal{O}_\S}(\V)} = \varnothing$.
        \item (Time-slice) For all $\U \subset \V$ such that $\U$ contains a Cauchy hypersurface for $\V$, $\Afrak^{(\operationalstate,\sigma)}_{\mathcal{O}_\S}(\U) \cong \Afrak^{(\operationalstate,\sigma)}_{\mathcal{O}_\S}(\V)$.
    \end{enumerate}
\end{theorem}

\begin{proof}
    We start with relational local algebras.
    \begin{enumerate}
        \item Let $\C^{(\operationalstate,\sigma)}_{\mathcal{O}_\S}(\U) := \{\hat{\Phi}^\R(\omega) \mid \phi \in \mathcal{O}_\S, \,  \omega \in \operationalstate \text{ s.t. } 
 \supp \mu^{\F_\R}_{\omega} \subseteq \U\}$.
If $\U \subseteq \V$, we have $\supp \mu^{\F_\R}_{\omega} \subseteq \U$ $\Rightarrow$ $\supp \mu^{\F_\R}_{\omega} \subseteq \V$. Thus, $\C^{(\operationalstate,\sigma)}_{\mathcal{O}_\S}(\U) \subseteq \C^{(\operationalstate,\sigma)}_{\mathcal{O}_\S}(\V)$ and since $\A^{(\operationalstate,\sigma)}_{\mathcal{O}_\S}(\U) = \C^{(\operationalstate,\sigma)}_{\mathcal{O}_\S}(\U)''$ the result follows from the isotony of the double commutant.
        \item Due to the $\Poincup$-covariance of $\F_\R$, we have that for all $\omega \in \operationalstate$ and $(a,\Lambda) \in \Poincup$
        \begin{equation}
            \supp \mu^{\F_\R}_{\omega \cdot (a,\Lambda)} = (a,\Lambda)^{-1} \cdot (\supp \mu^{\F_\R}_{\omega})\,,
        \end{equation}
        and hence we get
        \begin{equation}\begin{aligned}
            (a,\Lambda) \cdot \C^{(\operationalstate,\sigma)}_{\mathcal{O}_\S}(\U) &= \{(a,\Lambda) \cdot \hat{\Phi}^\R(\omega) \mid \phi \in \mathcal{O}_\S, \omega \in \operationalstate \text{ s.t. } \supp \mu^{\F_\R}_{\omega} \subseteq \U \} \\
            &= \{\hat{\Phi}^\R(\omega \cdot (a,\Lambda)^{-1}) \mid \phi \in \mathcal{O}_\S, \omega \in \operationalstate \text{ s.t. } \supp \mu^{\F_\R}_{\omega} \subseteq \U \} \\
            &= \{\hat{\Phi}^\R(\tilde{\omega}) \mid \phi \in \mathcal{O}_\S, \tilde{\omega} \cdot (a,\Lambda) \in \operationalstate \text{ s.t. } \supp \mu^{\F_\R}_{\tilde{\omega} \cdot (a,\Lambda)} \subseteq \U \} \\
            &= \left\{\hat{\Phi}^\R(\tilde{\omega}) \mid \phi \in \mathcal{O}_\S, \tilde{\omega} \in (a,\Lambda) \cdot \operationalstate \text{ s.t. } (a,\Lambda)^{-1} \cdot \left(\supp \mu^{\F_\R}_{\tilde{\omega}}\right) \subseteq \U \right\} \\
            &= \{\hat{\Phi}^\R(\omega) \mid \phi \in \mathcal{O}_\S, \omega \in (a,\Lambda) \cdot \operationalstate \text{ s.t. } \supp \mu^{\F_\R}_{\omega} \subseteq (a,\Lambda) \cdot \U \} \\
            &= \C^{((a,\Lambda) \cdot \operationalstate,\sigma)}_{\mathcal{O}_\S}((a,\Lambda) \cdot \U) \, ,
        \end{aligned}\end{equation}
        where the last line holds by Lem. \ref{lem:causality and covariance}.
        Since unitary conjugation commutes with commutants, we get
        \begin{multline}
            (a,\Lambda) \cdot \A^{(\operationalstate,\sigma)}_{\mathcal{O}_\S}(\U) = (a,\Lambda) \cdot \C^{(\operationalstate,\sigma)}_{\mathcal{O}_\S}(\U)'' = ((a,\Lambda) \cdot \C^{(\operationalstate,\sigma)}_{\mathcal{O}_\S}(\U))'' \\ = (\C^{((a,\Lambda) \cdot \operationalstate,\sigma)}_{\mathcal{O}_\S}((a,\Lambda) \cdot \U))'' = \A^{((a,\Lambda)\cdot \operationalstate,\sigma)}_{\mathcal{O}_\S}((a,\Lambda) \cdot \U) \, .
        \end{multline}
        \item If $\U \indepsigma \V$ then for all $A \in \C^{(\operationalstate,\sigma)}_{\mathcal{O}_\S}(\U)$ and $B \in \C^{(\operationalstate,\sigma)}_{\mathcal{O}_\S}(\V)$, the commutant vanishes $\comm{A}{B} = 0$. In particular, we have $\C^{(\operationalstate,\sigma)}_{\mathcal{O}_\S}(\V) \subseteq \C^{(\operationalstate,\sigma)}_{\mathcal{O}_\S}(\U)'$. 
        Since for two subsets $\mathfrak{A}$ and $\mathfrak{B}$ of $\bh$, $\mathfrak{A} \subseteq \mathfrak{B} \Rightarrow \mathfrak{B}' \subseteq \mathfrak{A}'$, taking the commutant of both sides gives $\A^{(\operationalstate,\sigma)}_{\mathcal{O}_\S}(\U) \subseteq \C^{(\operationalstate,\sigma)}_{\mathcal{O}_\S}(\V)'$, and applying it again yields $\A^{(\operationalstate,\sigma)}_{\mathcal{O}_\S}(\V) \subseteq \A^{(\operationalstate,\sigma)}_{\mathcal{O}_\S}(\U)'$.
    \end{enumerate}
    For deterministic relational local algebras, since $\U \subseteq \V$ $\Rightarrow$ ${\rm ch}(\U) \subseteq {\rm ch}(\V)$ and $\U \indepsigma \V$ $\Rightarrow$ ${\rm ch}(\U) \indepsigma {\rm ch}(\V)$, the claims $1.,2.$ and $3.$ follow exactly like above. Regarding $4.$, let $\U \subset \V$ contain a Cauchy hypersurface for $\V$. By isotony, we have $\Afrak^{(\operationalstate,\sigma)}_{\mathcal{O}_\S}(\U) \subset \Afrak^{(\operationalstate,\sigma)}_{\mathcal{O}_\S}(\V)$. Since $\U$ contains a Cauchy hypersurface for $\V$, we have $\V \subset {\rm ch}(\U)$ and isotony gives $\Afrak^{(\operationalstate,\sigma)}_{\mathcal{O}_\S}(\V) \subset \Afrak^{(\operationalstate,\sigma)}_{\mathcal{O}_\S}({\rm ch}(\U)) = \Afrak^{(\operationalstate,\sigma)}_{\mathcal{O}_\S}(\U)$.
\end{proof}

\begin{corollary}
    \label{RAQFT1 cor}
    Let $\R$ be a relativistic QRF, $\operationalstate \subseteq \framestate$ be convex closed under the action of the Poincar{\'e} group, i.e., such that $(a,\Lambda) \cdot \operationalstate = \operationalstate$ for all $(a,\Lambda) \in \Poincup$, $\sigma \geq 0$ and $\mathcal{O}_\S \subseteq \bhs$ be $(\operationalstate,\sigma)$-causal. We then have:
    \begin{enumerate}
        \item (Isotony) For all $\U \subseteq \V \subseteq \mink$, $\A^{(\operationalstate,\sigma)}_{\mathcal{O}_\S}(\U) \subseteq \A^{(\operationalstate,\sigma)}_{\mathcal{O}_\S}(\V)$,
        \item (Covariance) For all $(a,\Lambda) \in \Poincup$ and $\U \subseteq \mink$, $(a,\Lambda) \cdot \A^{(\operationalstate,\sigma)}_{\mathcal{O}_\S}(\U) = \A_{\mathcal{O}_\S}^{(\operationalstate,\sigma)}((a,\Lambda) \cdot \U)$, 
        \item (Causality) For all $\U \indepsigma \V \subset \mink$, $\comm{\A^{(\operationalstate,\sigma)}_{\mathcal{O}_\S}(\U)}{\A^{(\operationalstate,\sigma)}_{\mathcal{O}_\S}(\V)} = \varnothing$.
    \end{enumerate}
    Likewise, for deterministic relational local algebras we have:
    \begin{enumerate}
        \item (Isotony) For all $\U \subset \V \in \mink$, $\Afrak^{(\operationalstate,\sigma)}_{\mathcal{O}_\S}(\U) \subset \Afrak^{(\operationalstate,\sigma)}_{\mathcal{O}_\S}(\V)$,
        \item (Covariance) For all $(a,\Lambda) \in \Poincup$, $(a,\Lambda) \cdot \Afrak^{(\operationalstate,\sigma)}_{\mathcal{O}_\S}(\U) = \Afrak^{(\operationalstate,\sigma)}_{\mathcal{O}_\S}((a,\Lambda) \cdot \U)$,
        \item (Causality) For all $\U \indepsigma \V \subset \mink$, $\comm{\Afrak^{(\operationalstate,\sigma)}_{\mathcal{O}_\S}(\U)}{\Afrak^{(\operationalstate,\sigma)}_{\mathcal{O}_\S}(\V)} = \varnothing$,
        \item (Time-slice) For all $\U \subset \V$ such that $\U$ contains a Cauchy hypersurface for $\V$, $\Afrak^{(\operationalstate,\sigma)}_{\mathcal{O}_\S}(\U) \cong \Afrak^{(\operationalstate,\sigma)}_{\mathcal{O}_\S}(\V)$.
    \end{enumerate}
\end{corollary}

Thus, both the relational local algebras and the deterministic relational local algebras satisfy the core axioms of Algebraic QFT in an extended form: the covariance property coincides with the AQFT covariance axiom if $\operationalstate$ is closed under all $(a,\Lambda) \in \Poincup$, and the strict AQFT causality holds under the assumption of perfect spacelike resolution. The first condition is satisfied in the case, for example, of $\operationalstate^c$ but not of $\operationalstate^K$ for any compact $K \subset F$, as we saw in Prop. \ref{prop: SRK covariance}, while the second is only justifiable in the context of a frame admitting a sequence of states in $\operationalstate$ localizing the reference in all spacelike directions. Hence, these \enquote{extended axioms} are arguably more operational than the \enquote{standard} axioms of AQFT in this relational context. We now construct an explicit example of relational algebraic quantum field theory.

\begin{example}
    Let $\R$ be a relativistic QRF, $K \subset F$ be compact and $\hat{\phi}$ be a free Klein-Gordon Wightmanian quantum field. Let $\sigma > 0$ and $f_{\sigma,K}$ be such that $W_{\hat{\phi}}(f_{\sigma,K})$ is weakly $(\operationalstate^K,\sigma)$-microcausal.\footnote{This $f_{\sigma,K}$ exists by Thm. \ref{thm:Weyl approximate microcausality}.} Consider
    \begin{equation}
    \label{eqn: set of Weyl operators}
    \mathcal{O}_\S^{(W,K)} := \{W_{\hat{\phi}}(f) \mid f \in C^\infty_c(\mink) \text{ s.t. }\supp f \subseteq \supp f_{\sigma,K}\} \subset \bhs \, .
\end{equation}
It is easily seen that $\mathcal{O}_\S^{(W,K)}$ is weakly $(\operationalstate^K,\sigma)$-microcausal and so is $(\operationalstate^K,\sigma)$-causal. Note furthermore that if $K \subset F$ is compact then $(a,\Lambda) \cdot K$ is compact for any $(a,\Lambda) \in \Poincup$. This, alongside Lem.~\ref{lem:causality and covariance}, ensures that $\A^{(\operationalstate^K,\sigma)}_{\mathcal{O}_\S^{(W,K)}}(\U)$ is a relational local algebra and likewise $\mathfrak{A}^{(\operationalstate^K,\sigma)}_{\mathcal{O}_\S^{(W,K)}}(\U)$ is a deterministic relational local algebra for any $\U \subseteq \mink$. Note also that Prop. \ref{prop: SRK covariance} implies that $\forall \U \subseteq \mink$ and all $(a,\Lambda) \in \Poincup$,
\begin{equation}
    (a,\Lambda) \cdot \A^{(\operationalstate^K,\sigma)}_{\mathcal{O}_\S^{(W,K)}}(\U) = \A^{(\operationalstate^{(a,\Lambda) \cdot K}, \sigma)}_{\mathcal{O}_\S^{(W,K)}}((a,\Lambda) \cdot \U) \, , \qquad (a,\Lambda) \cdot \mathfrak{A}^{(\operationalstate^K,\sigma)}_{\mathcal{O}_\S^{(W,K)}}(\U) = \mathfrak{A}^{(\operationalstate^{(a,\Lambda) \cdot K},\sigma)}_{\mathcal{O}_\S^{(W,K)}}((a,\Lambda) \cdot \U) \, .
\end{equation}
    From a Wightmanian perspective this is meaningful: if we restrict allowed smearing functions (at the level of the frame) to be smooth functions supported in a fixed compact subset of Minkowski, then under Poincaré transformations the algebra ought to be generated by the associated quantum fields with the set of smearing functions changed accordingly (with supports shifted with respect to the Poincaré transformation). Indeed, consider the set
    \begin{equation}
        \A_K(\U) := \{W(f) \mid f \in C^\infty_c(\mink) \text{ s.t. } \supp f \subseteq \U \cap K\} \, .
    \end{equation}
    Then $\A_K(\U)$ satisfies isotony and (exact) causality, but 
    \begin{equation}
        (a,\Lambda) \cdot \A_K(\U) = \A_{(a,\Lambda) \cdot K}((a,\Lambda) \cdot \U) \, ,
    \end{equation}
    which mirrors the result above. If one had allowed all compactly supported smooth functions, i.e.
    \begin{equation}
        \A(\U) := \{W(f) \mid f \in C^\infty_c(\mink) \text{ s.t. } \supp f \subseteq \U\} \, ,
    \end{equation}
    then since the set of functions $C^\infty_c(\mink)$ is closed under Poincaré transformations, the algebra would have a ``standard" covariance law
    \begin{equation}
        (a,\Lambda) \cdot \A(\U) = \A((a,\Lambda) \cdot \U) \, .
    \end{equation}
\end{example}

An interesting topic of further study would be to examine the properties of such (and related) AQFTs, particularly identify condition under which statistical independence and Haag property will hold. Investigating the von Neumann types of relational local algebras would also be of great interest in the context of recent claims that AQFT becomes better behaved when treated relationally \cite{fewster_quantum_2025,witten_algebras_2023}. Using such a language may provide a new perspective on quantum measurement theory in QFT along the lines of the Fewster-Verch formalism \cite{fewster_measurement_2023} of AQFT, and shed light on the relativistic measurement problem. See the Outlook for a brief discussion of this last research direction.

\section{Summary and Outlook}
\label{sec:Summary}

In this work, we have established the mathematical and conceptual foundations for a relational theory of quantum fields, focusing on the case of scalar fields in Minkowski spacetime. The core of our framework rests on applying the operational approach to quantum reference frames in the context of relativistic symmetry structure given by the orthochronous Poincar{\'e} group. We began our investigation by motivating the definition of a relativistic QRF as a quantum system equipped with a Poincaré-covariant POVM on the space of classical inertial reference frames. Physical quantities of a system $\S$ are then formulated as relational local observables contingent on the state preparation of such a QRF. These are shown to give rise to a natural notion of relational local observables and quantum fields in many ways analogous to those encountered in the Wightman's axiomatic approach, the relational quantum fields being spacetime kernels of relational local observables.

A key result of our investigation is a novel formulation of relativistic covariance, where a Poincaré transformation on the system is shown to be equivalent to a corresponding transformation on the state of the relativistic QRF. This relational covariance naturally links the description of the system to the perspective of observers carrying QRFs. The covariance properties of relational local quantum fields have also been analysed and shown to resemble those of \enquote{physicists'} QFT, although adjusted to the relational nature of the formalism.

We have also introduced and analysed several distinct notions of relativistic causality within the proposed framework. These include an epistemic condition analogous to Einstein causality, $(\operationalstate,\sigma)$-causality, which highlights that relational local observables commute if the supports of the spacetime marginals of their respective oriented QRF's probability distributions are spacelike separated. We also examined a stronger, ontological condition, $(\operationalstate,\sigma)$-microcausality, which imposes pointwise spacelike commutativity on the underlying relational local quantum fields. We showed that $(\operationalstate,\sigma)$-microcausality implies $(\operationalstate,\sigma)$-causality. We constructed explicit examples of relational local quantum fields satisfying finite-precision relational microcausality with respect to operationally meaningful QRF preparations.

Furthermore, we have shown that this framework makes direct contact with established formalisms of mathematical QFT. We first established that the relational vacuum expectation values and associated time-ordered correlation functions satisfy many of the properties of those found in the context of Wightman QFT, including relativistic transformation laws, Hermiticity, local commutativity and, in the case of globally oriented QRFs, spectral conditions. Moreover, we carried out a detailed analysis of RQFT in the context of Wightman QFT and managed to bring the two formalism close to each other, highlighting striking similarities and important differences. The (arguably mysterious) role of the Wightmanian test functions is played in RQFT by the frame preparations (or general trace-class operators), providing a clear operational meaning to the former---they are analogous to the frame smearing functions describing the spacetime localisation of the QRF. Relational local quantum fields always exist as bounded operator-valued spacetime kernels of relational local observables, unlike their analogues in Wightman theory. We have also constructed algebras of relational local observables associated with spacetime regions and proved that these satisfy an extended version of the foundational axioms of AQFT, namely isotony, covariance, causality and the time-slice axiom, with the standard axioms satisfied in special cases. 

The framework presented here provides an operationally motivated and mathematically rigorous approach to scalar QFT on Minkowski spacetime stemming from a relational and operational perspective. It recasts fundamental concepts such as observables, covariance and causality in terms of the relationship between a system and the quantum frame by means of which it is being described.

This paper offers the foundation for a research programme with many possible outlooks. For example, we focused on the study of scalar quantum fields, especially at the level of the covariance properties we examined. Of course, one wants to discuss fermions as well as gauge bosons, which do not transform as scalars under Poincaré transformations. We believe that the framework can be extended to account for fermionic fields, with different covariance and causality conditions. The interplay between spacelike commutativity and covariance is usually captured by a spin-statistics theorem \cite{streater_pct_1989}, which rules out e.g. joint fermionic commutativity and bosonic covariance. Understanding whether such a spin-statistics theorem holds à la Wightman beyond Schwartz QRFs, or whether there exists “exotic” QRF preparations which can give rise to such violations, is an interesting open question. Indeed, one can ask whether spin-statistics is QRF-dependent: could an electron “appear” to be bosonic from the point of view of another electron? This will be the content of a forthcoming paper.

One may also wish to extend this construction for more general theories, including gauge theories which present some additional gauge group or groupoid structure on more general principal bundles. This should also allow for a discussion of RQFT on curved spacetimes. This extension is currently underway. Note that in general, the vacuum state $\Omega$ is not invariant under gauge group transformations if on top of $\Poincup$ there is another group $G$ playing a role in relativisation, so it will change under such a relativisation, which may highlight the notion of vacuum polarisation. Further note that in such cases, the way the vacuum polarises is dependent on the oriented QRF (even when relativised with respect to the same QRF, but with different marginal probability distributions over the gauge group). In physics, this is usually associated to loop diagrams associated to self-interactions in the propagator; here, it arises at the level of the frames. Whether this is the same notion of vacuum polarisation as that for which $H_I \ket{\Omega} \neq 0$, where $H_I$ is some interaction Hamiltonian, remains to be seen. 

The tools developed in this paper can plausibly be exported to the study of relational quantum field theory in the Euclidean setting. Indeed, let $F_\euc \cong \euc \times \mathcal{SO}(d)$ where $\euc$ is $d$-dimensional Euclidean space, understood as a manifold, and $\mathcal{SO}(d)$ is a torsor for the special orthogonal group $SO(d)$. The definition of relational (local) Euclidean quantum field then follows exactly as in the Lorentzian case, with similar notions of covariance (now with respect to $ISO(d)$). Such fields give a very similar story to Osterwalder-Schrader Euclidean quantum field theory \cite{osterwalder_axioms_1973}, and examining the links between both theories would be very insightful. Understanding Wick rotations and reconstruction
theorems allowing to go from one formulation to the other is another interesting topic to explore. It may also shed some light on the possibility to define relational path integrals, partition functions and Feynman diagrams in this language.

On another note, in conjunction with the previous discussion of RQFT on principal bundles, the Euclidean setup can be seen as QFT on a spacetime without specified causal structure (Lorentzian metric tensor). In the general context, this idea could be implemented by considering RQFT on the full frame bundle, as opposed to the (restricted) Lorentz bundle associated with a particular choice of a metric field. The ultimate goal of the presented formalism, alongside improving mathematical and conceptual foundations of QFT, would be to provide novel ways in which gravity can be peacefully reconciled, in operational and relational way, with Quantum Theory.

Other interesting outlooks include a discussion of relational quantum field dynamics, whereby one should be able to discuss solutions to differential equations associated to specific representations of $\Poincup$ for the relational quantum fields. Whether the symplectic structure of phase space remains unchanged under relativisation, or whether the localisability of QRFs generate genuine changes to the canonical commutation relations \cite{jorquera_riera_uncertainty_2025}, is to be investigated.

The relational account of interactions arise in the context of composite systems along the following lines. If we take the system Hilbert space to be $\mathcal{H}_\mathcal{S} = \mathcal{H}_1 \otimes \mathcal{H}_2$, relativizing a tensor product operator $\phi_1\phi_2 \equiv \phi_1 \otimes \phi_2$ gives rise to relational local observables of the form 
\begin{equation}
    \hat{\Phi}_{1,2}^\mathcal{R}(\omega) = \int_F (\hat{\phi}_1)_\lambda(x)(\hat{\phi}_2)_\lambda(x) \, d\mu^{\E_\R}_\omega(x,\lambda) = \int_\mink (\hat{\phi}_1\hat{\phi}_2)^\R_\omega(x) \, d\mu^{\F_\R}_\omega(x)\,,    
\end{equation}
which can be understood as describing an \emph{interaction} between the relational local quantum fields. The frame then not only specifies the coupling region, as is the case in the Fewster-Versch framework for measurement schemes (see \cite{fewster_quantum_2020,fewsteru}) combined with the QRFs (see \cite{fewster_quantum_2025}), but also defines the \emph{coupled theory}. One should then be able study the associated scattering maps and update rules, potentially providing a large number of tractable relational models for measurement schemes.\\

There also seems to exist a connection between the relational interaction terms and the detector models (see e.g. \cite{perche}), along the following lines. Upon a choice of a (space-like) slicing of the interaction region, it can be written in the form $\Sigma \times I$ where $I$ is a time interval of the interaction. Integrating over a slice $\Sigma_t$ and assuming existence of relational smearing function then gives a \emph{relational local interaction Hamiltonian} of the form
\begin{equation}
    H^\R_\omega (t) := \int_{\Sigma_t} (\hat{\phi}_1\hat{\phi}_2)^\R_\omega(\vec{x},t) \, f^\R_\omega(\vec{x},t)d\vec{x}\,,  
\end{equation}
with the function $f^\R_\omega(\vec{x},t)$ dictating the shape of the interaction describing how the detector, here modelled by the quantum field $\hat{\phi}_2$, couples to the quantum field $\hat{\phi}_1$; this has a direct analogue in the detector models approach, which provides an arena for further explorations.

Given the similarities between Wightman QFT and RQFT, it is also plausible that a collision theory à la Haag-Ruelle \cite{haag_local_1996}, which is traditionally implemented in the Wightman framework, can be translated to the language of RQFT. Likewise, the LSZ formalism \cite{lehmann_zur_1955} should have a similar formulation within RQFT. Understanding the notion of (relational) asymptotic states, the meaning of the unitarity axiom $\his^{-\infty} = \his^{+\infty}$ and of the asymptotic completeness axiom $\his =\his^{-\infty} = \his^{+\infty}$ within RQFT, and whether these should be implemented at the level of the absolute Hilbert space $\his$ or at the level of its relativisation with respect to some relativistic QRF $\R$, are important to link the foundations laid in this paper to well-understood computable quantities in physics. 

Some further results related to vacuum expectation values and time-ordered correlation functions would be interesting to analyse in the setting of RQFT. For example, we keep for future work the derivation of the spin-statistics and CPT theorems, which rely on continuing the vacuum expectation values to holomorphic functions into extended Jost tubes. Whether some violations of these core theorems can arise in certain classes of oriented QRFs (e.g. ones which are not Schwartz) is an interesting open question. Furthermore, in the light of the properties of the vacuum expectation values provided in Sec. \ref{sec:Wightman functions}, it seems likely that a reconstruction theorem à la Wightman could also be formulated in the context of RQFT.  

Finally, the concepts of regularisation and renormalisation can be naturally included in RQFT – these deserve (at the very least) a separate paper of their own, but we here outline the philosophy behind what we call relational renormalisation. Indeed, the process of \enquote{changing the scale at which one looks at the physics} is encoded in the POVM and preparation of the measurement apparatus with which one looks at a system. In effect, looking at the subatomic physics through the \enquote{lens} of the large hadron collider can be seen as localising a measurement apparatus at that given scale, while looking at the structure of molecules through a microscope is modelled by preparing a POVM in a certain state so that the measurement apparatus resolves the scale reachable by the microscope.

One way to implement this notion of \enquote{changing scales} in the context of RQFT is through external frame transformations \cite{glowacki_relativization_2024,glowacki_towards_2024}. A description of $\S$ relative to a QRF $\R$ can be transformed into a description of $\S$ relative to another QRF $\R'$ along a channel $\psi : \bhr \to \bhrp$ such that
\begin{equation}
    \label{eqn:equivariant channel}
    \E_\R' = \psi \circ \E_\R
\end{equation}
If the channel is equivariant, then the new frame observable $\E_{\R'}$ is covariant with respect to the same group as the covariance group of $\R$. In this case, we have that for all $\phi \in \bhs$ and $\omega \in \framestate$,
\begin{equation}
    \label{eqn:relationship btwn oriented QRFs}
    \yen^{\R'}(\phi) = (\mathbb{1}_{\bhs} \otimes \psi) \yen^\R(\phi) \Leftrightarrow \yen^{\R'}_{\omega'}(\phi) = \yen^\R_{\psi_*(\omega')}(\phi)
\end{equation}
where $\psi_* : \framestatep \to \framestate$ is the pre-dual of $\psi$. We envision the renormalisation group flows to be modelled as external QRF transformations: consider a collection of channels $\psi_x : \bhr \to \bhrp$ where $x \in \mathbb{R}^+$ for $\psi_0 = \mathbb{1}$, which vary the localizability and covariance properties of the resulting frame. This process may be understood as a \enquote{coarse graining} of the physics that the measurement apparatus can resolve.

\paragraph*{Acknowledgments}

The authors would like to thank Christy Kelly for interesting discussions and comments on the first draft of this work. S.F. would like to thank Adrian Kent for his continuous support. J.G. would like to thank Prof. Klaas Landsmann, Prof. Leon Loveridge, and Prof. Markus M{\"u}ller. J.G. would also like to express deep gratitude for various forms of hospitality and support he received during the time this work was developed from Dr. John Selby, Dr. Ana Belen and ICTQT, Prof. Aleks Kissinger, Prof. Jonathan Barrett and CS group at Oxford, Prof. Lucien Hardy and PI, Dr. Philipp Hoehn and OIST, as well as and Prof. Chris Fewster and Dr. Philipp Hoehn for inspiring interactions.

This research was funded in in part by the Austrian Science Fund (FWF) 10.55776/PAT1562525. For open access purposes, the author has applied a CC BY public copyright license to any author accepted manuscript version arising from this submission.

This publication was made possible through the support of the ID\# 62312 grant from the John Templeton Foundation, as part of the \href{https://www.templeton.org/grant/the-quantuminformation-structure-ofspacetime-qiss-second-phase}{‘The Quantum Information Structure of Spacetime’ Project (QISS)}. The opinions expressed in this project/publication are those of the author(s) and do not necessarily reflect the views of the John Templeton Foundation.

S.F. is funded by a studentship from the Engineering and Physical Sciences Research Council, Grant No. 2882481.

\paragraph*{Conflict of interest.} The authors have no conflict of interest to declare that are relevant to the content of this article.

\paragraph*{Data availability.} No datasets were generated or analysed in this work.

\printbibliography[title={References}]

\addtocontents{toc}{\protect\setcounter{tocdepth}{-1}}

\newpage
\begin{appendices}

\section{Technical preliminaries}\label{App: tech prem.}

\subsection{Functional analysis}\label{App: functional analysis}
\paragraph*{Operators.} An operator $A: \hi \to \hi$ on a Hilbert space $\hi$ is bounded iff its operator norm
\begin{equation}
    ||A|| := \sup_{||\xi|| = 1}||A\xi|| = \sup_{\rho \in \state}|\Tr[\rho A]|
\end{equation}
is finite. The vector space of bounded operators is complete under this norm; this Banach space will be denoted $\bh$---it is a subspace of $\lh$, the space of linear operators on $\hi$. A bounded operator is self-adjoint/positive if it has real/non-negative spectrum. Self-adjoint bounded operators $B(\hi)^{\rm sa}$ form a real Banach space under the operator norm; relation $A \geq B$ iff $A-B $ is positive gives partial order on $B(\hi)^{\rm sa}$, and $\id_\hi \in B(\hi)^{\rm sa}$ provides a unit making $B(\hi)^{\rm sa}$ an order unit space \cite{kuramochi_compact_2020}. The subset of effects is the unit interval in $B(\hi)^{\rm sa}$ written
\begin{equation}
\Eff(\hi) := \{\Ef \in B(\hi)^{\rm sa} \h | \h \mathbb{0}_\hi \leq \E \leq \id_\hi\}.
\end{equation}

A bounded operator $T: \hi \to \hi$ on a Hilbert space $\hi$ is trace-class iff its trace-class norm
\begin{equation}
||T||_1 := \Tr[\sqrt{T^\dagger T}]
\end{equation}
is finite; the trace-class norm of a positive operator is just its trace. The vector space of trace-class operators is complete under this norm; this Banach space will be denoted $\T(\hi)$. Self-adjoint trace-class operators $\T(\hi)^{\rm sa}$ form a real Banach space under the trace-class norm, the positive trace-class operators $\T(\hi)_+ \subset \T(\hi)^{\rm sa}$ forming a generating cone, and the subset of \emph{states}
\begin{equation}
\state := \{\rho \in \T(\hi)_+ \, | \, \Tr[\rho]=1\}
\end{equation}
forms a base for $\T(\hi)_+$, making $\T(\hi)^{\rm sa}$ a base-norm space \cite{kuramochi_compact_2020}. A von Neumann algebra is a $*$-algebra of bounded operators on a Hilbert space that is closed in the ultraweak topology and contains the identity operator.  \\

\paragraph*{Channels.} The linear maps between operator algebras
\begin{equation}
    \Phi: B(\hi) \to B(\hik)
\end{equation}
that are continuous with respect to the ultraweak topologies are referred to as \emph{normal}, \emph{unital} if $\Phi(\id_\hi) = \Phi(\id_{\hik})$, \emph{positive} if $\Phi(B(\hi)_+) \subseteq B(\hik)_+$. A linear map as above is called $n$\emph{-positive} if $\id_n \otimes \Phi: B(\Cn^n \otimes \hi) \to B(\Cn^n \otimes \hik)$ is positive, and \emph{completely positive (CP)} if it is $n$\emph{-positive} for all $n \in \Nn$. Normal unital CP maps which are trace-preserving (TP) are referred to as (quantum) \emph{channels} or CPTP maps. Normal \emph{functionals} on $B(\hi)$ are precisely those given by evaluating the corresponding bounded functionals on a chosen trace-class operator
\begin{equation}
\varphi_T: B(\hi) \ni A \mapsto \Tr[T A] \in \Cn,
\end{equation}
quantum states being characterised as normal unital CP functionals, 
i.e, channels into the complex numbers.\footnote{Note that in the case of functionals, positivity and complete positivity are equivalent.} Thus, since channels compose, a channel defines a \emph{predual map} between state spaces\footnote{Normality, positivity and unitality is sufficient for the existence of a predual map, complete positivity is unnecessary.}
\begin{equation}
    \Phi_*: \mathscr{D}(\hik) \ni \varphi_{\rho} \mapsto \Phi \circ \varphi_{\rho} \in \state,
\end{equation}
where states have been identified with the corresponding functionals. 
Equivalently, $\Phi_*$ is specified by
\begin{equation}
    \Tr[\rho \Phi(A)] = \Tr[\Phi_*(\rho)A] \text{ for all } A \in B(\hi), \h \rho \in \mathscr{D}(\hik).
\end{equation}
A quantum channel is said to be \emph{equivariant} if it commutes with the action of a group acting on both $\hi$ and $\hik$. \\

\paragraph*{Topologies.} The space of bounded operators is the Banach dual order unit space for $\T(\hi)$, written $B(\hi) \cong \T(\hi)^*$. This because we have a norm-preserving bijection between bounded linear operators $A \in B(\hi)$ and the continuous functionals on $\T(\hi)$ they give rise to via the trace,~i.e. \cite{takesaki_theory_2001}
\begin{equation}
    B(\hi) \ni A \longmapsto \{\phi_A: \T(\hi) \ni T \mapsto \Tr[TA] \in \Cn\} \in \T(\hi)^*.
\end{equation}
    The $\T(\hi)^* \cong B(\hi)$ duality allows to define the dual pair of useful and operationally justified topologies on $\bh$ and $\thi$ as follows
    \begin{itemize}
        \item $A_n \to A$ in $B(\hi)$ iff for any $T \in \T(\hi)$ we have $\Tr[TA_n] \to \Tr[TA]$ in $\Cn$,
        \item $T_n \to T$ in $\T(\hi)$ iff for any $A \in B(\hi)$ we have $\Tr[T_nA] \to \Tr[TA]$ in $\Cn$.
    \end{itemize}
    The first of these topologies of convergence of expectation values is locally convex and metrizable on bounded parts and is referred to as the \emph{ultraweak} or $\sigma$-\emph{weak} operator topology \cite{takesaki_theory_2001} ($\sigma(\bh,\thi)$); the second we call \emph{operational} \cite{carette_operational_2025} ($\sigma(\thi,\bh)$). The subsets of effects and states inherit ultraweak and operational topologies from $\bh$ and $\thi$, respectively. On any norm-bounded set the weak operator and ultraweak topologies are the same. \\

\paragraph*{Operator-valued measures.}\label{App: POVMs}
An operator-valued measure (OVM) is a direct analogue of a (complex) measure in Lebesgue theory: given a measurable space $(\Sigma,\F)$, where $\Sigma$ is a set and $\F$ a $\sigma$-algebra of subsets of $\Sigma$, an OVM on $(\Sigma,\F)$ is a set function with values in the space of bounded operators on a Hilbert space, i.e.,
\[
    \E: \F \to B(\hi),
\]
such that for any $\omega \in \S(\hi)$ the associated set function
\[
    \E_\omega: \F \ni X \longmapsto  \tr[\omega \E(X)] \in \Cn
\]
is a \emph{measure}; an OVM is normalized if $\E(\Sigma)=\id$ and positive (POVM) if $\E(X) \in B(\hi)_+$ for all $X \in \F$. The measures $\E_\omega$ are probability measures for all $\omega \in \S(\his)$ iff both these conditions are satisfied.\footnote{Notice here that the normalization condition is often understood as part of the definition of a POVM, although it is logically independent from positivity. We also acknowledge equivalent definitions of POVMs to be found in the literature. Namely, the set map $\F$ can be assumed to give a probability measure via $X \mapsto  \braket{\xi}{\E(X)\eta}$ for any $\xi,\eta \in \hi$ and $X \in \F$. Yet another equivalent definition can be given \cite{busch_quantum_2016} by requiring that $\E(\emptyset) = \mathbb{0}_\hi$, $\E(\Sigma) = \id_\hi$, and that for any sequence of \emph{disjoint} measurable subsets $\{X_n\}_{n \in \mathbb{N}} \subset \F$ we have $\E\left(\cup_{n=1}^\infty X_n\right) = \sum_{n=1}^\infty \E(X_n)$, with the sum understood in terms of ultraweak convergence. (In \cite{busch_quantum_2016} weak convergence is invoked, but since $\Eff(\hi) \subset B(\hi)$ is bounded these topologies agree.)} In this case we have that $\E(X) \in \Eff(\hi)$ for all $X \in \F$ and the operators $\E(X)$ are referred to as the \emph{effects of} $\E$. POVMs are a direct analog of probability measures and exhaust the probabilistic structure of quantum theory in the following sense: due to the discussed duality $\T(\hi)^* \cong B(\hi)$, any assignment
\[
\S(\hi) \ni \omega \mapsto \mu_\omega \in {\rm Prob}(\Sigma,\F)
\]
such that for any $X \in \F$ the map $\omega \mapsto \mu_\omega(X)$ is (trace-norm) \emph{continuous}, needs to be given via~a~POVM, i.e, there is a POVM $\E$ such that $\mu_\omega = \E_\omega$. A positive operator-valued measure is called \emph{sharp}, or a \emph{projection-valued measure} (PVM), if all its effects are projections. Prominent examples of PVMs are those arising from self-adjoint operators via the spectral theorem; they are always defined over the spectrum of the operator, which is a subset of the real line. All the effects of a PVM will commute, and those associated to disjoint measurable subsets compose to zero, i.e., if $\E$ is sharp we have
\[
    \E(X)\E(Y)= \E(Y)\E(X) \text{ for all } X,Y \in \F, \text{ and }
    \E(X)\E(Y) = 0 \text{ for all } X,Y \in \F \text{ such that } X \cap Y = \emptyset.
\]

Another interesting class of POVMs are those called \emph{localizable} \cite{carette_operational_2025,glowacki_operational_2023}. A POVM is localizable if for any $x \in \Sigma$ we can find a sequence of states $\{\omega_n^x\}_{n \in \mathbb{N}} \subset \S(\hi)$, called a \emph{localizing sequence}, such that the corresponding measures converge weakly to the Dirac measure $\delta_x$ in that 
we have\footnote{Usually purity of the states in localizing sequence is assumed, although it does not seem necessary. As shown in \cite{glowacki_operational_2023}, the definition we give is equivalent to the one given in \cite{heinonen_norm-1-property_2003} on metrizable sample spaces. Let us also note here that in principle one can consider localizable OVMs without assuming positivity.}
\[
    \lim_{n \to \infty} \int_\Sigma f(y) d\E_{\omega_n^x}(y) =
    f(x).
\]
Like ordinary measures, OVMs are subject to some natural constructions. For example, given a measurable function $\varphi: (\Sigma,\F) \to (\Sigma',\F')$ and a (P)OVM $\E: \F \to \Eff(\hi)$, the map
\[
\varphi_*\E := \E \circ \varphi^{-1}: \F' \ni X \longmapsto \E(\varphi^{-1}(X)) \in B(\hi)
\]
defines a \emph{push-forward} (P)OVM on $(\Sigma',\F')$; for all $\omega \in \S(\hi)$ we have $(\varphi_*\E)_\omega = \varphi_*(\E_\omega)$. Moreover, given a (P)OVM $\E: (\Sigma,\F) \to \Eff(\hi)$ and a quantum channel $\psi: B(\hi) \to B(\hik)$, the map
\[
    \psi \circ \E: \F \ni X  \longmapsto \psi(\E(X)) \in B(\hik)
\]
is another (P)OVM on $(\Sigma,\F)$ but now with $\E(X)$ in $B(\hik)$ ($\Eff(\hik)$). One easily verifies that (see \cite{glowacki_relativization_2024} for POVMs)
\begin{equation}\label{eq:postcompre}
(\psi \circ \E)_\omega = \E_{\psi_*\omega}.
\end{equation}
Lastly, given a pair of OVMs on the same quantum system but possibly different sample spaces, i.e.,
\[
\E: \F \to \Eff(\hi), \text{ and } \E': \F' \to \Eff(\hi),
\]
where $(\Sigma,\F)$ and $(\Sigma',\F')$ are the relevant measurable spaces, we can define a \emph{product OVM} via
\[
    \E \times \E': \F \times \F' \ni X\times Y \longmapsto \E(X)\E'(Y) \in B(\hi),
\]
where $\F \times \F'$ denotes the $\sigma$-algebra of subsets of $\Sigma \times \Sigma'$ generated by those of the form $X \times Y \subseteq \Sigma \times \Sigma'$ with $X \in \F$ and $Y \in \F'$. Such an OVM is positive e.g. if $\E$ and $\E'$ are both positive and $[\E(X),\E'(Y)]=0$ for all $X,Y$.

\subsection{Fubini-Tonelli theorem}\label{App: Fubini-Tonelli}

    An important measure-theoretic result that we use extensively in this paper is the Fubini-Tonelli theorem. It uses the notion of $\sigma$-finiteness: a measure space $(\Sigma,\mathcal{F},\mu)$ is called \emph{$\sigma$-finite} if $\Sigma$ is the union of a sequence of measurable spaces $(A_i,\mu)_{i\in\mathbb{N}}$ (i.e. $\cup_{i \in \mathbb{N}} A_i = \Sigma$ where $A_1, A_2, ... \in \mathcal{F}$) of finite measure $\mu(A_i)<\infty$ for all $i$. For example, the Lebesgue measure and probability measures on $\mathbb{R}^n$ are $\sigma$-finite. If $(\Sigma_1,\mathcal{F}_1,\mu_1)$ and $(\Sigma_2,\mathcal{F}_2,\mu_2)$ are $\sigma$-finite, then there is a unique product measure on $(\Sigma_1,\mathcal{F}_1) \times (\Sigma_2,\mathcal{F}_2)$.

    \begin{theorem}[Fubini-Tonelli]
    	If $(\Sigma_1,\mathcal{F}_1,\mu_1)$ and $(\Sigma_2,\mathcal{F}_2,\mu_2)$ are $\sigma$-finite measure spaces, and if $f : \Sigma_1 \times \Sigma_2 \to \mathbb{C}$ is a $(\mathcal{F}_1 \times \mathcal{F}_2)$-measurable function, then
    	\begin{equation}
    		\int_{\Sigma_1} \Bigg(\int_{\Sigma_2} \abs{f(x,y)} d\mu_2(y) \Bigg) d\mu_1(x) = \int_{\Sigma_2} \Bigg(\int_{\Sigma_1} \abs{f(x,y)} d\mu_1(x) \Bigg) d\mu_2(y) = \iint_{\Sigma_1 \times \Sigma_2} \abs{f(x,y)} d(\mu_1 \times \mu_2)(x,y)
    	\end{equation}
    	where $\mu_1 \times \mu_2$ is the (unique) product measure on $(\Sigma_1,\mathcal{F}_1) \times (\Sigma_2,\mathcal{F}_2)$, and if any one of these integrals is finite, then
    	\begin{equation}
    		\int_{\Sigma_1} \Bigg(\int_{\Sigma_2} f(x,y) d\mu_2(y) \Bigg) d\mu_1(x) = \int_{\Sigma_2} \Bigg(\int_{\Sigma_1} f(x,y) d\mu_1(x) \Bigg) d\mu_2(y) = \iint_{\Sigma_1 \times \Sigma_2} f(x,y) d(\mu_1 \times \mu_2)(x,y) \, .
    	\end{equation}
    \end{theorem}

  This theorem generalizes to Bochner integrals of Banach space valued functions (see \cite{bochnera,diestel,bourgin,Mikusinski1978} for the basic theory).

\begin{theorem}[Fubini-Tonelli for Bochner \cite{Mikusinski1978}]\label{thm:Fubini-Tonelli}
    \label{thm:Fubini-Tonelli Bochner}
    For any Bochner-integrable $\hat{\phi} \in L^1( \Sigma_1 \times \Sigma_2, \bhs)$ we have
    \begin{equation}
    		\int_{\Sigma_1} \Bigg(\int_{\Sigma_2} \hat{\phi}(x,y) d\mu_2(y) \Bigg) d\mu_1(x) = \int_{\Sigma_2} \Bigg(\int_{\Sigma_1} \hat{\phi}(x,y) d\mu_1(x) \Bigg) d\mu_2(y) = \iint_{\Sigma_1 \times \Sigma_2} \hat{\phi}(x,y) d(\mu_1 \times \mu_2)(x,y) \, .
    	\end{equation}
\end{theorem}

\begin{proof}
  (sketch) The Bochner integral can be seen as a bounded linear operator from the space of Bochner-integrable functions \cite{diestel}
  \begin{equation}
     \I = \int_\Sigma d\mu:L^1(\Sigma,\bhs) \to \bhs.
  \end{equation}
Now since for any Bochner-integrable function $f:\Sigma \to X$ and bounded linear operator $T:X \to Y$ we have \cite{diestel}
\begin{equation}
    T\left(\int_\Sigma f\,d\mu\right) = \int_\Sigma (T\circ f)\,d\mu
\end{equation}
and the space $L^1(\Sigma_1 \times \Sigma_2,\bhs)$ is isometrically isomorphic to $L^1(\Sigma_2,L^1(\Sigma_1,\bhs))$, we can write
\begin{equation}
\begin{aligned}
    		\int_{\Sigma_1} \Bigg(\int_{\Sigma_2} \hat{\phi}(x,y) d\mu_2(y) \Bigg) d\mu_1(x) &= \int_{\Sigma_1} \left(\I_2 \circ \hat{\phi}(x,y)\right) \, d\mu_1(x) \\&=
            \I_2 \left(\int_{\Sigma_1} \hat{\phi}(x,y) d\mu_1(x)\right) =
            \int_{\Sigma_2} \Bigg(\int_{\Sigma_1} \hat{\phi}(x,y) d\mu_1(y) \Bigg) d\mu_2(x)\, .
\end{aligned}
\end{equation}
\end{proof}

Since we are only dealing with finite integrals, this is enough for our applications in this paper.
  
    \subsection{Schwartz functions and distributions}
    \label{App: distributions}

    The Schwartz space is defined as the space of rapidly decreasing smooth functions on $\mathbb{R}^4$:
    \begin{equation}
        \mathscr{S}(\mathbb{R}^4,\mathbb{C}) := \left\{f \in C^{\infty}(\mathbb{R}^4,\mathbb{C}) \mid \forall \underline{\alpha}, \underline{\beta} \in \mathbb{N}^d, \norm{f}_{\underline{\alpha},\underline{\beta}} < \infty \right\}
    \end{equation}
    where $\underline{\alpha}, \underline{\beta} \in \mathbb{N}^d$ are multi-indices, $\norm{f}_{\underline{\alpha},\underline{\beta}} := \sup_{\mathbf{x} \in \mathbb{R}^4} \abs{\mathbf{x}^{\underline{\alpha}} (\mathbf{D}^{\underline{\beta}}f)(\mathbf{x})}$, $\mathbf{x}^{\underline{\alpha}} := x_1^{\alpha_1} \cdots x_d^{\alpha_d}$ and $\mathbf{D}^{\underline{\beta}} := \partial_1^{\beta_1} \cdots \partial_d^{\beta_d}$. Note that any smooth function with compact support (i.e. a bump function) is in $\mathscr{S}(\mathbb{R}^4,\mathbb{C})$. The dual space of $\mathscr{S}(\mathbb{R}^4,\mathbb{C})$ is denoted $\mathscr{S}(\mathbb{R}^4,\mathbb{C})^*$ and called the space of tempered distributions.

    \begin{theorem}[Nuclear theorem \cite{schwartz_theorie_1947}]
        \label{thm:Schwartz nuclear theorem}
        Let $T_n : \prod_{i=1}^n \mathscr{S}(\mathbb{R}^4,\mathbb{C}) \to \mathbb{C}$ be multilinear and continuous in each of its arguments (the others being fixed). Then there exists a unique distribution $\mathcal{T}_n \in \mathscr{S}(\mathbb{R}^4,\mathbb{C})^*$ such that 
    \begin{equation}
        T_n(f_1,...,f_n) = \mathcal{T}_n(f_1 \otimes f_2 \otimes ... \otimes f_n) \, .
    \end{equation} 
    \end{theorem}

    $\mathcal{T}_n$ has an integral kernel $\mathfrak{T}_n : \prod_{i=1}^n \mathbb{R}^4 \to \mathbb{C}$ iff for every sequence $(g_k)$ in $\prod_{i=1}^n \mathscr{S}(\mathbb{R}^4,\mathbb{C})$ such that $0 \leq g_k \leq g$ for some $g \in \prod_{i=1}^n \mathscr{S}(\mathbb{R}^4,\mathbb{C})$ and $g_k \to 0$ locally in measure, the sequence $(\mathcal{T}_n g_k)$ converges to $0$ almost everywhere \cite{schep_kernel_1979}. In such cases only, we can write
    \begin{equation}
        \mathcal{T}(f_1 \otimes f_2 \otimes ... \otimes f_n) = \int_{\mink} f_1(x_1) ... f_n(x_n) \mathfrak{T}_n(x_1,...,x_n) d^dx_1 ... d^dx_n \, .
    \end{equation}

    \section{No-go: Wizimirski}\label{App:no go Wizimirski}

One may be tempted to define scalar fields as those operators which lie in the Lorentz-invariant subalgebra $\bhs^{\Lup}$ of $\bhs$. This seemingly generalises nicely to higher spins: for example, in $1+3$ dimensions, a Dirac fermion could be a collection of four operators in $\bhs$ which do not belong to $\bhs^{\Lup}$ but which are related through Lorentz transformations as $\hat{\phi}_a(0) = \sum_{b=1}^4 (D^{(1/2,0)} \oplus D^{(0,1/2)})[\Lambda^{-1}]_{ab} \hat{\phi}_b(0)$ for all $\Lambda \in SL(2,\mathbb{C})$, where $(D^{(1/2,0)} \oplus D^{(0,1/2)})$ is the Dirac representation of $SL(2,\mathbb{C})$. However, we have to consider a no-go theorem by Wizimirski's, which puts a halt to these hopes.

\begin{theorem}[Wizimirski \cite{wizimirski_existence_1966}]
    \label{thm:Wizmirski}
	Let $\hat{\phi} : \mink \to \bhs$ be an operator-valued function and $U$ be a weakly continuous unitary representation of the Poincaré group on $\his$ such that
	\begin{enumerate}
		\item $U(y,\Lambda)^\dagger \hat{\phi}(x) U(y,\Lambda) = \hat{\phi}(\Lambda x + y)$ for all $(y,\Lambda) \in \Poincup$ and $x \in \mink$,
		\item There exists a unique pure translation-invariant state $\Omega = \ket{\Omega}\bra{\Omega} \in \systemstate^{T(1,d-1)}$.
	\end{enumerate}
    Then $\hat{\phi}(x) \ket{\Omega} = \hat{\phi}(0)\ket{\Omega}$ for all $x \in \mink$. Furthermore, if $\ket{\Omega} \in \his$ is a unique (up to scalar multiples) translation-invariant vector, then there is a $c \in \mathbb{C}$ such that $\hat{\phi}(x) \ket{\Omega} = c\ket{\Omega}$ for all $x \in \mink$.
\end{theorem}

\begin{proof}
    Let $f:\mink \to \mathbb{C}$ be defined as $f(x) = \expval{\hat{\phi}(x)^\dagger \hat{\phi}(0)}{\Omega}$ where $\Omega = \ket{\Omega}\bra{\Omega}$. Writing $U(x,e) \ket{\Omega} = e^{i\theta(x)} \ket{\Omega}$ for some $\theta(x) \in [0,2\pi)$, it is easily seen that $f(x) = e^{i\theta(x)} \expval{\hat{\phi}(0)^\dagger U(x,e) \hat{\phi}(0)}{\Omega}$. Moreover, $U(x,e) U(0,\Lambda) \ket{\Omega} = U(x,\Lambda) \ket{\Omega} = U(0,\Lambda)U(\Lambda^{-1}x,e) \ket{\Omega} = e^{i \theta(\Lambda^{-1}x)}U(0,\Lambda) \ket{\Omega}$ for all $x \in \mink$ and all $\Lambda \in \Lup$, so $U(0,\Lambda)\ket{\Omega} = e^{i \chi(\Lambda)} \ket{\Omega}$ for some $\chi(\Lambda) \in [0,2\pi)$ by uniqueness. Furthermore, for all $\Lambda \in \Lup$ and all $x \in \mink$,
    \begin{equation}\begin{aligned}
        f(\Lambda x) &= \expval{\hat{\phi}(\Lambda x)^\dagger \hat{\phi}(0)}{\Omega} \\
        &= \expval{U(0,\Lambda^{-1})^\dagger U(x,e)^\dagger U(0,\Lambda)^\dagger \hat{\phi}(0)^\dagger U(0,\Lambda)U(x,e)U(0,\Lambda^{-1}) \hat{\phi}(0)}{\Omega} \\
        &= e^{i\chi(\Lambda)} \expval{U(x,e)^\dagger U(0,\Lambda)^\dagger \hat{\phi}(0)^\dagger U(0,\Lambda)U(x,e)U(0,\Lambda^{-1}) \hat{\phi}(0)}{\Omega} \\
        &= \expval{U(x,e)^\dagger \hat{\phi}(0)^\dagger U(x,e) U(0,\Lambda^{-1}) \hat{\phi}(0) U(0,\Lambda^{-1})^\dagger} {\Omega} \\
        &= \expval{\hat{\phi}(x)^\dagger \hat{\phi}(0)}{\Omega} = f(x) \, .
    \end{aligned}\end{equation}
    Since $e^{i(\theta(x)+\theta(y)} \ket{\Omega} = U(x,e) U(y,e) \ket{\Omega} = U(x+y,e)\ket{\Omega}=e^{i \theta(x+y)} \ket{\Omega}$, we have $\theta(x+y) \equiv \theta(x) + \theta(y) \mod 2\pi$. Let $g(x):= \expval{\hat{\phi}(0)^\dagger U(x,e) \hat{\phi}(0)}{\Omega}$, $x_1,\cdots,x_N \in \mink$, $v_i := e^{i \theta(x_i)}$ and $D = \text{diag}(v_1,\cdots,v_n)$. Then $f(x_i - x_j) = v_i \overline{v_j} g(x_i - x_j)$ so $[f(x_i-x_j)] = D\cdot [g(x_i-x_j)]\cdot D^*$. But $g$ is a continuous function of positive type and $D$ is unitary, and the conjugation of a positive semi-definite matrix by a unitary matrix is also positive-semidefinite, so $f$ is a continuous function of positive type. Hence by Bochner's theorem $f$ is the Fourier transform of a bounded measure $\mu$ on $\mink$. But $f$ being Lorentz invariant implies that $\mu$ also is. But a bounded, Lorentz invariant measure is supported on $\{0\}$. Thus, $f(x) = 1$ i.e. it is constant. Writing $\ket{\psi} := \hat{\phi}(0)\ket{\Omega}$, $f(x) = f(0)$ for all $x \in \mink$ implies that
    \begin{equation}
        e^{i \theta(x)}\expval{U(x,e)}{\psi} = \braket{\psi} = \norm{\psi}^2 = \norm{\psi} \cdot \norm{U(x,e) \psi}
    \end{equation}
    so by the Cauchy-Schwartz inequality, $U(x,e) \ket{\psi} = e^{-i \theta(x)} \ket{\psi}$ for all $x \in \mink$. Hence we get
    \begin{equation}
    U(x,e) \hat{\phi}(0) \ket{\Omega} =  e^{i \theta(x)} \hat{\phi}(x) \ket{\Omega} = e^{-i \theta(x)} \hat{\phi}(0) \ket{\Omega} \Rightarrow \hat{\phi}(x) \ket{\Omega} = \hat{\phi}(0) \ket{\Omega} \text{ for all } x \in \mink.
    \end{equation}
    Moreover if $\ket{\Omega}$ is translation-invariant then $U(x,e) \hat{\phi}(y) \ket{\Omega} = \hat{\phi}(x+y) \ket{\Omega} = \hat{\phi}(0) \ket{\Omega}$ so $\hat{\phi}(x) \ket{\Omega}$ is translation-invariant for all $x \in \mink$ so by uniqueness the second result follows.
\end{proof}

Here, we avoid the theorem: $\hat{\phi} \mapsto \hat{\phi}_\Lambda$ which need not be equal to $\hat{\phi}$ unless $\phi \in \bhs^{\Lup}$. We are however not interested in such operators: they act trivially on the vacuum, and so lead to theories with constant 2-point expectation values. Likewise, relational local quantum fields are not pointwise Poincaré covariant as $\omega \mapsto \omega \cdot (a,\Lambda)^{-1}$, and even then the covariance would take place under an integral.

    \section{Proofs omitted from the main text}\label{proofs}

    \subsection{Proof of Prop. \ref{prop: Gamma is partial trace}}\label{proof 2.5}
    
    \begin{proof}
    Let $\{e_i\}_{i \in I} \subset \his$ and $\{f_j\}_{j \in J} \subset \hir$ be orthonormal bases for $\his$ and $\hir$, respectively, where $I$ and $J$ are countable index sets, so that $\{e_i \otimes f_j\}_{(i,j) \in I \times J} \subset \his \otimes \hir$ is an orthonormal basis for $\his\otimes\hir$. Then $\forall \rho \in \systemstate$,
    \begin{equation}\begin{aligned}
        \Tr[\rho \Gamma_\omega(\mathcal{O})] &= \Tr[(\rho \otimes \omega) \mathcal{O}] \\
        &= \Tr[(\rho \otimes \mathbb{1}_{\bhr})(\mathbb{1}_{\bhs} \otimes \omega) \mathcal{O}] \\
        &= \Tr[\mathbb{1}_{\bhs} \otimes \mathbb{1}_{\bhr} (\rho \otimes \mathbb{1}_{\bhr}) \mathbb{1}_{\bhs} \otimes \mathbb{1}_{\bhr}(\mathbb{1}_{\bhs} \otimes \omega) \mathcal{O}] \\
        &= \sum_{(i,j) \in I \times J} \sum_{(k,l) \in I \times J} \Tr[\ket{e_i \otimes f_j}\bra{e_i \otimes f_j} \rho \otimes \mathbb{1}_{\bhr} \ket{e_k \otimes f_l} \bra{e_k \otimes f_l} (\mathbb{1}_{\bhs} \otimes \omega) \mathcal{O}] \\
        &= \sum_{(i,j) \in I \times J} \sum_{(k,l) \in I \times J} \bra{e_i \otimes f_j} \rho \otimes \mathbb{1}_{\bhr} \ket{e_k \otimes f_l} \bra{e_k \otimes f_l} (\mathbb{1}_{\bhs} \otimes \omega) \mathcal{O} \ket{e_i \otimes f_j} \\
        &= \sum_{(i,j) \in I \times J} \sum_{(k,l) \in I \times J} \bra{e_i} \rho \ket{e_k} \delta_{jl} \bra{e_k \otimes f_l} (\mathbb{1}_{\bhs} \otimes \omega) \mathcal{O} \ket{e_i \otimes f_l} \\
        &= \sum_{(i,j) \in I \times J} \sum_{k\in I} \bra{e_i} \rho \ket{e_k} \bra{e_k \otimes f_j} (\mathbb{1}_{\bhs} \otimes \omega) \mathcal{O} \ket{e_i \otimes f_j} \\
        &= \sum_{i,k \in I} \bra{e_i} \rho \ket{e_k} \bra{e_k} \Tr_{\hir}[(\mathbb{1}_{\bhs} \otimes \omega) \mathcal{O}] \ket{e_i} \\
        &= \sum_{i,k \in I} \Tr[\ket{e_i}\bra{e_i} \rho \ket{e_k}\bra{e_k} \Tr_{\hir}[ (\mathbb{1}_{\bhs} \otimes \omega) \mathcal{O}]] \\
        &= \Tr[\rho \Tr_{\hir}[ (\mathbb{1}_{\bhs} \otimes \omega) \mathcal{O}]]
    \end{aligned}\end{equation}
    which holds for all $\rho \in \systemstate$.
\end{proof}

\subsection{Proof of Lemma \ref{lem:nu covariance}}\label{proof 3.3}

\begin{proof}
	First, by covariance, we have $\mu^{\E_\R}_{\omega \cdot (a,\Lambda)}(X) = \mu^{\E_\R}_\omega((a,\Lambda) \cdot X)$. By the definition of the conditional probability measure, for any $\U \in \Bor(\mink)$ and $\tilde{\Lambda} \in \Bor(\Lup)$, we have
	\begin{equation}
	\int_\U \nu^{\E_\R}_{\omega \cdot (a,\Lambda)}(\tilde{\Lambda} \mid x) \, d\mu^{\F_\R}_{\omega \cdot (a,\Lambda)}(x) = \mu^{\E_\R}_{\omega \cdot (a,\Lambda)}(\U \times \tilde{\Lambda}) = \mu^{\E_\R}_\omega((a,\Lambda) \cdot (\U \times \tilde{\Lambda})) \, .
	\end{equation}
The action of $(a,\Lambda)$ on the set $\U \times \tilde{\Lambda}$ is $(a,\Lambda) \cdot (\U \times \tilde{\Lambda}) = \{(\Lambda x' + a, \Lambda \lambda') \mid x' \in \U, \lambda' \in \tilde{\Lambda} \}$. This is the product set $(\Lambda X + a) \times (\Lambda A)$. Thus,
\begin{equation} \label{eq:identity_start}
 \int_X \nu^{\E_\R}_{\omega \cdot (a,\Lambda)}(A \mid x) \, d\mu^{\F_\R}_{\omega \cdot (a,\Lambda)}(x) = \mu^{\E_\R}_\omega((\Lambda \U + a) \times (\Lambda \tilde{\Lambda})) \, .
\end{equation}
We can express the right-hand side of Eqn. \eqref{eq:identity_start} using the disintegration of $\mu^{\E_\R}_\omega$. Let $\V = \Lambda \U + a$ and $B = \Lambda \tilde{\Lambda}$:
\begin{equation}
\mu^{\E_\R}_\omega(\V \times B) = \int_\V \nu^{\E_\R}_\omega(B \mid y) \, d\mu^{\F_\R}_\omega(y) \, .
\end{equation}
Substituting back $\V$ and $B$, we get:
\begin{equation} \label{eq:rhs_disintegrated}
\mu^{\E_\R}_\omega((\Lambda \U + a) \times (\Lambda \tilde{\Lambda})) = \int_{\Lambda \U + a} \nu^{\E_\R}_\omega(\Lambda \tilde{\Lambda} \mid y) \, d\mu^{\F_\R}_\omega(y).
\end{equation}
We now have the identity
\begin{equation}
\int_\U \nu^{\E_\R}_{\omega \cdot (a,\Lambda)}(\tilde{\Lambda} \mid x) \, d\mu^{\F_\R}_{\omega \cdot (a,\Lambda)}(x) = \int_{\Lambda \U + a} \nu^{\E_\R}_\omega(\Lambda \tilde{\Lambda} \mid y) \, d\mu^{\F_\R}_\omega(y).
\end{equation}
To compare the integrands, we must make the integration domain and measure the same on both sides. By covariance, $\mu^{\F_\R}_{\omega \cdot (a,\Lambda)}(Z) = \mu^{\F_\R}_\omega(\Lambda Z + a)$ for any $Z \in \Bor(\mink)$. This implies the following change of variables formula for any integrable function $h: \mink \to \mathbb{R}$:
\begin{equation}
\int_{\mink} h(y) \, d\mu^{\F_\R}_\omega(y) = \int_{\mink} h(\Lambda x + a) \, d\mu^{\F_\R}_{\omega \cdot (a,\Lambda)}(x)
\end{equation}
Let $h(y) = \mathbb{I}_{\Lambda \U + a}(y) \nu^{\E_\R}_\omega(\Lambda A \mid y)$, where $\mathbb{I}$ is the indicator function.
\begin{align*}
\int_{\Lambda \U + a} \nu^{\E_\R}_\omega(\Lambda A \mid y) \, d\mu^{\F_\R}_\omega(y) &= \int_{\mink} h(y) \, d\mu^{\F_\R}_\omega(y) \\
&= \int_{\mink} h(\Lambda x + a) \, d\mu^{\F_\R}_{\omega \cdot (a,\Lambda)}(x) \\
&= \int_{\mink} \mathbb{I}_{\Lambda \U + a}(\Lambda x + a) \, \nu^{\E_\R}_\omega(\Lambda A \mid \Lambda x + a) \, d\mu^{\F_\R}_{\omega \cdot (a,\Lambda)}(x).
\end{align*}
The indicator function $\mathbb{I}_{\Lambda \U + a}(\Lambda x + a)$ is 1 iff $\Lambda x + a \in \Lambda \U + a$, which simplifies to $x \in \U$. So the right-hand side becomes
\begin{equation}
\int_\U \nu^{\E_\R}_\omega(\Lambda \tilde{\Lambda} \mid \Lambda x + a) \, d\mu^{\F_\R}_{\omega \cdot (a,\Lambda)}(x) \, .
\end{equation}
Comparing this with the left-hand side of Eqn. \eqref{eq:identity_start}, we have
\begin{equation}
\int_\U \nu^{\E_\R}_{\omega \cdot (a,\Lambda)}(\tilde{\Lambda} \mid x) \, d\mu^{\F_\R}_{\omega \cdot (a,\Lambda)}(x) = \int_\U \nu^{\E_\R}_\omega(\Lambda \tilde{\Lambda} \mid \Lambda x + a) \, d\mu^{\F_\R}_{\omega \cdot (a,\Lambda)}(x).
\end{equation}
Since this equality holds for all Borel sets $\U \in \Bor(\mink)$, the integrands must be equal for $\mu^{\F_\R}_{\omega \cdot (a,\Lambda)}$-almost every $x \in \mink$. Therefore,
\begin{equation}
\nu^{\E_\R}_{\omega \cdot (a,\Lambda)}(\tilde{\Lambda} \mid x) = \nu^{\E_\R}_\omega(\Lambda \tilde{\Lambda} \mid \Lambda x + a) \quad \text{for } \mu^{\F_\R}_{\omega \cdot (a,\Lambda)}\text{-a.e. } x
\end{equation}
which concludes the proof.
\end{proof}

\subsection{Proof of Lemma \ref{lem:causality and covariance}} \label{proof causality and covariance}

\begin{proof}
    Notice first that for any $A,B \in \bhs$ and $(a,\Lambda) \in \Poincup$ we have:
    \begin{equation}\begin{aligned}
        (a,\Lambda) \cdot (AB) &= U_S(a,\Lambda)ABU_S(a,\Lambda)^\dagger\\ &= U_S(a,\Lambda)AU_S(a,\Lambda)^\dagger U_S(a,\Lambda)BU_S(a,\Lambda)^\dagger\\
        &= ((a,\Lambda) \cdot A)((a,\Lambda) \cdot B) \,,
    \end{aligned}\end{equation}
    so we have $(a,\Lambda)\cdot [A,B] = [(a,\Lambda)\cdot A,(a,\Lambda)\cdot B]$. Now since for $C \in \bhs$ and any $(a,\Lambda) \in \Poincup$, $C = 0$ iff $(a,\Lambda) \cdot C = 0$, we can conclude that
    \begin{equation}
        \comm{A}{B} = \comm{A^\dagger}{B} = 0 \; \Longleftrightarrow \;\comm{(a,\Lambda) \cdot A}{(a,\Lambda) \cdot B} = \comm{(a,\Lambda) \cdot A^\dagger}{(a,\Lambda) \cdot B} = 0 \, .
    \end{equation}
    By definition, $\mathcal{O}_\S$ is $(\operationalstate,\sigma)$-causal iff $\forall \phi_1,\phi_2 \in \mathcal{O}_\S$ we have
    \begin{align}
        \forall \omega_1,\omega_2 \in \operationalstate, \, \Bigg(\omega_1 \indepERsigma \omega_2 \; \Longrightarrow& \; \comm{\hat{\Phi}_1^\R(\omega_1)}{\hat{\Phi}_2^\R(\omega_2)} = 
        \comm{\hat{\Phi}_1^\R(\omega_1)^\dagger}{\hat{\Phi}_2^\R(\omega_2)} = 0\Bigg) \,,
    \end{align}
    which is equivalent, for any $(a,\Lambda) \in \Poincup$, to
    \begin{align}
        \forall \omega_1,\omega_2 \in \operationalstate, \, \Bigg(\omega_1 \indepERsigma \omega_2 \; \Longrightarrow& \; \comm{(a,\Lambda)^{-1} \cdot \hat{\Phi}_1^\R(\omega_1)}{(a,\Lambda)^{-1}\cdot \hat{\Phi}_2^\R(\omega_2)} \nonumber \\ & \qquad= 
        \comm{(a,\Lambda)^{-1} \cdot \hat{\Phi}_1^\R(\omega_1)^\dagger}{(a,\Lambda)^{-1} \cdot \hat{\Phi}_2^\R(\omega_2)} = 0\Bigg)\,.
    \end{align}
    
This condition in turn, by covariance (Thm. \ref{thm:full covariance}), is equivalent to
    
    \begin{align} 
        \forall \omega_1,\omega_2 \in \operationalstate, \, \Bigg(\omega_1 \indepERsigma \omega_2 \; \Longrightarrow& \; \comm{\hat{\Phi}_1^\R(\omega_1 \cdot (a,\Lambda))}{\hat{\Phi}_2^\R(\omega_2 \cdot (a,\Lambda))} \nonumber \\ & \qquad= 
        \comm{\hat{\Phi}_1^\R(\omega_1 \cdot (a,\Lambda))^\dagger}{\hat{\Phi}_2^\R(\omega_2 \cdot (a,\Lambda))} = 0 \Bigg) \, .
    \end{align}
    Now notice that for all $\omega_1,\omega_2 \in \operationalstate$, $\sigma \geq 0$ and $(a,\Lambda) \in \Poincup$, we have
    \[
    \omega_1 \indepERsigma \omega_2 \text{ iff } \omega_1 \cdot (a,\Lambda) \indepERsigma \omega_2 \cdot (a,\Lambda) \,.
    \]
    Indeed, Poincaré transformations cannot spoil $\sigma$-spacelike separation since they preserve the causal structure: for $U \indep V \in \mink$ the quantity $\sup \{(x-y)^2 \mid x \in U, y \in V\}$ is invariant under Poincaré transformations. Thus, $\mathcal{O}_\S$ is $(\operationalstate,\sigma)$-causal iff $\forall \phi_1,\phi_2 \in \mathcal{O}_\S$ and arbitrary $(a,\Lambda) \in \Poincup$ we have
    \begin{align}
        \forall \omega_1,\omega_2 \in \operationalstate, \, \Bigg(\omega_1 \cdot (a,\Lambda) \indepERsigma \omega_2 \cdot (a,\Lambda) \; \Longrightarrow& \; \comm{\hat{\Phi}_1^\R(\omega_1 \cdot (a,\Lambda))}{\hat{\Phi}_2^\R(\omega_2 \cdot (a,\Lambda))} \nonumber \\ & \qquad= 
        \comm{\hat{\Phi}_1^\R(\omega_1 \cdot (a,\Lambda))^\dagger}{\hat{\Phi}_2^\R(\omega_2 \cdot (a,\Lambda))} = 0\Bigg) \, .
    \end{align}
    Finally, setting $\tilde{\omega}_i = \omega_i \cdot (a,\Lambda) \in \operationalstate \cdot (a,\Lambda)$ we get that $\mathcal{O}_\S$ is $(\operationalstate,\sigma)$-causal iff $\forall \phi_1,\phi_2 \in \mathcal{O}_\S$ and arbitrary $(a,\Lambda) \in \Poincup$ we have
    \begin{align}
        \forall \tilde{\omega}_1,\tilde{\omega}_2 \in \operationalstate \cdot (a,\Lambda), \, \Bigg(\tilde{\omega}_1 \indepERsigma \tilde{\omega}_2 \; \Longrightarrow& \; \comm{\hat{\Phi}_1^\R(\tilde{\omega}_1)}{\hat{\Phi}_2^\R(\tilde{\omega}_2)} = 
        \comm{\hat{\Phi}_1^\R(\tilde{\omega}_1)^\dagger}{\hat{\Phi}_2^\R(\tilde{\omega}_2)} = 0\Bigg) \,,
    \end{align} 
    which concludes the proof.
\end{proof}

\subsection{Proof of Theorem \ref{thm:Weyl approximate microcausality}} \label{proof approximate commutativity}

\begin{proof}
    Let $K_{\Lup} := \pi_{\mathcal{\Lup}}(K)$ and $K_\mink := \pi_{\mathcal{\mink}}(K)$, where $\pi_\mink : F \ni (x,\lambda) \mapsto x \in \mink$ and $\pi_{\mathcal{\Lup}} : F \ni (x,\lambda) \mapsto \lambda \in \mathcal{\Lup}$. Let $C(K) := \sup_{\Lambda \in K_{\Lup}} \norm{\Lambda}_{op,\mathbb{E}}$, where for all $z \in \mink$,
    \begin{equation}
        \norm{\Lambda}_{op,\mathbb{E}} := \sup_{z \neq 0} \frac{\norm{\Lambda z}_{\mathbb{E}}}{\norm{z}_\mathbb{E}}
    \end{equation}
    is the Euclidean operator norm of $\Lambda \in \Lup$, and where we have fixed inertial coordinates on $\mink$ and the Euclidean norm is given by
$\|z\|_{\mathbb E}=\sqrt{(z^0)^2+ \sum_{i=1}^{d-1} (z^i)^2}$. By compactness, $C(K) < \infty$. Hence, for all $\Lambda \in K_{\Lup}$ and all $z \in \mink$,
    \begin{equation}
        \norm{\Lambda z}_{\mathbb{E}} \leq C(K) \norm{z}_{\mathbb{E}} \, .
    \end{equation}
    Let 
    \begin{equation}
        M_\sigma(K) := \{S = x_1-x_2 \mid x_i \in K_\mink, \eta(S,S) < -\sigma\} \subset \mink \, .
    \end{equation}
    If $M_\sigma(K) = \varnothing$, then $\sigma$-spacelike pairs never occur among points in $K_\mink$; thus, the $\sigma$-microcausality statement is vacuous and any choice of $f_{\sigma,K}$ works. Otherwise, we can define
    \begin{equation}
        R_\sigma(K):=\sup\{\ \norm{S}_{\mathbb E}\ :\ S\in M_\sigma(K)\ \}
    \end{equation}
    such that $0 < R_\sigma(K) < \infty$ (since $K_\mink$ is compact). Let $F:M_\sigma(K) \times \mink \ni (S,y)\mapsto \eta(S+y,S+y)$. For $S\in M_\sigma(K)$ and $y\in \mink$,
\begin{equation}
F(S,y)=\eta(S+y,S+y)=\eta(S,S)+2\eta(S,y)+\eta(y,y) \, .
\label{eq:F-expand}
\end{equation}   
    
    Let $J=\mathrm{diag}(1,-1,\dots,-1)$ so that $\eta(u,v)=\langle Ju,v\rangle_{\mathbb E}$, where
$\langle\cdot,\cdot\rangle_{\mathbb E}$ is the Euclidean inner product.
Since $J$ is orthogonal for $\|\cdot\|_{\mathbb E}$, for all $u,v\in \mink$,
\begin{equation}
|\eta(u,v)|=|\langle Ju,v\rangle_{\mathbb E}|\leq \|Ju\|_{\mathbb E}\,\norm{v}_{\mathbb E}=\norm{u}_{\mathbb E}\,\norm{v}_{\mathbb E},
\qquad
\eta(v,v)\leq \norm{v}_{\mathbb E}^2.
\label{eq:eta-bounds}
\end{equation}
Using \eqref{eq:eta-bounds} and $\norm{S}_{\mathbb E}\leq R_\sigma(K)$,
\begin{equation}
F(S,y)\leq -\sigma+2\norm{S}_{\mathbb E}\,\|y\|_{\mathbb E}+\|y\|_{\mathbb E}^2
\leq -\sigma+2R_\sigma(K)\,\|y\|_{\mathbb E}+\|y\|_{\mathbb E}^2 \, .
\label{eq:F-bound}
\end{equation}

Choose
\begin{equation}
r_\sigma(K):=\min\!\Big\{\frac{\sigma}{4R_\sigma(K)}\ ,\ \frac{\sqrt{\sigma}}{2}\Big\}>0.
\label{eq:r-def}
\end{equation}
If $\|y\|_{\mathbb E}\leq r_\sigma(K)$, then $2R_\sigma(K)\,\|y\|_{\mathbb E}\leq \sigma/2$ and $\|y\|_{\mathbb E}^2\leq \sigma/4$, hence by \eqref{eq:F-bound},
\begin{equation}
F(S,y)\leq -\frac{\sigma}{4}<0\qquad\text{for all }S\in M_\sigma(K).
\label{eq:spacelike-stability}
\end{equation}

Next set
\begin{equation}
\varepsilon_\sigma(K):=\frac{r_\sigma(K)}{2\,C(K)} > 0,
\label{eq:eps-def}
\end{equation}
and let
\begin{equation}
    \psi : \mathbb{R} \ni x \mapsto \begin{cases}
        e^{\frac{1}{x^2-1}} \qquad \text{ if } \abs{x} < 1 \\
        0 \qquad \quad \; \; \, \,\text{ if } \abs{x} \geq 1
    \end{cases} \; \in \mathbb{R} \; .
\end{equation}
be a standard bump function. In the coordinate frame $x = \{x^0,\cdots,x^{d-1}\}$, let
\begin{equation}
    f_{\sigma,K} : \mink \ni x \mapsto \psi\left(\frac{\sum_{i=0}^{d-1} \abs{x^i}^2}{\varepsilon_\sigma(K)^2}\right) \in \mathbb{R} \, .
\end{equation}
Then $\supp f_{\sigma,K}\subset \{u:\norm{u}_{\mathbb E}< \varepsilon_\sigma(K)\}$.
Let $x_i\in K_{\mink}$, $\lambda_i\in K_{\Lup}$, and $u,v$ satisfy $\norm{u}_{\mathbb E},\norm{v}_{\mathbb E}\leq \varepsilon_\sigma(K)$.
Set $S:=x_1-x_2$, $y:=\lambda_1 u-\lambda_2 v$.
From \eqref{eq:eps-def} and the triangle inequality,
\begin{equation}
\|y\|_{\mathbb E}\leq \|\lambda_1 u\|_{\mathbb E}+\|\lambda_2 v\|_{\mathbb E}
\leq C(K)\norm{u}_{\mathbb E}+C(K)\norm{v}_{\mathbb E}
\leq 2C(K)\varepsilon_\sigma(K)=r_\sigma(K).
\label{eq:y-bound}
\end{equation}
Hence, if additionally $\eta(S,S)<-\sigma$ (i.e. $S\in M_\sigma(K)$), then by \eqref{eq:spacelike-stability}
\begin{equation}
\eta\big((x_1+\lambda_1 u)-(x_2+\lambda_2 v),\ (x_1+\lambda_1 u)-(x_2+\lambda_2 v)\big)
=F(S,y)<0.
\label{eq:pairwise-spacelike}
\end{equation}
Therefore all the $x_1+\lambda_1\supp f_{\sigma,K}$ and $x_2+\lambda_2\supp f_{\sigma,K}$ are spacelike separated. Since spacelike separation is frame-independent, this (frame-dependent definition of) $f_{\sigma,K}$ satisfies the desired properties; hence, by \eqref{eq:pairwise-spacelike} and the spacelike commutativity for Weyl operators,
\begin{equation}
    \comm{W_{\hat\phi}((x_1,\lambda_1)\cdot f_{\sigma,K})}{W_{\hat\phi}((x_2,\lambda_2)\cdot f_{\sigma,K})}=0 \qquad \forall x_i \in K_\mink, \, \lambda_i \in K_{\Lup}, \, x_1-x_2 \in M_\sigma(K) \, .
\end{equation}
Let $\eta := W_{\hat\phi}(f_{\sigma ,K}) \in \bhs$. Then since the integrand commutators of $\hat{\eta}^\R_{\omega_i}(x)$ vanish on the conditional supports we have that
\begin{equation}
    \comm{\hat{\eta}^{\,\R}_{\omega_1}(x_1)}{\hat{\eta}^\R_{\omega_2}(x_2)}=0 \qquad \forall x_i\in\supp\mu^{\F_\R}_{\omega_i}.
\end{equation}
    Since $W_{\hat{\phi}}(f_{\sigma,K})^\dagger = W_{\hat{\phi}}(-f_{\sigma,K})$ and $\supp (-f_{\sigma,K}) = \supp f_{\sigma,K}$, the same argument (with $f_{\sigma,K}$ replaced by $-f_{\sigma,K}$) proves
    \begin{equation}
        \comm{\hat{\eta}^{\,\R}_{\omega_1}(x_1)^\dagger}{\hat{\eta}^\R_{\omega_2}(x_2)}=0 \qquad \forall x_i\in\supp\mu^{\F_\R}_{\omega_i} ,
    \end{equation}
    which completes the proof.
\end{proof}

\subsection{Proofs of Section \ref{sec:Wightman functions}}\label{proofs exp val}

\subsubsection{Proof of Prop. \ref{prop:symmetry of vevs}}

\label{proof symmetry of vevs}

\begin{proof}
    We have
    \begin{equation}\begin{aligned}
        \Wscr_n^{(\Omega,\R)}[\omega_1,\cdots,\omega_n](\phi_1,\cdots,\phi_n) &= \Tr\left[\Omega \prod_{i=1}^n \hat{\Phi}_i^\R(\omega_i) \right] \\
        &= \Tr\left[\Omega \cdot (a,\Lambda) \prod_{i=1}^n \hat{\Phi}_i^\R(\omega_i) \right] \\
        &= \Tr\left[\Omega \Big(\prod_{i=1}^n (a,\Lambda) \cdot \hat{\Phi}_i^\R(\omega_i) \Big)\right] \\
        &\stackrel{\ref{thm:full covariance}}{=} \Tr\left[\Omega \prod_{i=1}^n \hat{\Phi}_i^\R((a,\Lambda) \cdot \omega_i) \right] \\
        &= \Wscr_n^{(\Omega,\R)}[(a,\Lambda).\omega_1,\cdots,(a,\Lambda).\omega_n](\phi_1,\cdots,\phi_n)
    \end{aligned}\end{equation}
    which holds for all $\phi_1,\cdots,\phi_n \in \bhs$. Moreover,
    \begin{equation}\begin{aligned}
        W_n^{(\Omega,\R)}[\omega_1 \cdot (a,\Lambda),\cdots,\omega_n \cdot (a,\Lambda);x_1,\cdots,x_n]&(\phi_1,\cdots,\phi_n) \\ &= \Tr\left[\Omega \prod_{i=1}^n (\hat{\phi}_i)^\R_{\omega_i \cdot (a,\Lambda)}(x_i) \right] \\
        &= \Tr\left[\Omega \cdot (a,\Lambda) \prod_{i=1}^n (\hat{\phi}_i)^\R_{\omega_i \cdot (a,\Lambda)}(x_i) \right] \\
        &= \Tr\left[\Omega \Big(\prod_{i=1}^n (a,\Lambda) \cdot (\hat{\phi}_i)^\R_{\omega_i \cdot (a,\Lambda)}(x_i)\Big) \right] \\
        &= \Tr\left[\Omega\prod_{i=1}^n (\hat{\phi}_i)^\R_{\omega_i}(\Lambda x_i+a) \right] \\
        &= W_n^{(\Omega,\R)}[\omega_1,\cdots,\omega_n;\Lambda x_1+a,\cdots,\Lambda x_n+a](\phi_1,\cdots,\phi_n)
    \end{aligned}\end{equation}
    which holds for all $\phi_1,\cdots,\phi_n\in \bhs$, so the result follows.
\end{proof}

\subsubsection{Proof of Prop. \ref{prop:global orientation}}

\label{proof global orientation}

\begin{proof}
    For all $\phi_1,\cdots,\phi_n \in \bhs$, 
    \begin{equation}\begin{aligned}
        &W_n^{(\Omega,\R)}[\omega_1,\cdots,\omega_n;x_1,\cdots,x_n](\phi_1,\cdots,\phi_n) \\ &= \Tr\left[\Omega (\hat{\phi}_1)^\R_{\omega_1}(x_1) \cdots (\hat{\phi}_n)^\R_{\omega_n}(x_n) \right] \\
        &= \Tr\left[\Omega U_\S(x_1,e) (\hat{\phi}_1)^\R_{\omega_1}(0) U_\S(x_1-x_2,e)^\dagger \cdots U_\S(x_{n-1}-x_n,e)^\dagger (\hat{\phi}_n)^\R_{\omega_n}(0) U_\S(x_n,e)^\dagger \right] \\
        &= \Tr\left[\Omega U_\S(x_1 - x_n,e)(\hat{\phi}_1)^\R_{\omega_1}(0) U_\S(x_1-x_2,e)^\dagger \cdots U_\S(x_{n-1}-x_n,e)^\dagger (\hat{\phi}_n)^\R_{\omega_n}(0) \right] \\
        &= \Tr\left[\Omega U_\S((x_1 - x_2) + \cdots + (x_{n-1} - x_n),e)(\hat{\phi}_1)^\R_{\omega_1}(0) U_\S(x_1-x_2,e)^\dagger \cdots U_\S(x_{n-1}-x_n,e)^\dagger (\hat{\phi}_n)^\R_{\omega_n}(0) \right] \\
        &= \Tr\left[\Omega U_\S\left(\sum_{i=1}^{n-1} (x_i - x_{i+1}),e\right)(\hat{\phi}_1)^\R_{\omega_1}(0) \prod_{j=1}^{n-1} \left(U_\S(x_j-x_{j+1},e)^\dagger (\hat{\phi}_j)^\R_{\omega_j}(0)\right) \right] \\
        &= \Wbf_n^{(\Omega,\R)}[\omega_1,\cdots,\omega_n;x_1-x_2,\cdots,x_{n-1} - x_n](\phi_1,\cdots,\phi_n)
    \end{aligned}\end{equation}
    which concludes the first part of the proof. For the second part of the proof, we write $U(z) \equiv U_\S(z,e)$ for conciseness. For all $\phi_1,\phi_2 \in \bhs$,
\begin{equation}\begin{aligned}
    W_n^{(\Omega,\R)}[\omega_1,\omega_2;x,y](\phi_1,\phi_2) &= \Tr[\Omega (\hat{\phi}_1)^\R_{\omega_1}(x) (\hat{\phi}_2)^\R_{\omega_2}(y)] \\
    &= \Tr[\Omega (U(x) (\hat{\phi}_1)^\R_{\omega_1}(0) U(x)^\dagger) (U(y) (\hat{\phi}_2)^\R_{\omega_2}(0) U(y)^\dagger)] \\
    &= \Tr[(U(y)^\dagger \Omega U(y)) (U(-y) (\hat{\phi}_1)^\R_{\omega_1}(0) U(-y)^\dagger)(U(-x) (\hat{\phi}_2)^\R_{\omega_2}(0) U(-x)^\dagger)] \\
    &= \Tr[\Omega (\hat{\phi}_1)^\R_{\omega_1}(-y) (\hat{\phi}_2)^\R_{\omega_2}(-x)] \\
    &= W_n^{(\Omega,\R)}[\omega_1,\omega_2;-y,-x](\phi_1,\phi_2)
\end{aligned}\end{equation}
which holds for all $\phi_1,\phi_2 \in \bhs$.
\end{proof}

\subsubsection{Proof of Prop. \ref{prop:spectrum condition vevs}}

\label{proof spectrum condition vevs}

\begin{proof}
    For all $p_1,\cdots,p_n \in \mathbb{R}^d$ and all $\phi_1,\cdots,\phi_n \in \bhs$, we have
    \begin{equation}\begin{aligned}
        \tilde{W}_n^{(\Omega,\R)}&[\omega_1,\cdots,\omega_n;p_1,\cdots,p_n](\phi_1,\cdots,\phi_n) \\ &= \idotsint \exp\left(i \sum_{i=1}^n p_i x_i \right) W_n^{(\Omega,\R)}[\omega_1,\cdots,\omega_n;x_1,\cdots,x_n](\phi_1,\cdots,\phi_n) dx_1\cdots dx_n \\
        &= \idotsint \exp\left(i [p_1(x_1 - x_2) + (p_2 + p_3)(x_2 - x_3) + \cdots + (p_1 + \cdots + p_n)(x_{n-1} - x_n)\right) \nonumber \\
        \times \exp&\left(\left(i \sum_{j=1}^n p_j\right) x_n\right) W_{n}^{(\Omega,\R)}[\omega_1,\cdots,\omega_n;x_1 - x_2,\cdots,x_{n-1}-x_n](\phi_1,\cdots,\phi_n) \, dx_1 \cdots dx_n \\
        &= (2\pi)^d \delta\left(\sum_{j=1}^n p_j\right) \tilde{\Wbf}_{n}^{(\Omega,\R)}[\omega_1,\cdots,\omega_n;p_1,p_1+p_2,\cdots,p_1 + \cdots + p_{n-1}](\phi_1,\cdots,\phi_n)
    \end{aligned}\end{equation}
    which holds for all $\phi_1,\cdots,\phi_n \in \bhs$, which concludes the first part of the proof. For the second claim, we again follow \cite{streater_pct_1989} for this proof. One can write $U_\S(a,e)$ as an integral over momentum space:
        \begin{equation}
            U_\S(a,e) = \int_{\sigma} e^{i p \cdot a} d\E(p)
        \end{equation}
        where $\E$ is a PVM on momentum space and $\sigma$ is the associated spectrum on which $\E$ takes nonzero values. Thus a projection operator $\E(S)$ is defined for each sphere of momentum space and each set which can be obtained from spheres by a countable number of unions, intersections and complements. In terms of $\E$, the statement that a set $S$ is not in the physical energy-momentum spectrum is simply that $\E(S) = 0$ or, equivalently,
        \begin{equation}
            \int_{\mathbb{R}^d} \tilde{f}(p) d\E(p) = 0
        \end{equation}
        if $\supp \tilde{f} \subset S$. Defining
        \begin{equation}
            f(x) = \frac{1}{(2\pi)^{d/2}} \int_{\mathbb{R}^d} e^{-i p \cdot a} \tilde{f}(p) dp \, ,
        \end{equation}
        the above is equivalent to
        \begin{equation}
            \int_{\mathbb{R}^d} f(a) U(a,e) d^da = 0
        \end{equation}
        which implies that, for all $T \in \ths$,
        \begin{equation}
            \int_{\mathbb{R}^d} f(a) \Tr[T\, U(a,e)] \,d^d a = 0
        \end{equation}
        and, in particular,
        \begin{equation}
            \int_{\mathbb{R}^d} e^{ip\cdot a} \Tr[T \,U(a,e)] \, d^da = 0
        \end{equation}
        unless $p$ lies in the energy-momentum spectrum of the states. Thus, for
        \begin{equation}
            T = (\hat{\phi}_{j+1})^\R_{\omega_{j+1}}(x_{j+1}) \cdots (\hat{\phi}_n)^\R_{\omega_n}(x_{n}) \Omega (\hat{\phi}_1)^\R_{\omega_1}(x_1) \cdots (\hat{\phi}_j)^\R_{\omega_j}(x_{j})
        \end{equation}
        where $\phi_1,\cdots,\phi_n \in \bhs$ and $x_1,\cdots,x_n \in \mink$,
        \begin{equation}
            \int_{\mathbb{R}^d} e^{i p \cdot a} \Tr[\Omega \, (\hat{\phi}_1)^\R_{\omega_1}(x_1) \cdots (\hat{\phi}_j)^\R_{\omega_j}(x_{j}) U(a,e)^\dagger (\hat{\phi}_{j+1})^\R_{\omega_{j+1}}(x_{j+1}) \cdots (\hat{\phi}_n)^\R_{\omega_n}(x_{n})] \, d^d a = 0
        \end{equation}
        by the cyclic property of the trace for all $j \in \{1,\cdots,n\}$, which implies
        \begin{equation}
            \int_{\mathbb{R}^d} e^{i p \cdot a} \Wbf_n^{(\Omega,\R)}[\omega_1,\cdots,\omega_n;\xi_1,\cdots,\xi_{j-1},\xi_j + a,\xi_{j+1},\cdots,\xi_{n-1}](\phi_1,\cdots,\phi_n) d^d a = 0 
        \end{equation}
        for all $j \in \{1,\cdots,n\}$ for $p$ not in the physical spectrum, which holds for all $\phi_1,\cdots,\phi_n \in \bhs$, i.e.
        \begin{equation}
            \tilde{\Wbf}_n^{(\Omega,\R)}[\omega_1,\cdots,\omega_n;q_1,\cdots,q_{n-1}] = 0
        \end{equation}
        unless each $q_i$ lies in the physical spectrum. 
\end{proof}

\subsubsection{Proof of Prop. \ref{prop:Hermiticity Wightman functions}}

\label{proof Hermiticity Wightman functions}

\begin{proof}
    We have
        \begin{multline}
            W_n^{(\Omega,\R)}[\omega_1,\cdots,\omega_n;x_1,\cdots,x_n](\phi_1,\cdots,\phi_n) = \Tr\left[\Omega \prod_{i=1}^n (\hat{\phi}_i)^\R_{\omega_i}(x_i)\right] = \overline{\Tr\left[\left(\Omega \prod_{i=1}^n (\hat{\phi}_i)^\R_{\omega_i}(x_i)\right)^\dagger\right]} \\ = \overline{\Tr\left[\Omega\prod_{i=1}^n (\hat{\phi}_{n-i-1})^\R_{\omega_{n-i-1}}(x_{n-i-1})^\dagger\right]} = \overline{W_n^{(\Omega,\R)}[\omega_n,\cdots,\omega_1;x_n,\cdots,x_1](\phi_n^\dagger,\cdots,\phi_1^\dagger)}
        \end{multline}
        and
        \begin{multline}
            \Wbf_n^{(\Omega,\R)}[\omega_1,\cdots,\omega_n;\xi_1,\cdots,\xi_{n-1}](\phi_1,\cdots,\phi_n) = W_{n}^{(\Omega,\R)}[\omega_1,\cdots,\omega_n;x_1,\cdots,x_n](\phi_1,\cdots,\phi_n) \\ = \overline{W_n^{(\Omega,\R)}[\omega_n,\cdots,\omega_1;x_n,\cdots,x_1](\phi_n^\dagger,\cdots,\phi_1^\dagger)} = \overline{\Wbf^{(\Omega,\R)}_n[\omega_n,\cdots,\omega_1;-\xi_{n-1},\cdots,-\xi_1](\phi_n^\dagger,\cdots,\phi_1^\dagger)} \, .
        \end{multline}
        Thus,
        \begin{multline}
            \Wscr_n^{(\Omega,\R)}[\omega_1,\cdots,\omega_n](\phi_1,\cdots,\phi_n) = \idotsint_{\mink^n} W_n^{(\Omega,\R)}[x_1,.\cdots,x_n](\phi_1,\cdots,\phi_n) \, d\mu_{\omega_1}^{\F_\R}(x_1) \cdots d\mu^{\F_\R}_{\omega_n}(x_n) \\
            \stackrel{\ref{thm:Fubini-Tonelli}}{=} \idotsint_{\mink^n} \overline{ W_n^{(\Omega,\R)}[x_n,.\cdots,x_1](\phi_n^\dagger,\cdots,\phi_1^\dagger)} d\mu^{\F_\R}_{\omega_n}(x_n) \cdots d\mu^{\F_\R}_{\omega_n}(x_1) \\ \stackrel{*}{=} \overline{\idotsint_{\mink^n} W_n^{(\Omega,\R)}[x_n,.\cdots,x_1](\phi_n^\dagger,\cdots,\phi_1^\dagger) d\mu^{\F_\R}_{\omega_n}(x_n) \cdots d\mu^{\F_\R}_{\omega_n}(x_1) } 
            \\ = \overline{\W_n^{(\Omega,\R)}[\omega_n,\cdots,\omega_1](\phi_n^\dagger,\cdots,\phi_1^\dagger)}
        \end{multline}
        where the equality (*) holds as the measures $\mu^{\F_\R}_{\omega_i}$ are real for all $i \in \{1,\cdots,n\}$.
\end{proof}

\subsubsection{Proof of Prop. \ref{prop:local commutativity vevs}}

\label{proof of local commutativity vevs}

\begin{proof}
    If
        \begin{enumerate}
        \item $\mathcal{O}_\S$ is $(\operationalstate,\sigma)$-causal and $\omega_j \indepERsigma \omega_{j+1}$ and $\phi_j,\phi_{j+1} \in \mathcal{O}_\S$, then
        \begin{equation}\begin{aligned}
            \Wscr_n^{(\Omega,\R)}&[\omega_1,\cdots,\omega_n](\phi_1,\cdots,\phi_j,\phi_{j+1},\cdots,\phi_n) \\ &= \Tr\left[\Omega \left(\prod_{i=1}^{j-1} \hat{\Phi}_i^\R(\omega_i)\right) \hat{\Phi}_j^\R(\omega_j) \hat{\Phi}_{j+1}^\R(\omega_{j+1}) \left(\prod_{k=j+2}^n \hat{\Phi}_k^\R(\omega_k)\right)\right] \\
            &=\Tr\left[\Omega \left(\prod_{i=1}^{j-1} \hat{\Phi}_i^\R(\omega_i)\right) \hat{\Phi}_{j+1}^\R(\omega_{j+1})\hat{\Phi}_{j}^\R(\omega_{j})  \left(\prod_{k=j+2}^n \hat{\Phi}_k^\R(\omega_k)\right)\right] \\
            &= \Wscr_n^{(\Omega,\R)}[\omega_1,\cdots,\omega_{j-1},\omega_{j+1},\omega_j,\omega_{j+2},\cdots,\omega_n](\phi_1,\cdots,\phi_{j-1},\phi_{j+1},\phi_j,\phi_{j+2},\cdots,\phi_n) \, .
        \end{aligned}\end{equation}
        \item $\mathcal{O}_\S$ is weakly $(\operationalstate,\sigma)$-microcausal and $\omega_j \indepERsigma \omega_{j+1}$ and $\phi_j,\phi_{j+1} \in \mathcal{O}_\S$, then
        \begin{equation}\begin{aligned}
            W_n^{(\Omega,\R)}&[\omega_1,\cdots,\omega_n;x_1,\cdots,x_n](\phi_1,\cdots,\phi_j,\phi_{j+1},\cdots,\phi_n) \\ &= \Tr\left[\Omega \left(\prod_{i=1}^{j-1} (\hat{\phi}_i)^\R_{\omega_i}(x_i)\right) (\hat{\phi}_j)^\R_{\omega_j}(x_j) (\hat{\phi}_{j+1})^\R_{\omega_{j+1}}(x_{j+1}) \left(\prod_{k=j+2}^n (\hat{\phi}_k)^\R_{\omega_k}(x_k)\right)\right] \\
            &=\Tr\left[\Omega \left(\prod_{i=1}^{j-1} (\hat{\phi}_i)^\R_{\omega_i}(x_i)\right) (\hat{\phi}_{j+1})^\R_{\omega_{j+1}}(x_{j+1})(\hat{\phi}_{j})^\R_{\omega_{j}}(x_j)  \left(\prod_{k=j+2}^n (\hat{\phi}_k)^\R_{\omega_k}(x_k)\right)\right] \\
            &= W_n^{(\Omega,\R)}[\omega_1,\cdots,\omega_{j-1},\omega_{j+1},\omega_j,\omega_{j+2},\cdots,\omega_n;x_1,\cdots,x_{j-1},x_{j+1},x_j,x_{j+2},\cdots,x_n] \nonumber \\& \qquad \qquad \qquad \qquad \qquad \qquad \qquad \qquad (\phi_1,\cdots,\phi_{j-1},\phi_{j+1},\phi_j,\phi_{j+2},\cdots,\phi_n) \, .
        \end{aligned}\end{equation}
        \item $\mathcal{O}_\S$ is strongly $(\operationalstate,\sigma)$-microcausal and $x_j \indepsigma x_{j+1}$ and $\phi_j,\phi_{j+1} \in \mathcal{O}_\S$, then
        \begin{equation}\begin{aligned}
            W_n^{(\Omega,\R)}&[\omega_1,\cdots,\omega_n;x_1,\cdots,x_n](\phi_1,\cdots,\phi_j,\phi_{j+1},\cdots,\phi_n) \\ &= \Tr\left[\Omega \left(\prod_{i=1}^{j-1} (\hat{\phi}_i)^\R_{\omega_i}(x_i)\right) (\hat{\phi}_j)^\R_{\omega_j}(x_j) (\hat{\phi}_{j+1})^\R_{\omega_{j+1}}(x_{j+1}) \left(\prod_{k=j+2}^n (\hat{\phi}_k)^\R_{\omega_k}(x_k)\right)\right] \\
            &=\Tr\left[\Omega \left(\prod_{i=1}^{j-1} (\hat{\phi}_i)^\R_{\omega_i}(x_i)\right) (\hat{\phi}_{j+1})^\R_{\omega_{j+1}}(x_{j+1})(\hat{\phi}_{j})^\R_{\omega_{j}}(x_j)  \left(\prod_{k=j+2}^n (\hat{\phi}_k)^\R_{\omega_k}(x_k)\right)\right] \\
            &= W_n^{(\Omega,\R)}[\omega_1,\cdots,\omega_{j-1},\omega_{j+1},\omega_j,\omega_{j+2},\cdots,\omega_n;x_1,\cdots,x_{j-1},x_{j+1},x_j,x_{j+2},\cdots,x_n] \nonumber \\& \qquad \qquad \qquad \qquad \qquad \qquad \qquad \qquad (\phi_1,\cdots,\phi_{j-1},\phi_{j+1},\phi_j,\phi_{j+2},\cdots,\phi_n) \, .
        \end{aligned}\end{equation}
        \end{enumerate}
\end{proof}

\subsubsection{Proof of Prop. \ref{prop:Positivity vevs}}

\label{proof of Positivity vevs}

\begin{proof}
    For all $(\phi_{lm})_{l \leq m = 1}^n \in \bhs$ and all $(\omega_{lm})_{l \leq m = 1}^n \in \bhs$, let $A_j := \prod_{p=1}^j \hat{\Phi}_{pj}^\R(\omega_{pj}) \in \bhs$. Then
        \begin{equation}
            \sum_{j,k = 1}^n \Wscr_{j+k}^{(\Omega,\R)}[\omega_{1j},\cdots,\omega_{jj},\omega_{1k},\cdots,\omega_{kk}](\phi_{jj}^\dagger,\cdots,\phi_{j1}^\dagger,\phi_{k1},\cdots,\phi_{kk}) = \sum_{j,k=1}^n \Tr[\Omega A_k^\dagger A_j] \, .
        \end{equation}
        We can write $\Omega \in \systemstate^{\Poincup}$ in terms of its eigendecomposition
        \begin{equation}
            \Omega = \sum_{\alpha=1}^\infty p_\alpha  \ket{\psi_\alpha}\bra{\psi_\alpha} \, . 
        \end{equation}
        Thus,
        \begin{equation}\begin{aligned}
            \sum_{j,k = 1}^n \Wscr_{j+k}^{(\Omega,\R)}[\omega_{1j},\cdots,\omega_{jj},\omega_{1k},\cdots,\omega_{kk}]&(\phi_{jj}^\dagger,\cdots,\phi_{j1}^\dagger,\phi_{k1},\cdots,\phi_{kk}) \\ &= \sum_{j,k=1}^n \Tr\left[\sum_{\alpha=1}^\infty p_\alpha \ket{\psi_\alpha}\bra{\psi_\alpha} A_k^\dagger A_j \right] \\
            &= \sum_{\alpha=1}^\infty p_\alpha \sum_{j,k=1}^n \expval{A_k^\dagger A_j}{\psi_\alpha} d\alpha
        \end{aligned}\end{equation}
        since the trace is continuous in the ultraweak topology. Let $\ket{v_\alpha} := \sum_{j=1}^n A_j \ket{\psi_\alpha}$. Then
        \begin{equation}
            \norm{v_\alpha}^2 = \braket{v_\alpha} = \braket{\sum_{k=0}^n A_k \psi_\alpha}{\sum_{j=0}^n A_j \psi_\alpha} = \sum_{j,k=1}^n \expval{A_k^\dagger A_j}{\psi_\alpha} = C \geq 0
        \end{equation}
        where $C \in \mathbb{R}^+$ as $A_j \in \bhs$. Thus,
        \begin{equation}
            \sum_{j,k = 1}^n \Wscr_{j+k}^{(\Omega,\R)}[\omega_{1j},\cdots,\omega_{jj},\omega_{1k},\cdots,\omega_{kk}](\phi_j^\dagger,\cdots,\phi_1^\dagger,\phi_1,\cdots,\phi_k) = \sum_{\alpha=1}^\infty p_\alpha C \geq 0
        \end{equation}
        as each $p_\alpha \geq 0$ by the positivity of $\Omega$.
\end{proof}

\subsection{Proof of Section \ref{subsec:cluster decomposition}}
\label{app:proofs cluster decomposition}

\subsubsection{Proof of Lem. \ref{lem:pure unique vacuum implies unique invariant vector}}

\label{proof pure unique vacuum implies unique invariant vector}

\begin{proof} Consider
    \begin{align}
        U(a,e) U(0,\Lambda(A)) \Omega U(0,\Lambda(A))^\dagger U(a,e)^\dagger &= U(0,\Lambda(A)) U(\Lambda(A)^{-1},e) \Omega U(\Lambda(A)^{-1},e)^\dagger U(0,\Lambda(A))^\dagger \\
        &= U(0,\Lambda(A)) \Omega U(0,\Lambda(A))^\dagger
    \end{align}
    i.e. $\Omega \in \systemstate^{\Lupd}$ i.e. $\Omega \in \vacstate$.

    If $\Omega$ is pure, we can write $\Omega = \ket{\Omega}\bra{\Omega}$ for some normalised $\ket{\Omega} \in \his$. From $a \cdot \Omega = \Omega$ we have $U(a,e) \ket{\Omega} = e^{i f(a)} \ket{\Omega}$ for some $f : \mink \to \mathbb{C}$. Since $U(\cdot,e)$ is an ultraweakly continuous unitary representation (thus strongly continuous) of $(T(1,d-1),+)$, $a \mapsto e^{i f(a)} = \braket{\Omega}{U(a,e) \Omega}$ is continuous. Moreover, $U(a+b,e) = U(a,e)U(b,e)$ so $e^{i f(a+b)} = e^{i (f(a) + f(b))}$ so since $f(0) = 0$ we have by continuity that $f(a+b) = f(a)+f(b)$ so $f$ is linear, i.e. $\exists q \in \mink$ such that $f(a) = q \cdot a$. Thus, $U(a,e) \ket{\Omega} = e^{i q \cdot a} \ket{\Omega}$ for some $q \in \mink$. But $\Omega = \Lambda(A) \cdot \Omega$ for all $A \in \Lupd$ so 
    \begin{equation}
        U(a,e) U(0,\Lambda(A)) \ket{\Omega} = U(0,\Lambda(A))U(\Lambda(A)^{-1}a,e) \ket{\Omega} = e^{i q \cdot (\Lambda(A)^{-1}a)} U(0,\Lambda) \ket{\Omega} = e^{i (\Lambda(A)q) \cdot a} U(0,\Lambda) \ket{\Omega}
    \end{equation}
    i.e. $U(0,\Lambda)\ket{\Omega}$ is a translation eigenvector with associated character $\Lambda(A) q$. But by uniqueness of $\Omega \in \vacstate$, $U(0,\Lambda(A)) \ket{\Omega}$ lies on the same ray as $\ket{\Omega}$, i.e. $q = \Lambda(A) q$ for all $A \in \Lupd$, i.e. $q=0$. Hence, $U(a,e) \ket{\Omega} = \ket{\Omega}$ for all $a \in T(1,d-1)$. The other direction of the equivalence is immediate.

\end{proof}

\subsubsection{Proof of Lem. \ref{lem:cluster}}

\label{proof cluster}

\begin{proof}
    \begin{enumerate}
        \item If $\his$ has the cluster property, then by Thm. \ref{thm:Cluster property HS} $\Omega$ is pure, is unique and is strongly translation-invariant. Moreover, take $\ket{\psi_0} \equiv \ket{\Omega}$, $\ket{\Phi} = \prod_{i=1}^j(\hat{\phi}_i)^\R_{\omega_i}(x_i)^\dagger \ket{\Omega}$ and $\ket{\Psi}=\prod_{i=j+1}^n (\hat{\phi}_i)^\R_{\omega_i}(x_i) \ket{\Omega}$ in Eqn. \eqref{eqn:cluster property HS}, and since
        \begin{multline}
            \left(\prod_{i=j+1}^n (\hat{\phi}_i)^\R_{(\lambda a, e) \cdot \omega_i}(x_i + \lambda a)\right) \ket{\Omega} = \left(\prod_{i=j+1}^n (\lambda a,e) \cdot (\hat{\phi}_i)^\R_{\omega_i}(x_i)\right)\ket{\Omega} \\ = (\lambda a,e) \cdot \left(\prod_{i=j+1}^n (\hat{\phi}_i)^\R_{\omega_i}(x_i)\right) \ket{\Omega} = U_\S(\lambda a,e)\prod_{i=j+1}^n (\hat{\phi}_i)^\R_{\omega_i}(x_i) \ket{\Omega} = U_\S(\lambda a,e) \ket{\Psi} \, ,
        \end{multline}
        the statement follows. \qed
        \item We have
        \begin{align}
            &\lim_{\lambda \to \infty} \Wscr_{n}^{(\Omega,\R)}[\omega_1,\cdots,\omega_j,(\lambda a,e) \cdot\omega_{j+1},\cdots,(\lambda a,e) \cdot \omega_n](\phi_1,\cdots,\phi_n) \nonumber \\ &= \lim_{\lambda \to \infty} \Wscr_{n}^{(\Omega,\R)}[\omega_1,\cdots,\omega_j,\omega_{j+1},\cdots,\omega_n]((\lambda a,e) \cdot \phi_1,\cdots,(\lambda a, e) \cdot \phi_n) \\
            &= \lim_{\lambda \to \infty} \idotsint_{\mink^n} \left(\prod_{i=1}^n d\mu^{\F_\R}_{\omega_i}(x) \right) \nonumber \\ &\qquad  \qquad W_{n}^{(\Omega,\R)}[\omega_1,\cdots,\omega_j,\omega_{j+1}, \cdots,\omega_n;x_1,\cdots,x_j,x_{j},\cdots,x_n]((\lambda a,e) \cdot \phi_1,\cdots,(\lambda a, e) \cdot \phi_n)  \\
            &=  \lim_{\lambda \to \infty} \idotsint_{\mink^n} \left(\prod_{i=1}^n d\mu^{\F_\R}_{\omega_i}(x) \right) \nonumber \\ & \quad W_{n}^{(\Omega,\R)}[\omega_1,\cdots,\omega_j,(\lambda a,e) \cdot \omega_{j+1}, \cdots,(\lambda a,e) \cdot\omega_n;x_1,\cdots,x_j,x_{j}+\lambda a,\cdots,x_n + \lambda a](\phi_1,\cdots,\phi_n)  \, .
        \end{align}
        Since the integrand is uniformly bounded, we can place the limit inside the integrals by Arzelà's bounded convergence theorem. Thus, assuming the strong relational cluster decomposition property, we have
        \begin{align}
            &\lim_{\lambda \to \infty} \Wscr_{n}^{(\Omega,\R)}[\omega_1,\cdots,\omega_j,(\lambda a,e) \cdot\omega_{j+1},\cdots,(\lambda a,e) \cdot \omega_n](\phi_1,\cdots,\phi_n) \\ &= \idotsint_{\mink^n}\left(\prod_{i=1}^n d\mu^{\F_\R}_{\omega_i}(x) \right) W_{j}^{(\Omega,\R)}[\omega_1,\cdots,\omega_j;x_1,\cdots,x_j](\phi_1,\cdots,\phi_j) \nonumber \\ & \qquad \qquad \qquad \qquad \qquad \qquad \qquad \qquad \cdot W_{n-j}^{(\Omega,\R)} [\omega_{j+1},\cdots,\omega_n;x_{j+1},\cdots,x_n](\phi_{j+1},\cdots,\phi_n) \\
            &= \Wscr_{j}^{(\Omega,\R)}[\omega_1,\cdots,\omega_j](\phi_1,\cdots,\phi_j) \cdot \Wscr_{n-j}^{(\Omega,\R)} [\omega_{j+1},\cdots,\omega_n](\phi_{j+1},\cdots,\phi_n)
        \end{align}
        which concludes the proof.
        \end{enumerate}
\end{proof}

\subsubsection{Proof of Thm. \ref{thm:Cluster properties}}

\label{proof Cluster properties}

\begin{proof} 
        We already know (from Lem. \ref{lem:cluster} and Thm. \ref{thm:Cluster property HS}) that $4. \Leftrightarrow 3. \Rightarrow 2. \Rightarrow 1.$ Thus, it is sufficient to show that $1. \Rightarrow 3.$, i.e. if $\ket{\Omega}$ is cyclic for $\mathscr{P}(\mathcal{O}_\S)_{\operationalstate}$ and $(\mathcal{O}_\S,\Omega,\operationalstate)$ satisfies the weak relational cluster decomposition property, then $\his$ has the cluster property. It is sufficient to show the cluster property for all $\ket{\Psi},\ket{\Phi} \in \hi$ with $\ket{\psi_0} \equiv \ket{\Omega}$. Since $\ket{\Omega}$ is cyclic for $\mathscr{P}(\mathcal{O}_{\S})_{\operationalstate}$, it follows that for all $\ket{\Psi} \in \his$ there exists a $k,l \in \mathbb{N}$ and $c_{ij} \in \mathbb{C}$ (without loss of generality take them to be nonzero) and $\phi_{ij} \in \mathcal{O}_\S$ and $\omega_{ij} \in \operationalstate$, $i=1,\cdots,k$ and $j = 1,\cdots,l$, such that $\sum_{i=1}^k \prod_{j=1}^l c_{ij} \hat{\Phi}^\R_{ij}(\omega_{ij}) \ket{\Omega} = \ket{\Psi}$. Likewise, for $\ket{\Psi} \in \his$ there exists a $\alpha,\beta \in \mathbb{N}$ and $d_{\mu \nu} \in \mathbb{C}$ (without loss of generality take them to be nonzero) and $\phi_{\mu \nu} \in \mathcal{O}_\S$ and $\omega_{\mu \nu} \in \operationalstate$, $\mu=1,\cdots,\alpha$ and $\nu = 1,\cdots,\beta$, such that $\sum_{\mu=1}^\alpha \prod_{\beta=1}^\nu d_{\mu\nu} \hat{\Phi}^\R_{\mu\nu}(\omega_{\mu\nu}) \ket{\Omega} = \ket{\Phi}$. Since $\mathcal{O}_\S$ is closed under adjoints $\exists \phi_{\mu\nu} \in \mathcal{O}_{\S}$ such that $\phi_{\mu\nu} = \psi_{\mu\nu}^\dagger$ for all $\mu,\nu$ and so since $\yen^\R_\omega$ is adjoint-preserving for all $\omega \in \framestate$, $\sum_{\mu=1}^\alpha \prod_{\nu=1}^\beta d_{\mu \nu} \hat{\Phi}^\R_{\mu\nu}(\omega_{\mu\nu})^\dagger \ket{\Omega} = \ket{\Phi}$. Then
        \begin{align}
            \lim_{\lambda \to \infty} \bra{\Phi}U_\S(\lambda a,e) \ket{\Psi} &= \lim_{\lambda \to \infty} \bra{\Omega} \left(\sum_{\mu=1}^\alpha \prod_{\nu=1}^\beta \overline{d_{\mu \nu}} \hat{\Phi}^\R_{\mu\nu}(\omega_{\mu\nu})\right) U_\S(\lambda a,e) \left(\sum_{i=1}^k \prod_{j=1}^l c_{ij} \hat{\Phi}^\R_{ij}(\omega_{ij})\right)\ket{\Omega} \\
            &= \lim_{\lambda \to \infty} \sum_{\mu=1}^\alpha \sum_{i=1}^k \bra{\Omega} \left( \prod_{\nu=1}^\beta \overline{d_{\mu \nu}} \hat{\Phi}^\R_{\mu\nu}(\omega_{\mu\nu})\right) (\lambda a,e) \cdot \left( \prod_{j=1}^l c_{ij} \hat{\Phi}^\R_{ij}(\omega_{ij})\right)\ket{\Omega} \\
            &= \lim_{\lambda \to \infty} \sum_{\mu=1}^\alpha \sum_{i=1}^k \bra{\Omega} \left( \prod_{\nu=1}^\beta \overline{d_{\mu \nu}} \hat{\Phi}^\R_{\mu\nu}(\omega_{\mu\nu})\right)  \left( \prod_{j=1}^l c_{ij} \hat{\Phi}^\R_{ij}((\lambda a,e)\cdot \omega_{ij})\right)\ket{\Omega} \, .
        \end{align}
        Let $\kappa_{\mu,\beta} := \prod_{\nu=1}^\beta \overline{d_{\mu\nu}}$ and $\eta_{i,l} := \prod_{j=1}^l c_{ij}$. Then
        \begin{align}
            \lim_{\lambda \to \infty} \bra{\Phi}U_\S(\lambda a,e) \ket{\Psi} &= \lim_{\lambda \to \infty} \sum_{\mu=1}^\alpha \sum_{i=1}^k \kappa_{\mu,\beta} \eta_{i,l} \bra{\Omega}\left(\prod_{\nu=1}^\beta \hat{\Phi}^\R_{\mu\nu}(\omega_{\mu\nu})\right) \left(\prod_{j=1}^l \hat{\Phi}_{ij}^\R((\lambda a,e) \cdot \omega_{ij})\right)\ket{\Omega} \\
            &= \lim_{\lambda \to \infty} \sum_{\mu=1}^\alpha \sum_{i=1}^k \kappa_{\mu,\beta} \eta_{i,l} \Wscr_{\beta+l}^{(\Omega,\R)}[\omega_{\mu 1},\cdots,\omega_{\mu \beta},(\lambda a,e) \cdot\omega_{i1},\cdots,(\lambda a,e) \cdot \omega_{il}] \nonumber \\ & \qquad \qquad \qquad \qquad \qquad \qquad \qquad(\phi_{\mu 1},\cdots,\phi_{\mu \beta},\phi_{i1},\cdots,\phi_{il}) \, .
        \end{align}
        Supposing $(\mathcal{O}_\S,\Omega,\operationalstate)$ satisfies the weak relational cluster decomposition property, this is
        \begin{align}
            \lim_{\lambda \to \infty} \bra{\Phi}U_\S(\lambda a,e) \ket{\Psi} &= \sum_{\mu=1}^\alpha \sum_{i=1}^k \kappa_{\mu,\beta} \eta_{i,l}\Wscr_{\beta}^{(\Omega,\R)}[\omega_{\mu 1},\cdots,\omega_{\mu \beta}](\phi_{\mu 1},\cdots,\phi_{\mu \beta}) \nonumber \\ & \qquad \qquad \qquad \qquad \qquad\cdot \Wscr_{l}^{(\Omega,\R)}[\omega_{i1},\cdots,\omega_{il}](\phi_{i1},\cdots,\phi_{il}) \\
            &= \sum_{\mu=1}^\alpha \sum_{i=1}^k \kappa_{\mu,\beta} \eta_{i,l} \expval{\left(\prod_{\nu=1}^\beta \hat{\Phi}^\R_{\mu\nu}(\omega_{\mu\nu})\right)}{\Omega} \expval{\left( \prod_{j=1}^l c_{ij} \hat{\Phi}^\R_{ij}(\omega_{ij})\right)}{\Omega} \\
            &= \expval{\left(\sum_{\mu=1}^\alpha \prod_{\nu=1}^\beta \overline{d_{\mu \nu}} \hat{\Phi}^\R_{\mu\nu}(\omega_{\mu\nu})\right)}{\Omega}\expval{\left(\sum_{i=1}^k \prod_{j=1}^l c_{ij} \hat{\Phi}^\R_{ij}(\omega_{ij})\right)}{\Omega} \\
            &= \braket{\Phi}{\Omega}\braket{\Omega}{\Psi}
        \end{align}
        which concludes the proof.
\end{proof}

\subsection{Proof of Lemma \ref{lem:time-ordered}}\label{proof 5.11}

\begin{proof}
    \begin{enumerate}
        \item This follows directly from Eqns. \eqref{def: time-ordering} and \eqref{def:Wightman kernels}.
        
        \item For all $\phi_1,\cdots,\phi_n \in \mathcal{O}_\S$, we have
        \begin{equation}\begin{aligned}
            \Delta_n^{(\Omega,\R)}&[\omega_1 \cdot (a,\Lambda),\cdots, \omega_n \cdot (a,\Lambda); x_1,\cdots,x_n](\phi_1,\cdots,\phi_n) \nonumber \\ &\stackrel{\eqref{eqn:time-ordered in terms of Wightman}}{=} \sum_{\sigma \in S_n} \left(\prod_{i=1}^{n-1} \Theta(\tau_{\sigma(i)} - \tau_{\sigma(i+1)})\right) \nonumber \\ & \qquad \qquad \qquad  W^{(\Omega,\R)}_n[\omega_{\sigma(1)} \cdot (a,\Lambda),\cdots,\omega_{\sigma(n)} \cdot (a,\Lambda);x_{\sigma(1)},\cdots,x_{\sigma(n)}](\phi_{\sigma(1)},\cdots,\phi_{\sigma(n)}) \\
            &\stackrel{\eqref{eqn:Poincare covariance Wightman}}{=} \sum_{\sigma \in S_n} \left(\prod_{i=1}^{n-1} \Theta(\tau_{\sigma(i)} - \tau_{\sigma(i+1)})\right)  \nonumber \\ & \qquad \qquad \qquad W^{(\Omega,\R)}_n[\omega_{\sigma(1)},\cdots,\omega_{\sigma(n)};\Lambda x_{\sigma(1)} + a,\cdots,\Lambda x_{\sigma(n)} + a](\phi_{\sigma(1)},\cdots,\phi_{\sigma(n)}) \\
            &= \Delta_n^{(\Omega,\R)}[\omega_1,\cdots, \omega_n; \Lambda x_1 + a,\cdots, \Lambda x_n+a](\phi_1,\cdots,\phi_n)
        \end{aligned}\end{equation}
        since $\Theta(\tau_i-\tau_j) \stackrel{(a,\Lambda)}{\mapsto} \Theta(\tau_i - \tau_j)$, which holds for all $\phi_1,\cdots,\phi_n$.

        \item We have
        \begin{align*}
\Delta_2^{(\Omega,\R)}&[\omega_1,\omega_2;x_1,x_2](\phi_1,\phi_2) \nonumber \\ &= \Tr\left[\Omega \, \T^\R_2[\omega_1,\omega_2;x_1,x_2](\phi_1,\phi_2)\right] \\
&= \Tr\left[\Omega \left( \Theta(\tau_{x_1} - \tau_{x_2}) (\hat{\phi}_1)^\R_{\omega_1}(x_1)(\hat{\phi}_2)_{\omega_2}^\R(x_2) + \Theta(\tau_{x_2} - \tau_{x_1}) (\hat{\phi}_2)^\R_{\omega_2}(x_2)(\hat{\phi}_1)^\R_{\omega_1}(x_1) \right)\right] \\
&= \Theta(\tau_{x_1} - \tau_{x_2}) \Tr\left[\Omega (\hat{\phi}_1)_{\omega_1}^\R(x_1)(\hat{\phi}_2)^\R_{\omega_2}(x_2)\right] + \Theta(\tau_{x_2} - \tau_{x_1}) \Tr\left[\Omega (\hat{\phi}_2)^\R_{\omega_2}(x_2)(\hat{\phi}_1)^\R_{\omega_1}(x_1)\right] \\
&= \Theta(\xi^0) W_2^{(\Omega,\R)}[\omega_1,\omega_2;x_1,x_2](\phi_1,\phi_2) + \Theta(-\xi^0) W_2^{(\Omega,\R)}[\omega_2,\omega_1;x_2,x_1](\phi_2,\phi_1) \, .
\end{align*}
    \end{enumerate}
\end{proof}

\subsection{Proof of Lem. \ref{lem: yenRchi properties}}
    \label{proof lem yenRchi properties}
    \begin{proof}
        \begin{enumerate}
        \item Let $a,b \in \mathbb{C}$ and $\chi_1,\chi_2 \in \thr$. First, note that $\forall X \in \Bor(F)$, $\mu^{\E_\R}_{a\chi_1 + b \chi_2}(X) = \Tr[(a\chi_1 + b \chi_2)\E_\R(X)] = a \Tr[\chi_1 \E_\R(X)] + b \Tr[\chi_2 \E_\R(X)] = a\mu^{\E_\R}_{\chi_1}(X) + b \mu^{\E_\R}_{\chi_2}(X)$. Thus, $\forall \phi \in \bhs$,
        \begin{multline}
            \yen^\R_{a\chi_1 + b \chi_2}(\phi) = \int_{\Poincup} (a,\Lambda) \cdot \phi \, d\mu^{\E_\R}_{a \chi_1 + b \chi_2}(g) = a\int_{\Poincup} (a,\Lambda) \cdot \phi \, d\mu^{\E_\R}_{\chi_1}(g) + b \int_{\Poincup} (a,\Lambda) \cdot \phi \, d\mu^{\E_\R}_{\chi_2}(g) \\ = a \yen^\R_{\chi_1}(\phi) + b \yen^{\R}_{\chi_2}(\phi) \, . 
        \end{multline}
        \qed
        \item Let $a,b \in \mathbb{C}$ and $\phi_1,\phi_2 \in \bhs$. The case where $a=0$ is trivial ($b$ (a scalar in $\mathbb{C}$) commutes with $U_\S(a,\Lambda)$ for every $(a,\Lambda) \in \Poincup$), so suppose $a \neq 0$. Then
        \begin{multline}
            \yen^\R_{\chi}(a \phi_1 + b \phi_2) = \int_G U_\S(a,\Lambda) (a \phi_1 + b \phi_2) U_\S(a,\Lambda)^\dagger \, d\mu^{\E_\R}_{\chi}(g) = a\int_G U_\S(a,\Lambda) (\phi_1 + \frac{b}{a} \phi_2) U_\S(a,\Lambda)^\dagger \, d\mu^{\E_\R}_{\chi}(g) \\ = a \left(\int_G U_\S(a,\Lambda) \phi_1 U_\S(a,\Lambda)^\dagger d\mu^{\E_\R}_{\chi}(g) + \frac{b}{a} \int_G U_\S(a,\Lambda) \phi_2 U_\S(a,\Lambda)^\dagger \, d\mu^{\E_\R}_{\chi}(g) \right) \\= a \int_G (a,\Lambda) \cdot \phi_1 \, d\mu^{\E_\R}_{\chi}(g) + b \int_G (a,\Lambda) \cdot \phi_2 \, d\mu^{\E_\R}_{\chi}(g)  = a \yen^\R_{\chi}(\phi_1) + b \yen^\R_\chi(\phi_2) \, .
        \end{multline}
        \qed
        \item Since $\yen^\R$ is unital, $\yen^\R(\mathbb{1}_{\bhs}) = \mathbb{1}_{\bhs} \otimes \mathbb{1}_{\bhr}$, so
        \begin{equation}
            \yen^\R_{\chi}(\mathbb{1}_{\bhs}) = \Gamma_{\chi}(\mathbb{1}_{\bhs} \otimes \mathbb{1}_{\bhr}) = \mathbb{1}_{\bhs} \omega_\chi(\mathbb{1}_{\bhr}) = \Tr[\chi] \mathbb{1}_{\bhs} \, ,
        \end{equation}
        so $\yen^\R_{\chi}$ is unital exactly when $\Tr[\chi] = 1$. \qed
        \item First, note that $\forall X \in \Bor(F)$, $\mu^{\E_\R}_{\chi^\dagger}(X) = \Tr[\chi^\dagger \E_\R(X)] = \overline{\Tr[\chi \E_\R(X)]} = \overline{\mu^{\E_\R}_{\chi}(X)}$. Moreover, $\forall \phi \in \bhs$ and $\forall (a,\Lambda) \in \Poincup$, $(U_\S(a,\Lambda) \phi U_\S(a,\Lambda)^\dagger)^\dagger = U_\S(a,\Lambda) \phi^\dagger U_\S(a,\Lambda)^\dagger$. Thus,
        \begin{equation}
            \yen^\R_\chi(\phi)^\dagger = \left(\int_G (a,\Lambda) \cdot \phi \, d\mu^{\E_\R}_{\chi}(g) \right)^\dagger = \int_G ((a,\Lambda) \cdot \phi)^\dagger \, d\overline{\mu^{\E_\R}_{\chi}(g)} = \int_G (a,\Lambda) \cdot (\phi^\dagger) \, d\mu^{\E_\R}_{\chi^\dagger}(g) = \yen^\R_{\chi^\dagger}(\phi^\dagger) \, . 
        \end{equation}
        In particular, when $\chi^\dagger = \chi$, $\yen^\R_{\chi}(\phi)^\dagger = \yen^\R_{\chi}(\phi^\dagger)$. \qed
        \item Note that $\yen^\R$ is completely positive \cite{loveridge_symmetry_2018}, and if $\chi \geq 0$ it follows that $\Gamma_\chi = \mathbb{1}_{\bhs} \otimes \omega_\chi$ is also completely positive. Thus, since the composition of two completely positive maps is positive, the result follows. \qed
        \item Let $\phi \in \ths$. Note that $\mu^{\E_\R}_{\chi}(F) = \Tr[\chi \E_\R(F)] = \Tr[\chi]$. Then 
        \begin{multline}
            \Tr[\yen^\R_{\chi}(\phi)] = \Tr\left[\int_G U_\S(a,\Lambda) \phi U_\S(a,\Lambda) ^\dagger \, d\mu^{\E_\R}_{\chi}(g) \right] = \int_G \Tr[U_\S(a,\Lambda) \phi U_\S(a,\Lambda)^\dagger] \, d\mu^{\E_\R}_{\chi}(g) \\ = \int_G \Tr[\phi] \, d\mu^{\E_\R}_\chi(g) = \Tr[\phi] \int_G \, d\mu^{\E_\R}_{\chi}(g) = \Tr[\phi] \cdot \Tr[\chi] \, .
        \end{multline}
        Thus, $\yen^\R_{\chi}$ is trace-preserving iff $\Tr[\chi] = 1$. \qed 
        \item A map $\bhs \to \bhs$ is normal iff it is ultraweakly continuous. Note that $\yen^\R$ is normal, and $\Gamma_\chi = \mathbb{1} \otimes \omega_\chi$ where $\omega_\chi$ is a normal functional on $\bhr$ for any trace-class $\chi \in \thr$, hence so is $\Gamma_\chi$. The composition of normal maps is normal, so $\yen^\R_\chi$ is normal for every $\chi \in \thr$.
        \item $\yen^\R_{\chi}$ is effect-preserving iff $\mathbb{0} \leq A \leq \mathbb{1} \Rightarrow \mathbb{0} \leq \yen^\R_\chi(A) \leq \mathbb{1}$. For a positive map $\Phi : \bhs \to \bhs$, one always has $\mathbb{0} \leq A \leq \mathbb{1} \Rightarrow \mathbb{0} \leq \Phi(A) \leq \Phi(\mathbb{1})$. Thus, since $\yen^\R_\chi$ is completely positive (hence positive) if $\chi \geq 0$ then $\yen^\R_\chi(\phi) \geq 0$. Moreover, by complete positivity, $\norm{\yen^\R_\chi} = \norm{\yen^\R_\chi(\mathbb{1})} = \Tr[\chi] \geq 0$, so for $\Tr[\chi] \leq1$ and $\phi \leq \mathbb{1}$, $0 \leq \yen^\R_\chi(\phi) \leq \yen^\R_\chi(\mathbb{1}) \leq 1$. \qed
        \item First, note that $\norm{\omega_\chi} = \norm{\chi}_1$ (trace norm), so $\Gamma_\chi$ has norm $\norm{\Gamma_\chi} = \norm{\chi}_1$. Moreover, $\yen^\R$ is a contraction \cite{loveridge_symmetry_2018} i.e. $\norm{\yen^\R} \leq 1$. Thus,
        \begin{equation}
            \norm{\yen^\R_\chi} \leq \norm{\Gamma_\chi} \norm{\yen^\R} \leq \norm{\chi}_1
        \end{equation}
        so $\yen^\R_\chi$ is bounded (and linear thus continuous). If $\chi \geq 0$ then $\norm{\yen^\R_\chi} = \norm{\yen^\R_\chi(\mathbb{1})} = \Tr[\chi] = \norm{\chi}_1$ since for positive trace-class $\chi$, $\norm{\chi}_1 = \Tr[\chi]$. \qed
        \item Since $\norm{\yen^\R_\chi} \leq \norm{\chi}_1$, it follows that if $\norm{\chi}_1 \leq 1$, $\norm{\yen^\R_\chi} \leq 1$. Moreover, if $\chi \geq 0$, then $\norm{\yen^\R_\chi} = \norm{\yen^\R_\chi(\mathbb{1})} = \Tr[\chi]$ i.e. $\yen^\R_\chi$ is a contraction iff $\Tr[\chi] \leq 1$.
    \end{enumerate}
    \end{proof}

\subsection{Proof of Thm. \ref{thm: R-ireducibility}}\label{proof 6.4}
    \begin{proof}
    Let $\tilde{\omega}_i \in \operationalstate$ and $\omega_i := a \cdot \tilde{\omega}_i \in \operationalstate$ (since $\operationalstate$ is closed under translations) where $i=1,\cdots,n \in \mathbb{N}$ for any $a \in T(1,d-1)$. Suppose $A \in \mathcal{O}_\S$ is such that $\comm{A}{\hat{\Phi}^\R(\tilde{\omega}_i)} = 0$ for all $\tilde{\omega}_i \in \operationalstate$. Then for all $\rho \in \systemstate$,
    \begin{equation}
        \Tr[\rho A \hat{\Phi}^\R(\tilde{\omega}_1) \cdots \hat{\Phi}^\R(\tilde{\omega}_n)] = \Tr[\rho \hat{\Phi}^\R(\tilde{\omega}_1) \cdots \hat{\Phi}^\R(\tilde{\omega}_n) A] \, .
    \end{equation}
    In particular, by the existence of a pure vacuum, we have
    \begin{equation}\begin{aligned}
    \expval{A \hat{\Phi}^\R(a \cdot \omega_1) \cdots \hat{\Phi}^\R(a \cdot \omega_n)}{\Omega} &= \expval{\hat{\Phi}^\R(a \cdot \omega_1) \cdots \hat{\Phi}^\R(a \cdot \omega_n) A}{\Omega} \\
    \stackrel{\ref{thm:full covariance}}{\Leftrightarrow} \expval{A a \cdot \hat{\Phi}^\R(\omega_1) \cdots a \cdot \hat{\Phi}^\R(\omega_n)}{\Omega} &= \expval{a \cdot \hat{\Phi}^\R(\omega_1) \cdots a \cdot \hat{\Phi}^\R(\omega_n) A}{\Omega} \\
    \Leftrightarrow \expval{A U_\S(a,e) \hat{\Phi}^\R(\omega_1) \cdots \hat{\Phi}^\R(\omega_n)}{\Omega} &= \expval{\hat{\Phi}^\R(\omega_1) \cdots \hat{\Phi}^\R(\omega_n) U_\S(-a,e) A}{\Omega} \\
    \Leftrightarrow \expval{A \E(X) \hat{\Phi}^\R(\omega_1) \cdots \hat{\Phi}^\R(\omega_n)}{\Omega} &= \expval{\hat{\Phi}^\R(\omega_1) \cdots \hat{\Phi}^\R(\omega_n) \E(-X) A}{\Omega}
    \end{aligned}\end{equation}
    where the last equation is obtained by Fourier transforming in $a$, where $\E$ is the spectral projection of the energy-momentum generators and $X \in \Bor(\mathbb{R}^d)$. But since the spectrum of energy-momentum lies in or on the plus cone (by the spectrum condition), if $-X$ lies in the physical spectrum and does not include $p=0$, the left-hand side vanishes, and \textit{vice-versa}. This means that $A \ket{\Omega}$ is orthogonal to all states $\E(-X) \hat{\Phi}^\R(\omega_n)^\dagger \cdots \hat{\Phi}^\R(\omega_1)^\dagger \ket{\Omega}$, which implies that $A \ket{\Omega} = c \ket{\Omega}$ for some $c \in \mathbb{C}$. Thus, for all $\eta \in \his$,
    \begin{equation}
        \bra{\eta}A \hat{\Phi}^\R(\omega_1) \cdots \hat{\Phi}^\R(\omega_n) \ket{\Omega} = \bra{\eta} \hat{\Phi}^\R(\omega_1) \cdots \hat{\Phi}^\R(\omega_n) A \ket{\Omega} = c \bra{\eta} \hat{\Phi}^\R(\omega_1) \cdots \hat{\Phi}^\R(\omega_n) \ket{\Omega}
    \end{equation}
    which, by $(\mathcal{O}_\S,\operationalstate)$-cyclicity, implies that $\forall \eta,\chi \in \his$,
    \begin{equation}
        \bra{\eta}A \ket{\chi} = c \braket{\eta}{\chi}
    \end{equation}
    so $A = c \mathbb{1}_{\mathcal{O}_\S}$, which concludes the proof.
\end{proof}    
\end{appendices}

\end{document}